\documentclass[8pt, a4paper]{article}
\usepackage[verbose=true,letterpaper]{geometry}
\AtBeginDocument{
  \newgeometry{
    textheight=9in,
    textwidth=5.5in,
    top=1in,
    headheight=12pt,
    headsep=25pt,
    footskip=30pt
  }
}
\usepackage{arxiv_style}
\usepackage[utf8]{inputenc} 
\usepackage[T1]{fontenc}    
\usepackage{hyperref}       
\usepackage{url}            
\usepackage{booktabs}       
\usepackage{amsfonts}       
\usepackage{nicefrac}       
\usepackage{microtype}      
\usepackage{xcolor}        
\usepackage{parskip}
\usepackage{bm}
\usepackage{enumitem}
\usepackage{graphicx}
\usepackage[small]{caption}
\usepackage{subcaption}
\usepackage{amsmath, amssymb}
\usepackage{amsthm}
\usepackage{xfrac}
\usepackage{multirow}
\usepackage{wrapfig}
\usepackage{nicematrix}
\usepackage{makecell}
\usepackage{rotating}
\usepackage{tikz}
\usepackage{doi}
\usetikzlibrary{arrows.meta, positioning, shapes.geometric, fit, backgrounds, calc}

\usepackage{algorithm}
\usepackage[noEnd=true]{algpseudocodex}

\makeatletter
\providecommand{\theHALG@line}{}
\renewcommand{\theHALG@line}{\thealgorithm.\arabic{ALG@line}}
\makeatother

\numberwithin{equation}{section}
\usepackage{cleveref}
\creflabelformat{equation}{#2\textup{#1}#3}

\hypersetup{colorlinks=true, 
            citecolor=blue, 
            linkcolor=blue, 
            urlcolor=blue}

\newtheorem{definition}{Definition}[section]
\newtheorem{remark}{Remark}[section]
\newtheorem{assumption}{Assumption}[section]
\newtheorem{proposition}{Proposition}[section]

\usepackage[numbers, compress]{natbib}
\newcommand{\termDS}{\substack{v_H q_{inf}^D \\ + \beta_{DD}\mu(DI) \\ + \beta_{UD}\mu(UI)}}
\newcommand{\termUS}{\substack{v_H q_{inf}^U \\ + \beta_{UU}\mu(UI) \\ + \beta_{DU}\mu(DI)}}

\newcommand{\params}[0]{\gamma}

\newcommand{\cA}{\mathcal{A}}
\newcommand{\cM}{\mathcal{M}}
\newcommand{\cP}{\mathcal{P}}

\newcommand{\PP}{\mathbb{P}}
\newcommand{\EE}{\mathbb{E}}
\newcommand{\RR}{\mathbb{R}}
\newcommand{\NN}{\mathbb{N}}
\newcommand{\MM}{\mathbb{M}}

\title{\textbf{Neural Parameter Calibration for Finite-State\\Mean Field Games}}

\author{Anna C.M.~Th\"oni\thanks{Radboud University, Nijmegen, the Netherlands}
\and
  Grégoire Lambrecht\thanks{New York University, New York, NY}
  \and
  G\"okçe Dayanıklı\thanks{University of Illinois Urbana-Champaign, Champaign, IL}
  \and
  Yonathan Efroni\thanks{Tel Aviv University, Tel Aviv, Israel}
  \and
  Tal Kachman\footnotemark[1]
  \and
  Mathieu Laurière\thanks{New York University Shanghai, Shanghai, People's Republic of China}
  }
\date{}

\begin{document}

\maketitle 

\begin{abstract}
	Mean field games efficiently approximate a very large population of strategic agents. While these games can aid the understanding of complex systems, their deployment in real-world settings is challenged by the specification of their parameters: mean field games (MFGs) often involve hidden preferences, constraints, and interactions that can rarely be theoretically derived or directly observed. To address this gap, we present a neural network-based framework for learning parametric, finite-state MFGs from observed population dynamics. To do so, we formulate the parameter calibration as an inverse problem and use implicit differentiation to backpropagate through the games' equilibrium. The resulting approach is fully differentiable and enables us to estimate flexible trajectory-wise parameter paths, including state- and time-dependent specifications without requiring observations of the individual agents' actions or rewards. We provide a proof for the exactness of the gradient computation in a discrete-time formulation. We validate our framework through numerical experiments across four systems of increasing complexity, ranging from synthetic linear-quadratic benchmarks to real-world urban mobility datasets.
\end{abstract}

\section{Introduction}
Mean field games (MFGs) offer a powerful mathematical framework for analyzing the strategic behavior of massive populations of interacting agents. By reducing the intractable complexity of a large number of individual interactions to a tractable problem involving a single representative agent and a macroscopic population distribution, MFGs have enabled important progress in modelling large-scale systems. Originally introduced independently by Lasry and Lions \cite{LasryLionsMFG2007} and Huang, Malhamé, and Caines \cite{HuangMFG2006}, this paradigm has been successfully applied across diverse domains, including urban traffic routing \cite{Chevalier2015,Festa2018,Huang2019,Tanaka2021}, epidemic spread \cite{Hubert2018,Doncel2022,Lee2021,aurell2022finite,Buckley2025,liu2025,liu2026modelingepidemicspreadstrategic}, financial markets \cite{Cardaliaguet2018,Carmona2020,Carmona2023}, and cybersecurity \cite{Kolokoltsov2016,Qian2020,Shen2025}. While a broad part of the literature traditionally considers MFGs in continuous state-action spaces \cite{Bensoussan2013, Carmona2018}, there has been growing interest in discrete-time, discrete-space settings \cite{anahtarci2021learningfittedaverage, yongacoglu2022independent, yardim2023policy} and finite-state, continuous-time MFGs, which are highly suited for discrete decision-making processes~\cite{Gomes2013}.

Despite their theoretical elegance and advances, the deployment of MFGs in real-world scenarios remains bottlenecked by the challenge of parameter specification. Constructing an accurate MFG requires defining complex structural parameters, such as hidden preference costs, transition constraints, and interaction rates. In real-world environments, these parameters are abstract and rarely directly observable, barring the data-driven adoption of MFGs. While purely descriptive computational models or finite-agent simulations can often be calibrated from historical data \cite{Wang2021, Hirano2022, Gaskin2023, Giampiccolo2024}, they inherently lack the causal, strategic, and forward-looking structure that MFGs provide. Consequently, a rigorous methodology to calibrate the structural MFG parameters from observed population data forms a valuable addition to the current modelling landscape.

\textbf{Contributions.} We introduce a neural network-based framework for learning parametric MFGs entirely from macroscopic trajectory data. An overview of the method is presented in Fig.~\ref{fig:method_overview}. We formulate the parameter calibration as an inverse problem: by treating the fixed-point equilibrium solver as an implicit neural network component, our method uses implicit differentiation to backpropagate through the MFG equilibrium. This allows us to learn highly flexible, time-dependent representations of the game's underlying parameters without observing individual agent actions or rewards. Our main contributions are summarized as follows:

\begin{itemize}[topsep=0pt]
    \item We propose a novel, fully differentiable approach that estimates trajectory-wise parameter paths obtained by evaluating a \textbf{time- and mean-field-dependent neural map} on \textbf{observed population-density trajectories}, using implicit differentiation through the equilibrium fixed point.
    \item We provide \textbf{theoretical foundations} supporting the method under suitable conditions, including the exactness of the adjoint gradient computation, the consistency of a computational decoupling strategy, and a uniform generalization bound.
    \item We validate the framework on \textbf{four systems of increasing complexity}, demonstrating that it \textbf{(1)}~accurately calibrates MFG parameters from both synthetic and real data, \textbf{(2)}~outperforms a mean-field dynamics baseline lacking game-theoretic structure, and \textbf{(3)}~maintains robustness to noise and dimensionality scaling in synthetic benchmarks.
\end{itemize}

\textbf{Related Work.} Our approach intersects with several active areas of research. Inverse Reinforcement Learning (IRL) aims to recover reward functions from agent behavior \cite{chen2022individualinverseRL, chen2021adversarial, anahtarci2024maximum, anahtarci2025maximum, alkir2025inverse, anahtarci2025kernel}. It is important to clarify that, unlike standard IRL, our method is not model-free: we assume the model is known up to some parameterization. However, because these parameters can depend on both time and the mean field distribution, learning them remains a challenging task. In the context of MFGs, recent advances in deep learning have primarily focused on solving the forward problem---computing equilibria for known parameters using neural networks or operator learning \cite{Cohen2024, Hofgard2026}---or exploring neural ordinary differential equations or model-based reinforcement learning for forward MFGs \cite{Thoni2026, huang2024model}. While classical parameter estimation has been studied for standard partial differential equations and multi-agent simulations \cite{Gaskin2023}, the calibration of MFGs from macroscopic density data remains largely unexplored. We contribute to filling this gap by considering the continuous-time, finite-state case.

\textbf{Paper Outline.} The remainder of this paper is structured as follows. Sec.~\ref{sec: background} introduces the background of finite-state MFGs. Sec.~\ref{sec: method} details our parameter calibration methodology and its theoretical foundations. Sec.~\ref{sec: examples} presents the test cases we study. Sec.~\ref{sec: experiments} provides the experimental evaluation across synthetic and real-world datasets. We conclude in Sec.~\ref{sec: conclusion}.

\textbf{Code Availability.} All code for the data generation and implementation of the methods can be found at \url{https://github.com/KachmanLab/MFGParameterCalibration}. Experimental details are also provided in the appendices.

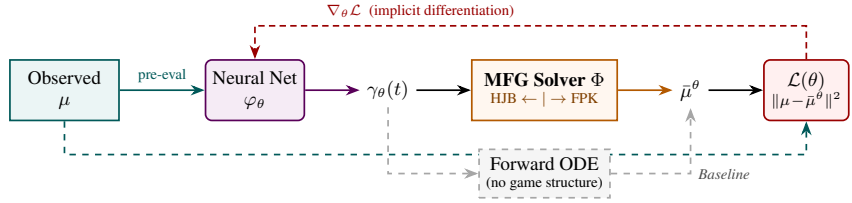
\begin{figure}[!htbp]
\centering
\resizebox{0.8\textwidth}{!}{%
\begin{tikzpicture}[
    >=Stealth,
    node distance=1.0cm and 1.4cm,
    databox/.style={rectangle, draw=teal!70!black, fill=teal!8, thick,
                    minimum height=1.0cm, minimum width=1.8cm, align=center,
                    font=\small},
    nnbox/.style={rectangle, draw=violet!70!black, fill=violet!8, thick,
                  rounded corners=3pt, minimum height=1.0cm, minimum width=1.6cm,
                  align=center, font=\small},
    solverbox/.style={rectangle, draw=orange!70!black, fill=orange!6, thick,
                      minimum height=1.0cm, minimum width=2.4cm, align=center,
                      font=\small},
    lossbox/.style={rectangle, draw=red!60!black, fill=red!6, thick,
                    rounded corners=3pt, minimum height=1.0cm, minimum width=1.4cm,
                    align=center, font=\small},
    baselinebox/.style={rectangle, draw=gray!70, fill=gray!8,
                        thick, dashed, minimum height=0.8cm, minimum width=1.7cm,
                        align=center, font=\small},
    arrowlabel/.style={font=\footnotesize, midway},
]

\node[databox] (obs) {Observed\\$\mu$};
\node[nnbox, right=of obs] (nn) {Neural Net\\$\varphi_\theta$};
\node[right=0.9cm of nn, font=\small] (gamma) {$\gamma_\theta(t)$};
\node[solverbox, right=0.9cm of gamma] (solver) {\textbf{MFG Solver} $\Phi$\\[-2pt]{\scriptsize\color{orange!50!black} HJB $\leftarrow$ $|$ $\rightarrow$ FPK}};
\node[right=0.9cm of solver, font=\small] (mubar) {$\bar\mu^\theta$};
\node[lossbox, right=0.9cm of mubar] (loss) {$\mathcal{L}(\theta)$\\[-1pt]{\scriptsize$\|\mu\!-\!\bar\mu^\theta\|^2$}};

\draw[->, thick, teal!70!black] (obs) -- node[arrowlabel, above] {\scriptsize pre-eval} (nn);
\draw[->, thick, violet!70!black] (nn) -- (gamma);
\draw[->, thick] (gamma) -- (solver);
\draw[->, thick, orange!70!black] (solver) -- (mubar);
\draw[->, thick] (mubar) -- (loss);

\draw[->, thick, teal!70!black, dashed] (obs.south) -- ++(0,-0.55) -| node[arrowlabel, below, pos=0.25] {} (loss.south);

\draw[->, thick, red!60!black, densely dashed]
  (loss.north) -- ++(0,0.55) -| node[arrowlabel, above, pos=0.35, font=\scriptsize, text=red!60!black] {$\nabla_\theta\mathcal{L}$ \;(implicit differentiation)} (nn.north);

\node[baselinebox, below=0.45cm of solver] (baseline) {Forward ODE\\[-2pt]{\scriptsize(no game structure)}};
\draw[->, gray!70, thick, dashed] (gamma.south) |- (baseline.west);
\draw[->, gray!70, thick, dashed] (baseline.east) -| node[arrowlabel, right, pos=0.5, font=\scriptsize, text=gray!60!black] {\textit{Baseline}} (mubar.south);

\end{tikzpicture}
}%
\caption{\textbf{Method overview.} The network $\varphi_\theta$ is pre-evaluated on observed trajectories to produce $\gamma_\theta(t)$, which are passed to the MFG solver~$\Phi$ (coupled HJB and FPK equations). Gradients flow back via implicit differentiation through the fixed point (dashed red path). The dashed gray path shows the forward-only baseline that omits the solver process as it does not incorporate any game structure.}
\label{fig:method_overview}
\end{figure}

\section{Background}
\label{sec: background}

\paragraph{General Notation.} Throughout this work, we use the following notation. Let $\NN = \{0,1,\dots\}$ denote the set of natural numbers and $\NN^* = \NN \setminus \{0\}$. For any $k \in \NN$, we define $[k] = \{1,\dots,k\}$. Let $\RR$ denote the set of real numbers. 
For any $d \in \NN^*$, $\RR^d$ denotes the set of real-valued vectors of dimension $d$, and $\RR^{d\times d}$ the set of real-valued square matrices of size $d$. Let $\MM^d = \{ A \in \RR^{d\times d} \, \mid A[x, y] \ge 0 \text{ for any distinct } x,y, \text{ and for any }x\in [d], \; A[x,x]=-\sum_{y\in[d],\,y\ne x} A[x,y] \}$ denote the set of transition rate matrices (infinitesimal generators).

\paragraph{Finite-state Mean Field Games.} 
We consider MFGs with a finite-state space as introduced by \citet{Gomes2013}. Specifically, we let $[d] := \{1,\dots,d\}$ denote the state space. We consider a continuous-time horizon $t \in [0, T]$ with $T > 0$. Let $\cA$ be a Polish action space and $\cP([d])$ the set of probability measures on $[d]$ (the $(d-1)$-simplex). The choice of $\cA$ is model-dependent. A control is a function $\alpha: [0, T] \times [d] \mapsto \cA$. We restrict our attention to admissible controls, i.e., Borel-measurable feedback maps $\alpha:[0,T]\times[d]\to\cA$ such that, for every population flow $\mu$, the induced transition rates define a valid generator in $\MM^d$ and the cost $J^{\gamma,\mu_0}(\alpha;\mu)$ is finite.
We denote by $\gamma: [0, T] \times \cP([d]) \to \RR^n$ the dynamic parameters of the game, belonging to a space of continuous functions $\Gamma$. In practice, we may consider more restrictive parameter classes, such as constant parameters $\gamma \in \RR^n$ or time-only functions $\gamma: [0,T] \to \RR^n$, which are special cases of this general formulation. The game is defined by: a transition rate matrix $Q : \RR^n \times \cA \times \cP([d]) \to \MM^d$, a running cost $f : \RR^n \times [d] \times \cA \times \cP([d]) \to \RR$, and a terminal cost $g : \RR^n \times [d] \times \cP([d]) \to \RR$.
For a specific parameter function $\gamma \in \Gamma$ and values $(t,x,a,\mu) \in [0,T]\times [d]\times \cA \times \cP([d])$, we define the evaluated parametric functions as:
\begin{align}
    \begin{split}
        Q^{\gamma}(t, a,\mu) &:= Q(\gamma(t, \mu), a,\mu), \\
	   f^{\gamma}(t, x,a,\mu) &:= f(\gamma(t, \mu),x,a,\mu),\\
       g^{\gamma}(x,\mu) &:= g(\gamma(T, \mu),x,\mu).
	\label{eq: parametric evaluated functions}
    \end{split}
\end{align}
We denote by $ \cM^{\gamma} := (d,\cA,T,Q^{\gamma},f^{\gamma},g^{\gamma})$ a parametric MFG, and by $(\cM^{\gamma})_{\gamma\in\Gamma}$ the family.
\paragraph{Controls, State Transitions and Markovian Dynamics.}
Analogous to general MFG theory \cite{LasryLionsMFG2007, HuangMFG2006}, finite-state MFGs consider a representative player interacting with a given distribution of agents $\mu: [0, T] \to \cP([d])$. The representative agent influences their transition between states. This process is governed by the transition matrix $Q^{\gamma}(t, \alpha_t, \mu_t) \in \MM^d$, which depends on the parameter function $\gamma$, the player's control $\alpha$, and the distribution over states $\mu_t$. For each $x\ne y\in [d]$, each matrix entry $Q^{\gamma}(t, \alpha_t, \mu_t)[x, y]$ describes the transition rate from state $x$ to $y$ at time $t$. For the parameter function $\gamma$ and given an initial distribution $\mu_0 \in \cP([d])$, the representative agent controls their state $X^{\mu_{0}, \alpha,\gamma}_{t}$ following:
\begin{equation}
	\PP\big(X^{\gamma,\mu_0,\alpha}_{t+h} = y \mid X^{\gamma,\mu_0,\alpha}_{t} = x\big)
	= h\,Q^{\gamma}(t, \alpha_t,\mu_t)[x,y] + o(h), \quad h \to 0^+,
	\label{eq: transition probability}
\end{equation}
where $X_{0}^{\gamma, \mu_0, \alpha} \sim \mu_0$.
The state process $(X^{\gamma,\mu_0,\alpha}_t)_{t\in[0,T]}$ is a continuous-time Markov chain with initial law $\mu_0$, transition rates $Q^\gamma$ and control $\alpha$. It implicitly depends on $\mu = (\mu_t)_{t \in [0,T]}$.

\paragraph{The Expected Cost and Mean Field Equilibrium.} Throughout the game, the representative player aims to minimize their expected cost, a nonlinear function that depends on both the player's control as well as the population distribution $\mu_t \in \mathcal{P}([d])$ \cite{Kolokoltsov2010}. The total expected cost $J$ is the sum of a running cost $f$ and a terminal cost $g$.  If the representative player uses control $\alpha$, its expected cost is
\begin{equation}
	J^{\gamma,\mu_0}(\alpha; \mu) = \mathbb{E}\Big[\int_0^Tf^{\gamma}\left(s, X_{s}^{\gamma, \mu_0, \alpha}, \alpha_{s}(X_{s}^{\gamma, \mu_0, \alpha}), \mu_s \right)ds + g^{\gamma}(X_{T}^{\gamma, \mu_0, \alpha}, \mu_T)\Big].
	\label{eq: cost}
\end{equation}
Because the representative player cannot directly influence the population distribution $\mu$, the player solves a standard optimal control problem, parameterized by $\mu$ (and $\gamma$). 

At the Nash equilibrium, the distribution of the population and the representative player are the same.
\begin{definition}
	For $\gamma \in {\Gamma}$, a Nash equilibrium of the parametric MFG $\cM^{\gamma}$, associated with an initial distribution $\mu_0 \in \mathcal{P}([d])$, is a pair $(\bar{\alpha}^{\gamma,\mu_0}, \bar{\mu}^{\gamma,\mu_0})$ such that: \textbf{(1)} $\bar{\alpha}^{\gamma,\mu_0}$ minimizes the cost functional $J^{\gamma,\mu_0}(\cdot\,; \bar\mu^{\gamma,\mu_0})$ and \textbf{(2)} for every $t \in [0, T]$ it holds that $\bar\mu^{\gamma,\mu_0}_t = \text{Law}(X_{t}^{\gamma,\mu_0, \bar\alpha})$.
	\label{def: MFG equilibrium}
\end{definition}

The second requirement, known as the consistency condition \cite{Caines2017}, is one of the features distinguishing the MFG equilibrium from a mean field control problem, in which a single decision maker controls the mean field.

\paragraph{Characterizing the Mean Field Equilibrium: the Forward-Backward System.} To characterize the equilibrium introduced in Definition \ref{def: MFG equilibrium}, the MFG $\cM^{\gamma}$ is described by a coupled system consisting of a Hamilton--Jacobi--Bellman equation (HJB) and a Fokker--Planck--Kolmogorov equation (FPK). The HJB equation encodes the optimization problem of the representative player, runs backward in time, and describes the value function $\bar{u}^{\gamma,\mu_0}$. The FPK equation encodes the consistency condition, runs forward in time, and governs the evolution of the population distribution $\mu^{\gamma,\mu_0}$. More precisely, the following system of equations characterizes the equilibrium (see~\cite[Section 7.2]{Carmona2018} for details):
\begin{equation}
	\begin{cases}
		-\partial_t \bar{u}^{\gamma,\mu_0}_t(x) = H^{\gamma}\big(t,x,\bar{\mu}^{\gamma,\mu_0}_t, \bar{u}^{\gamma,\mu_0}_t, \bar\alpha^{\gamma,\mu_0}_t(x)\big)                            & (i)   \\
		\partial_t \bar{\mu}^{\gamma,\mu_0}_t(x) = \sum_{y \in [d]} \bar{\mu}^{\gamma,\mu_0}_t(y)\, Q^{\gamma}\big(t, \bar\alpha^{\gamma,\mu_0}_t(y), \bar{\mu}^{\gamma,\mu_0}_t\big)[y,x] & (ii)  \\
		\bar{\mu}^{\gamma,\mu_0}_{|t=0} = \mu_0, \quad \bar{u}^{\gamma,\mu_0}_{T}(x) = g^{\gamma}(x,\bar{\mu}^{\gamma,\mu_0}_T), \quad x \in [d],                                   & (iii)
	\end{cases}
	\label{eq: forward-backward}
\end{equation}
where the equilibrium (feedback) control is $\bar\alpha^{\gamma,\mu_0}_t(x) = \arg\min_{a \in \cA} H^{\gamma}\big(t,x,\bar{\mu}^{\gamma,\mu_0}_t, \bar{u}^{\gamma,\mu_0}_t( \cdot), a\big)$, and the Hamiltonian $H^{\gamma}$ is defined for $p \in \RR^d$ by:
\begin{equation}
	H^{\gamma}(t,x,\mu,p,a) =  f^{\gamma}(t,x,a,\mu) + \sum_{y \in [d]} Q^{\gamma}(t, a,\mu)[x,y]\, p_y.
	\label{eq: hamiltonian}
\end{equation}
Equivalently, since $Q^\gamma(t,a,\mu) \in \MM^d$ has rows summing to zero, this can be written in discrete-gradient form as
\begin{equation*}
H^\gamma(t,x,\mu,p,a)
=
f^\gamma(t,x,a,\mu)
+
\sum_{y\neq x} Q^\gamma(t,a,\mu)[x,y]\,(p_y-p_x).
\end{equation*}
Intuitively, the Hamiltonian combines the instantaneous cost $f^{\gamma}$ with the expected variation of the value function induced by the controlled transitions $Q^{\gamma}$. 
In the sequel, we may omit $\gamma$ or $\mu_0$ from the superscript when the context is clear.

\paragraph{Objectives.} From here, our goal is to propose a method that achieves three main objectives: \textbf{(1) calibrate} the parameterized MFG model $\mathcal{M}_\gamma$ by finding a suitable $\gamma$ so that the mean-field flow matches the observations, \textbf{(2) highlight the forward-looking} quality of the mean field game by outperforming an approach based solely on mean-field dynamics, and \textbf{(3) maintain robustness} against noisy data, (reasonably) out-of-distribution samples, and novel scenarios.

\section{Method}
\label{sec: method}
We describe our algorithm for finding the mean field equilibrium that reflects the learned parameters $\gamma$ in Sec.~\ref{sec: learning parameter estimates}. A theoretical analysis considering the contractivity of our approach, uniqueness of solutions, and associated approximation error is provided in Sec.~\ref{sec: theoretical analysis}.

\subsection{Estimating MFG Parameters}\label{sec: learning parameter estimates}

\paragraph{General Idea.} We estimate the parameters $\gamma_\theta$ to minimize the discrepancy between observed mean field trajectories and the predicted mean field trajectories $\bar\mu^{\gamma_\theta}$, where $\bar\mu^{\gamma_\theta}$ is the mean-field trajectory at equilibrium of the MFG parametrized by $\gamma_\theta$ and starting with the same initial condition as the observation. To this end, we train a neural network $\varphi_\theta$ to approximate the unknown $\gamma$. In general, the network $\varphi_\theta: [0, T] \times \cP([d]) \to \RR^n$ returns time- and distribution-dependent parameter estimates; simpler architectures (constant or time-only) are obtained by restricting the input; see Sec.~\ref{sec: experiments}.

\paragraph{Data.} We observe a dataset $\mathcal{D} = \{(\mu^i_t)_{t\in [0, T]}\}_{1\leq i\leq N_{\mathrm{traj}}}$ of mean field trajectories, generated under an unknown parameter $\gamma$ and possibly from different initial distributions $(\mu^i_0)_{1\leq i\leq N_{\mathrm{traj}}}$. These trajectories may be inexact observations of Nash equilibrium flows due to noise in the data.

\paragraph{Learning the Parameter Estimates.} We approximate $\gamma$ by a neural network $\varphi_\theta$ with parameters $\theta$. We train the parameters to minimize a squared $L_2$ discrepancy between observed mean field trajectories and generated mean field trajectories after solving the MFG corresponding to our current parameter estimate. In practice, we use a minibatch and a random time sub-interval in $[0,T]$ to improve generalization and avoid getting stuck in local loss minima. For an interval $I\subset[0,T]$, we use the continuous-time notation
\[
\|\mu-\mu'\|_{L^2(I)}^2
:=
\int_I \sum_{x\in[d]}(\mu_t(x)-\mu'_t(x))^2\,dt,
\]
which is approximated by a discrete-time sum in the implementation. The average loss over a minibatch of size $M$ is

\begin{equation}
\textstyle
	\mathcal{L}(\theta) = \frac{1}{M} \sum_{i=1}^M
	\left\|\mu^i - \bar\mu^{i, \theta}\right\|_{L^2([t_i, t_i+\delta])}^2,
	\label{eq: training loss}
\end{equation}
where $\bar\mu^{i, \theta} = \bar\mu^{\gamma_\theta, \mu^i_0}$ represents the Nash equilibrium trajectory with $\gamma_\theta$, and $\delta \in (0, T]$ is the sub-trajectory length ($\delta = T$ recovers the full-trajectory loss). In practice, we use the Adam optimizer~\cite{Kingma2014}. The full training and parameter calibration procedure is summarized in Algo.~\ref{al: parameter calibration}. 
\paragraph{Picard Solver.}  Consider a fixed initial distribution $\mu_0$ and parameter $\gamma$. To find the Nash equilibrium, we iteratively update the control, value function, and population distribution \cite{Carlini2014}: at iteration $k$, we have the mean field $\mu^{k}$ and we solve the backward equation to compute $u^{(k+1)}$; from this, we compute the new control $\alpha^{(k+1)}$; finally we solve the forward dynamics using this new control and obtain $\mu^{(k+1)}$.

We stress that this depends on $\gamma$.
We denote a single full cycle of this forward-backward update as an operator $\Xi$, such that $\mu^{(k+1), \gamma} = \Xi(\mu^{(k), \gamma}, \gamma)$. The mean field is expected to converge towards a fixed point $\bar\mu^\gamma$ satisfying $\bar\mu^\gamma = \Xi(\bar\mu^\gamma, \gamma)$, which corresponds to the Nash equilibrium for parameter $\gamma$. The result depends on the initial distribution $\mu_0$ at time $0$, and we denote the fixed-point solver that computes this equilibrium as $\Phi$, mapping the parameters and initial condition to the equilibrium trajectory: $\bar\mu^{\gamma,\mu_0} = \Phi(\gamma, \mu_0)$. 
See App.~\ref{appendix: picard iteration} for details on Picard iteration for MFGs.

\paragraph{Pre-Evaluation.} 
	Instead of embedding the neural network dynamically within the Picard solver which would require evaluating and differentiating the network at every solver iteration to find a coupled fixed point $\bar\mu^{\gamma,\mu_0} = \Phi(\varphi_\theta(\cdot, \bar\mu^{\gamma,\mu_0}),\mu_0)$ where $\gamma = \varphi_\theta(\cdot, \bar\mu^{\gamma,\mu_0})$, we evaluate $\varphi_\theta$ \textit{a priori} on the observed trajectory $\mu^i$ at every point in time $t \in [0,T]$. Because the neural network is evaluated exclusively on the \emph{observed} distributions rather than the solver's intermediate distributions, the resulting parameter $\gamma^i_{\theta}(t) := \varphi_\theta(t, \mu^i_t)$ fed into the Picard solver is a function of time only. This approach completely decouples the parameter prediction from the solver's internal iterations, rendering the computational graph significantly more tractable. For a fixed network, if the observed trajectory is close to the corresponding self-consistent trajectory, then the pre-evaluated parameter path is close to the path that the same network would produce along that self-consistent pair; if these trajectories coincide, the two parameter paths coincide. App.~\ref{app:math_foundations} makes this precise as a conditional stability statement for the computational shortcut.

\paragraph{Implicit Differentiation.} Even with the parameter prediction decoupled, computing the parameter gradients $\nabla_\theta \mathcal{L}(\theta)$ naively requires unrolling the Picard iteration and backpropagating through the entire sequence of solver steps. To avoid the massive memory overhead and gradient instabilities associated with unrolling, we use implicit differentiation through the fixed point \cite{Bai2019, Blondel2022}. Since the converged trajectory $\bar\mu^{i, \theta}$ is a fixed point of the iterative step operator $\Xi$, i.e., $\bar\mu^{i, \theta} = \Xi(\bar\mu^{i, \theta}, \gamma_\theta^i)$, we can apply the Implicit Function Theorem to the residual $\Psi(\bar\mu, \gamma) = \bar\mu - \Xi(\bar\mu, \gamma) = 0$. This yields the exact Jacobian $\tfrac{\partial \bar\mu}{\partial \gamma} = ( I - \tfrac{\partial \Xi}{\partial \bar\mu} )^{-1} \tfrac{\partial \Xi}{\partial \gamma}$. This allows us to compute the vector-Jacobian products required for the backward pass by solving a single adjoint linear system at the equilibrium state, bypassing the need to store or unroll the forward solver iterations (see next subsection for justification). This idea generalizes to other ODE systems (see App.~\ref{appendix: mfc_derivation} for an extension to a mean field control example).

\begin{algorithm}[!htbp]
	\caption{Finding a data-driven optimal control in a finite-state MFG.}
	\begin{algorithmic}[1]
		\Require $K \in \NN$ epochs; $M \in \mathbb{N}$ minibatch size; data $\mathcal{D} \subseteq C([0,T]; \cP([d]))$; Picard solver $\Phi$; learning rate $\rho$; sub-trajectory length $\delta \in (0,T]$
		\State Initialize the neural network $\varphi_\theta$ with parameters $\theta = \theta_1$
		\For{$k = 1, \dots, K$}
		\State Sample a minibatch of $M$ trajectories $\mu^i \in \mathcal{D}$ and $M$ initial times $t_i \in [0,T-\delta]$
		\For{$i = 1, \dots, M$}
		\State For all $t \in [0, T]$, evaluate $\gamma^i_{\theta_k}(t) \gets \varphi_{\theta_k}(t, \mu^i_t)$ \Comment{{\footnotesize Pre-evaluate $\gamma(t)$ from observations}}
		\State $\bar\mu^{i, \theta_k} \gets \Phi\left(\gamma^i_{\theta_k}(\cdot), \mu^i_0\right)$ \Comment{{\footnotesize Solve MFG over the full time horizon $[0,T]$}}
		\State $L_i \gets \|\mu^{i} - \bar\mu^{i,\theta_k}\|_{L^2([t_i, t_i+\delta])}^2$ \Comment{{\footnotesize Calculate loss on the sub-trajectory}}
		\EndFor
		\State $\mathcal{L}(\theta_k) \gets \frac{1}{M} \sum_{i=1}^M L_i$ \Comment{{\footnotesize Compute average loss}}
		\State $\theta_{k+1} \gets \theta_k - \rho\nabla_\theta \mathcal{L}(\theta_k)$ \Comment{{\footnotesize Implicit differentiation update}}
		\EndFor
		\State \Return{trained neural network parameters $\theta_{K+1}$}
	\end{algorithmic}
	\label{al: parameter calibration}
\end{algorithm}

\subsection{Theoretical Guarantees}\label{sec: theoretical analysis}

While a theoretical analysis is not the main focus of this work,  we investigate the key mechanisms guaranteeing the soundness of our method. In this section, we give informal statements; the full assumptions, rigorous statements, and complete proofs are deferred to App.~\ref{app:math_foundations}. Since the algorithm uses discrete-time trajectories, we consider this setting for the analysis. The equilibrium trajectory $\bar\mu^{\gamma,\mu_0}$ is obtained as the output of $\Phi(\gamma, \mu_0)$. Next, we assume that:
\begin{assumption}[Contraction and Differentiability (Informal)]
We operate under standard structural conditions: namely that the Picard operator is a contraction map (often guaranteed by sufficiently small time horizons or the use of damped iterations) and that it is smoothly differentiable w.r.t. $\mu$ and $\gamma$. 
\end{assumption}
We first establish the correctness of the adjoint iteration used for scalable parameter calibration in Algo.~\ref{al: parameter calibration}. The following guarantees that our gradient descent updates are correct, bypassing the need for explicit formation and inversion of the high-dimensional Jacobian matrix $(I - \mathcal{J}_\mu)$.

\begin{proposition}[Adjoint Gradient Computation (Informal)]\label{prop:adjoint_informal}
	The sequence of vectors $w^{(k)}$ computed by the adjoint backward iteration converges to a unique limit $w^*$. Furthermore, the implicit gradient of the mean field loss $\mathcal{L}(\gamma)$ with respect to $\gamma$ is exactly recovered by $\nabla_\gamma \mathcal{L} = \mathcal{J}_\gamma^\top w^*$, where $\mathcal{J}_\gamma$ denotes the Jacobian matrix of the Picard operator with respect to the parameters $\gamma$.
\end{proposition}

See Prop.~\ref{prop:adjoint} for the details. In Algorithm~\ref{al: parameter calibration}, this result is applied to the pre-evaluated parameter path
\(
\Gamma_i(\theta):=(\varphi_\theta(t_k,\mu^i_{t_k}))_{k=0}^N.
\)
The adjoint solve gives the gradient with respect to this path,
\(
\nabla_{\Gamma_i}\mathcal L_i=\mathcal J_\gamma^\top w_i^*.
\)
The gradient with respect to the neural-network parameters is then obtained by the ordinary chain rule,
\(
\nabla_\theta \mathcal L_i(\theta)
=
D_\theta\Gamma_i(\theta)^\top \mathcal J_\gamma^\top w_i^*,
\)
and the minibatch gradient is the average over sampled trajectories.

Next, we justify the computational decoupling idea described in the Pre-Evaluation paragraph of Sec.~\ref{sec: learning parameter estimates}. Instead of evaluating $\varphi_\theta$ on the solver's intermediate trajectories, we evaluate it directly on the observed trajectory $\mu^{\mathrm{obs}}$. The following result quantifies the error introduced by this pre-evaluation step.

\begin{proposition}[Pre-Evaluation Consistency (Informal)]\label{prop:consistency_informal}
    For a fixed network $\varphi_\theta$, the discrepancy between the parameter path obtained by evaluating $\varphi_\theta$ on $\mu^{\mathrm{obs}}$ and the path obtained by evaluating the same network on the coupled equilibrium trajectory is bounded proportionally by $\|\bar\mu-\mu^{\mathrm{obs}}\|_\infty$. In particular, if the observed and equilibrium trajectories coincide, then the decoupled and coupled parameter paths coincide.
\end{proposition}

 See Prop.~\ref{prop:consistency} for details. Finally, we establish a generalization bound that characterizes the population risk (the expected loss on unseen initial distributions) of the learned parameters. A central challenge is that the equilibrium map $\Phi$ is defined implicitly. However, the contractive property of the Picard operator ensures that $\Phi$ is Lipschitz continuous, which tightly controls the complexity of the neural network's hypothesis space when composed with the fixed-point solver.

\begin{proposition}[Uniform Generalization Bound (Informal)]\label{prop:generalization_informal}
    Let $\Theta \subset \RR^p$ be a compact parameter space, and suppose we observe $M$ independent training trajectories. Then, with high probability, 
    the population risk $\mathcal{R}(\theta)$ is bounded above by the empirical training risk $\hat{\mathcal{R}}(\theta)$, uniformly for all $\theta \in \Theta$, by an error term of order $\mathcal{O}(M^{-1/2})$ up to an additive term proportional to the covering radius.
\end{proposition}

See Prop.~\ref{prop:generalization} for details. 
This guarantee applies to the i.i.d. sampling setting stated in the proposition; the real-data experiments below are empirical evaluations and may involve temporal dependence or domain shift. We empirically validate the above convergence rate on an example in App.~\ref{app:generalization_algo}.

\section{Test Cases}\label{sec: examples}
We consider the following examples to provide numerical evidence for the accuracy and generalization of our approach. We provide a brief overview here, and refer to App.~\ref{appendix: experimental_details} for the full game descriptions, synthetic data generation and experimental details.

\paragraph{Linear-Quadratic.} We consider a linear-quadratic (LQ) MFG that is adapted from \citet[Section 7.1]{Cohen2024} and can be regarded as a crowd motion model that is purely based on the individual player's control. The model has a flexible number of states $d$, bounded continuous transition-rate controls $a_{xy}\in[1,3]$ for $y\neq x$, and one learnable parameter ($\gamma \equiv b \in \RR$). We instantiate the game using synthetically generated data under different conditions of synthetic noise drawn from a Dirichlet distribution. 

\paragraph{Cybersecurity.}
The cybersecurity model considers a botnet-like population of computers that may become victims of a cyberattack. It is originally introduced by \citet{Kolokoltsov2016}, and has subsequently been revisited by \citet[Section 7.2.3]{Carmona2018}, solved numerically by \citet[Section 7.2]{lauriere2021} and solved approximately via the master equation by \citet[Section 7.4]{Cohen2024}. We focus on the implementation found in \cite{Carmona2018}. The game has $d=4$ states, corresponding to protected/unprotected and susceptible/infected computers, a discrete action space $\cA=\{0,1\}$, and eight learnable parameters ($\gamma\in\RR^8$). Just like the linear-quadratic example, the game is instantiated using synthetically generated data under different conditions of synthetic noise drawn from a Dirichlet distribution.

\paragraph{Susceptible-Infected-Recovered.}
The Susceptible-Infected-Recovered (SIR) model is well-studied in the MFG literature \cite{Kermack1927}. Related works include, but are not limited to, dynamics related to vaccination \cite{Hubert2018, Doncel2022}, social distancing \cite{Lee2021, aurell2022finite,Buckley2025}, both vaccination and social distancing \cite{liu2025,liu2026modelingepidemicspreadstrategic} and optimal interventions via governing institutions \cite{Aurell2022}. In this section, we consider a simplified variant of the Stackelberg MFG proposed by \citet{Aurell2022}, allowing a population of homogeneous players to control their own interaction rate. The game has $d=3$ states ($[d]=\{S, I, R\}$), a continuous action space $\cA =[0, 1]$ and four learnable parameters ($\gamma \in \RR^4$). 
We analyze a publicly available influenza monitoring dataset from the \textit{National Respiratory and Enteric Virus Surveillance System (NREVSS)} \cite{CDC},, which describes the weekly positive rates of influenza tests per U.S. state as reported to the \textit{U.S. Centers for Disease Control and Prevention} (CDC). We filter the data by U.S. state and year, and use the percentage of positive Influenza A tests as an observation proxy for the infected coordinate $\mu(I)$, rather than as a fully observed population distribution. We capture the yearly epidemic peaks, where each state and year constitutes one observed trajectory ($N=30$ weeks). We approximately fit our 30-week window such that each flu season starts around $t=0$. The filtering is approximate on purpose to generate different initial states $\mu_0$. We drop all trajectories that contain more than 15 missing values, yielding 279 observed trajectories that we chronologically split into training and testing sets. 
\paragraph{Urban Mobility.}
As a last example, we consider a bike-sharing model where we track the cumulative arrivals and departures at six stations. The example allows for the infusion of strategic behavior in the players' commute by penalizing both the deviation from natural bike demand patterns and arriving at congested stations. Therefore, the game has $d=6$ states, continuous nonnegative outgoing transition-rate controls, and learnable base-rate and cost parameters.
We use the publicly available \textit{Citi Bike trip-level} dataset from New York City \cite{Citibike}. Each record contains a trip's start time, end time, origin station, and destination station. We select a cluster of $K=5$ stations in Midtown Manhattan and introduce a sixth ``external'' state representing the rest of the network. By tracking cumulative arrivals and departures at each station in 15-minute bins over the interval $[06{:}00, 20{:}00]$, we reconstruct the empirical mean field distribution $\hat{\mu}(t) \in \cP([d])$. Each weekday constitutes one observed trajectory ($N = 64$ time steps), providing approximately 20 trajectories per month that we split into training and testing sets.

\section{Experiments \& Results}\label{sec: experiments}
We evaluate the objectives that we introduced in Sec.~\ref{sec: background} through numerical experiments based on the examples from Sec.~\ref{sec: examples}. 
The synthetic experiments are closest to the i.i.d. setting used in the theoretical generalization result, whereas the influenza and bike-sharing experiments use real data with chronological structure and should be interpreted as empirical evaluations beyond that formal guarantee.
All experiments have been run $n=5$ times with different random seeds. We report the mean $\pm$ standard deviation over all runs. 

\paragraph{Objective 1: MFG Calibration.}
We assess the calibration quality by evaluating the predicted mean field flow $\mu^\theta_t$ and $L_2$ performance on the LQ and Cybersecurity MFGs with synthetic data. The true parameters $\gamma_i \in [\gamma_{i,\min}, \gamma_{i,\max}]$ take three functional forms: constant, time-dependent and mean field-dependent (see App.~\ref{appendix: experimental_details}). Fig.~\ref{fig: calibration_lq_performance} shows the predictions of the calibrated MFG on two data points from the test set, together with the $L_2$ error on the predicted mean field flow (train and test) and parameters for the LQ example. From the figure, it follows that the predicted mean field flows closely match the observations as the training converges and the $L_2$ loss on the parameter estimates is minimized. Results for the cybersecurity example are in App.~\ref{appendix: extra_analyzes}.
\begin{figure}[!htbp]
    \centering
    \begin{subfigure}[b]{0.24\textwidth}
    \includegraphics[width=\textwidth]{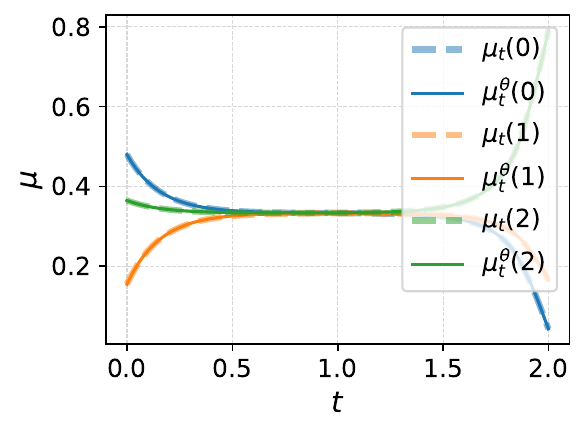}
    \end{subfigure}
    \hfill
    \begin{subfigure}[b]{0.24\textwidth}
        \includegraphics[width=\textwidth]{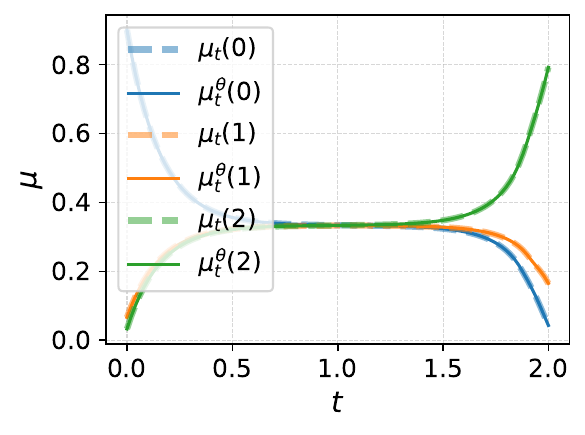}
    \end{subfigure}
    \hfill
        \begin{subfigure}[b]{0.24\textwidth}
    \includegraphics[width=\textwidth]{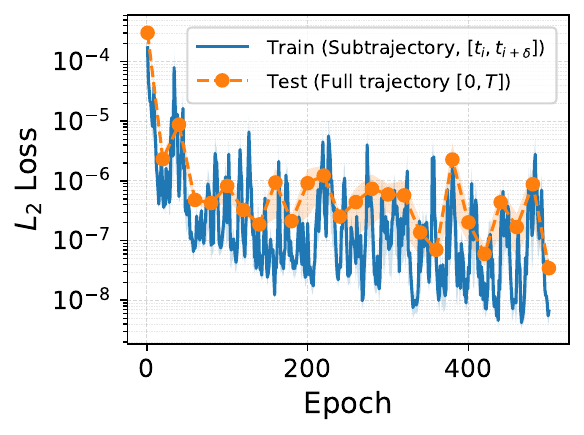}
    \end{subfigure}
        \begin{subfigure}[b]{0.24\textwidth}
    \includegraphics[width=\textwidth]{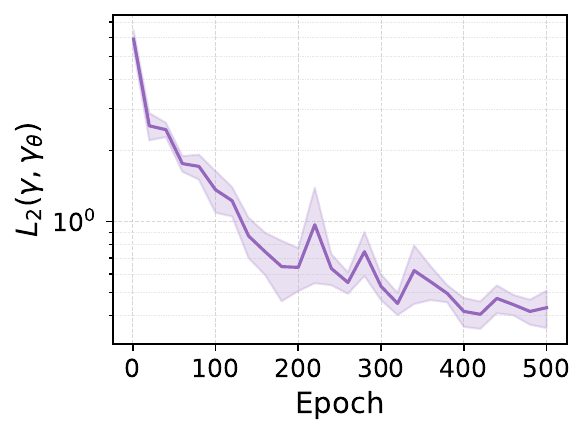}
    \end{subfigure}
    \caption{LQ MFG, from left to right: Predicted mean field flow on two random samples from the test set, $L_2$ loss on mean fields (train and test), and $L_2$ error on predicted parameters.}
        \label{fig: calibration_lq_performance}
\end{figure}

\paragraph{Objective 2: Comparing MFGs and Forward Models.} 
To isolate and demonstrate the forward-looking power of the MFG formulation, we compare our model to a mean field dynamics baseline that describes the population flow via exogenous transition rates. The MFG model contains the forward dynamics baseline as a special case when the strategic costs and terminal cost vanish and the HJB solution is selected as $u\equiv0$, in which case $a_{xy}^*=q_{xy}^{\mathrm{base}}$. We evaluate the two models under the following three testing regimes on the urban mobility example. More elaborate experimental descriptions and results are provided in App.~\ref{appendix: biking_running}.

\textit{Mean field dynamics model (baseline).} In this model, agents follow exogenous, time-dependent transition rates without any strategic behavior. The population distribution evolves according to the Kolmogorov forward equation:
$
	\dot{\mu}_i = \sum_{j \neq i} \mu_j \, q_{ji}(\gamma) - \mu_i \sum_{j \neq i} q_{ij}(\gamma),
$

where the rates $q_{ij}(t) = \mathrm{softplus}(\gamma_{ij}(t))$ are parameterized by a neural network $\varphi_\theta : [0,T] \to \RR^{d(d-1)}$. Depending on the test regime (see below), $\varphi_\theta$ is either a constant vector, a single affine layer, or a multi-layer perceptron. The network is trained by minimizing the $L_2$ discrepancy between the predicted and observed trajectories, analogously to the procedure described in Sec.~\ref{sec: method}.

\textit{Mean field game model.} In the MFG formulation, each commuter optimizes their destination choice to minimize a cost that penalizes both deviating from the natural demand pattern and arriving at congested stations. The running cost is:
\begin{equation}\label{eq:bike_cost}
\textstyle
	f(x, a, \mu) = b \sum_{y \neq x} (a_{xy} - q_{xy}^{\mathrm{base}})^2 + c \sum_{y \neq x} \mu(y) \, a_{xy},
\end{equation}
where $q_{xy}^{\mathrm{base}}(t) = \mathrm{softplus}(\gamma_{xy}(t))$ are the same base rates as in the baseline model, $b > 0$ controls the cost of deviating from the natural demand, and $c \geq 0$ weights the destination congestion penalty. The terminal cost is $g(x) = w_g(x)$, where $w_g \in \RR^d$ is a learnable vector initialized to zero. We consider three test cases of increasing complexity and compare all model predictions with the observed mean field mass:
\begin{itemize}[leftmargin=20pt,itemsep=1pt, topsep=0pt]
    \item\textbf{Test 1: Constant Base Rates.} We restrict the base transition rates to be entirely constant over time: $\varphi_\theta(t) = \gamma_0$ (a learned but time-independent vector). A standard ODE with constant rates decays toward a stationary distribution and cannot recreate the complex, double-peaked morning and evening commute patterns. The MFG, however, learns a non-trivial terminal cost $w_g(x)$ representing commuters' end-of-day location preferences (Fig.~\ref{fig:MFG_MF_comparison}a). 
    \item\textbf{Test 2: Linear Base Rates.} We allow the base rates to be affine functions of time: $\varphi_\theta(t) = W t + b_{\text{bias}}$, where $W \in \RR^{d(d-1) \times 1}$ and $b_{\text{bias}} \in \RR^{d(d-1)}$. This enables the baseline to capture monotonic trends over the day but prevents it from fitting highly non-linear dynamics (Fig.~\ref{fig:MFG_MF_comparison}b). 
    \item\textbf{Test 3: Restricted Dynamics with General Cost.} We consider a practitioner who relies on a rigid phenomenological model for the baseline dynamics (linear base rates, as in Test~2) and does not wish to alter it. Instead, we augment the MFG with a flexible cost network $\psi_\phi : [0,T] \times \cP([d]) \to \RR^{d(d-1)}$ (a 2-layer MLP with 64 hidden units and $\tanh$ activations), which replaces the scalar congestion cost $c \cdot \mu(y)$ in~\eqref{eq:bike_cost} with a general, potentially negative, cost vector $\gamma_{\text{cost}}(t, \mu_t) := \psi_\phi(t, \mu_t) \in \RR^{d(d-1)}$. The running cost becomes:
    $
    	f(x, a, \mu, t) = b \sum_{y \neq x} (a_{xy} - q_{xy}^{\mathrm{base}})^2 - \sum_{y \neq x} [\gamma_{\text{cost}}]_{xy}(t, \mu) \, a_{xy},
    $

    and the corresponding optimal control is:
    $
    	a_{xy}^* = \max\!\Big(0, \; q_{xy}^{\mathrm{base}} + \frac{[\gamma_{\text{cost}}]_{xy} - (u(y) - u(x))}{2b}\Big).
    $

    The sign of the linear cost term is reversed compared to~\eqref{eq:bike_cost}, and the network $\psi_\phi$ can output negative values, enabling the model to learn both congestion penalties and incentives (Fig.~\ref{fig:MFG_MF_comparison}b). 
\end{itemize}
We show the predicted mean field mass for a single station (Station 0) for the three test cases in Fig.~\ref{fig:MFG_MF_comparison}. 
\begin{figure}[!htbp]
    \centering
    \begin{subfigure}[b]{0.24\textwidth}
    \includegraphics[width=\textwidth]{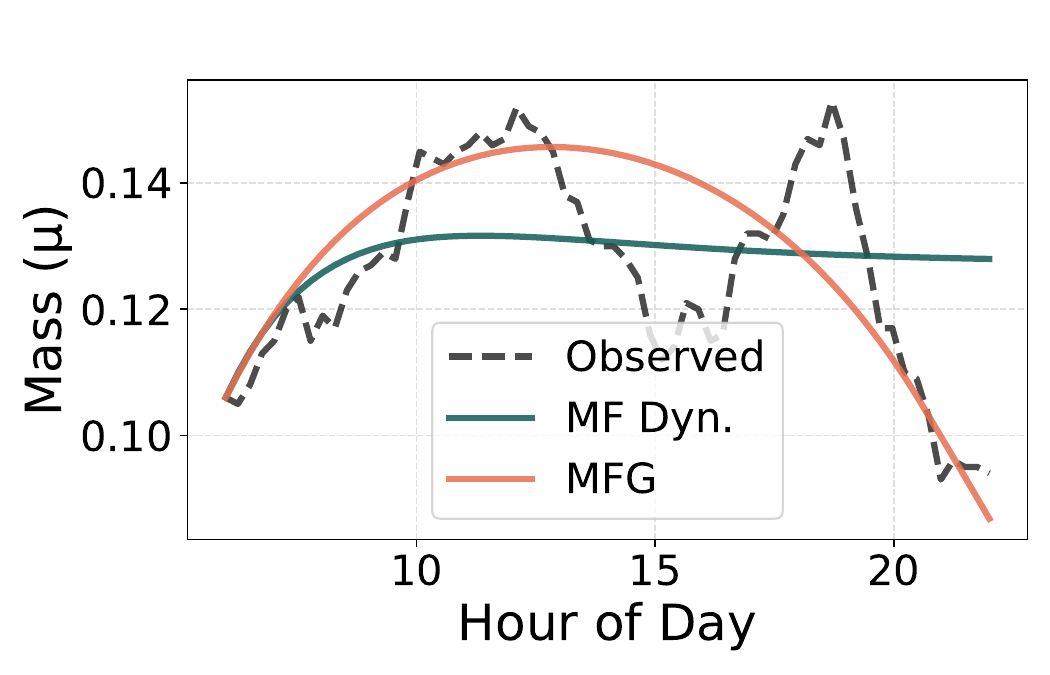}
    \caption{}
    \end{subfigure}
    \hfill    \begin{subfigure}[b]{0.24\textwidth}
    \includegraphics[width=\textwidth]{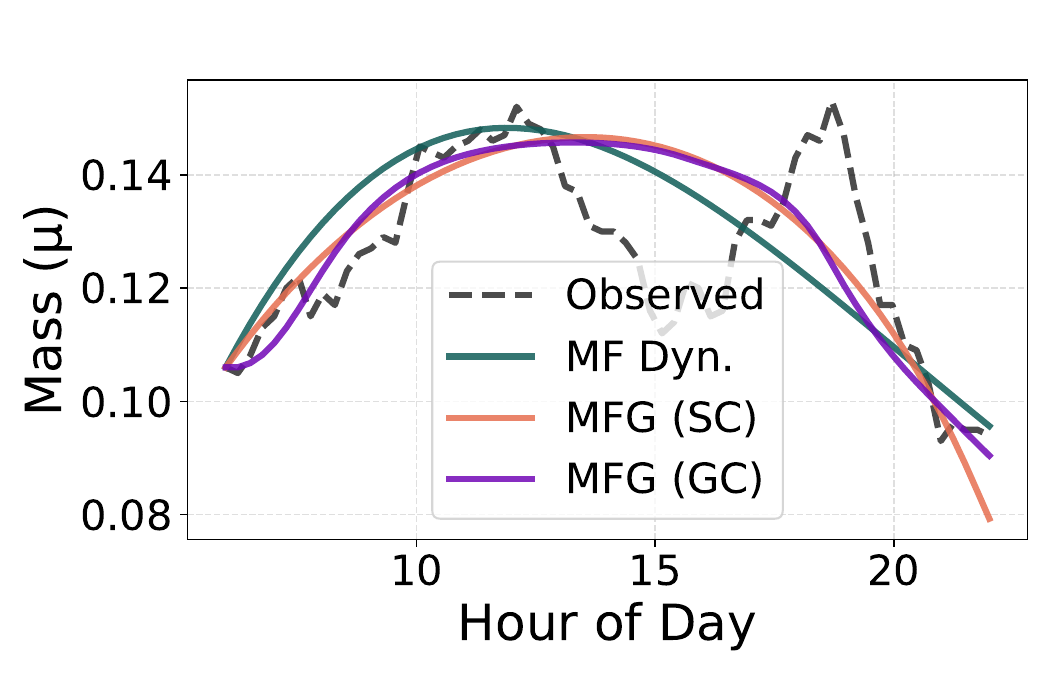}
    \caption{}
    \end{subfigure}
    \hfill
        \begin{subfigure}[b]{0.24\textwidth}
    \includegraphics[width=\textwidth]{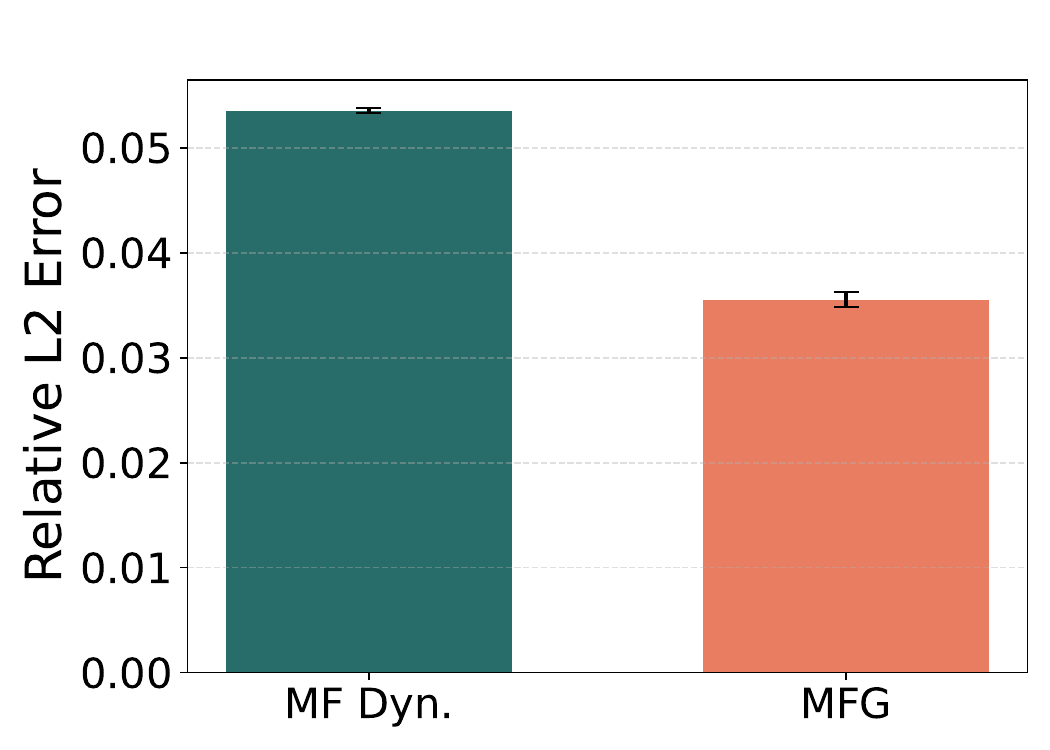}
    \caption{}
    \end{subfigure}
    \hfill
     \begin{subfigure}[b]{0.24\textwidth}
    \includegraphics[width=\textwidth]{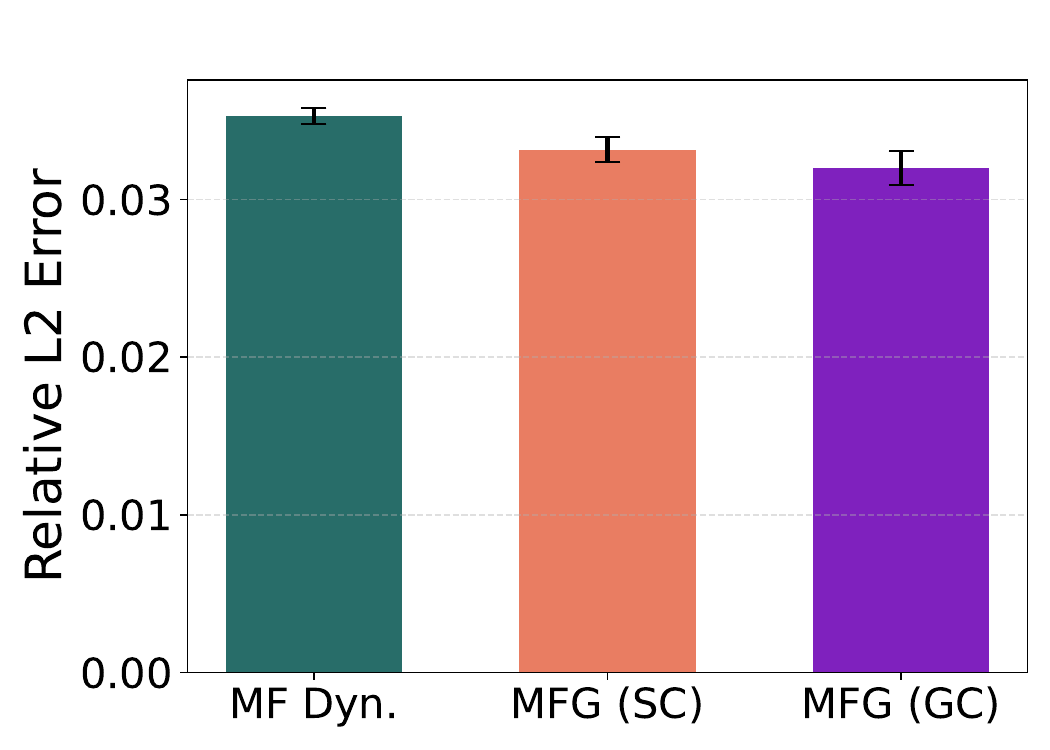}
    \caption{}
    \end{subfigure}
    \vspace{-6pt}
    \caption{Predicted mean field mass for Station 0 under three test cases: (a) constant base rates, (b) linear base rates with scalar costs (SC) and linear base rates with general costs (GC). Results are shown for both the mean field dynamics baseline (MF Dyn.) and the MFG throughout. Relative errors with respect to the observed mean field mass compare the models with (c) constant base rates and (d) linear base rates with scalar and general costs.}
    \label{fig:MFG_MF_comparison}
\end{figure}
Although neither model perfectly matches the observations (because it comes from real data and not synthetic MFG-based data), the MFG benefits from the inclusion of the terminal cost $g_w$ that is propagated backward in time by the HJB. This causes agents to strategically adjust their flows throughout the day and outperform the MF baseline on the in-distribution weekday test shown in this experiment (Fig.~\ref{fig:MFG_MF_comparison}c, d).

\paragraph{Objective 3: Noise Robustness, Out-of-Distribution Forecasts and Novel Scenarios.}
Fig.~\ref{fig:noise_robustness} presents an analysis of robustness to noise in the observed trajectories $\mu_t$ for the LQ model. As the noise level increases, the mean-field $L_2$ error worsens, while the parameter estimates remain qualitatively close to the ground truth in the displayed time- and mean-field-dependent examples.

\begin{figure}[ht]
    \centering
    \begin{subfigure}[b]{0.24\textwidth}
        \includegraphics[width=\textwidth]{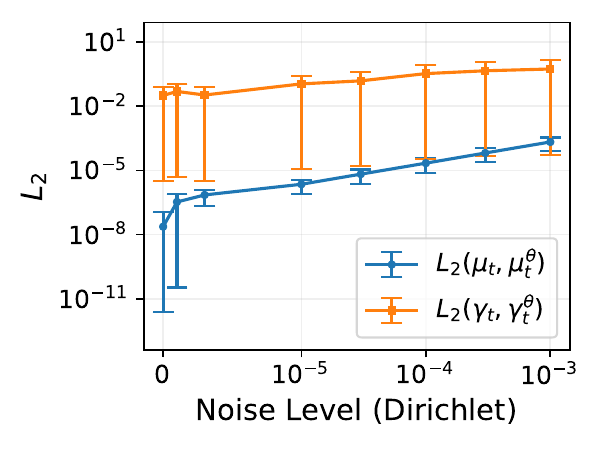}
    \end{subfigure}
    \begin{subfigure}[b]{0.24\textwidth}
        \includegraphics[width=\textwidth]{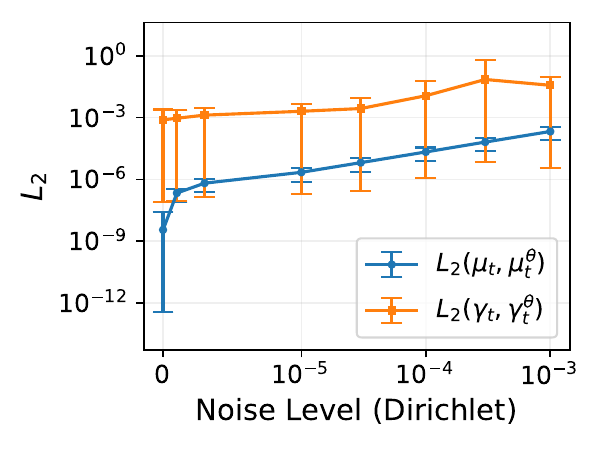}
    \end{subfigure}
    \hfill
    \begin{subfigure}[b]{0.24\textwidth}
        \includegraphics[width=\textwidth]{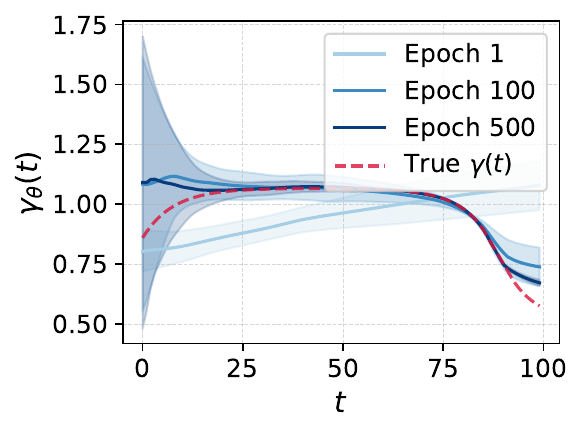}
    \end{subfigure}
    \hfill
    \begin{subfigure}[b]{0.24\textwidth}
        \includegraphics[width=\textwidth]{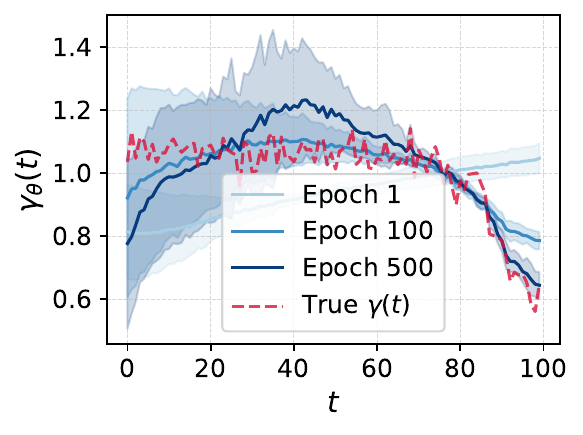}
    \end{subfigure}
    \caption{Noise robustness for the parameter estimates $\gamma_\theta$ and the mean field predictions $\mu_\theta$ for the LQ model. From left to right: time-dependent and mean field-dependent, predictions for the mean-field dependent $\gamma_t$ at $\sigma=0$ and $\sigma=10^{-3}$.}\label{fig:noise_robustness}
\end{figure}

To further demonstrate the noise robustness of the mean field flow predictions, we assess the calibration accuracy in a real-world system by modelling Influenza activity in the United States. We train the calibration method on historical weekly observations and evaluate it on held-out flu seasons. The resulting out-of-distribution prediction of the SIR dynamics is associated with an error of $L_2(\mu_t(I), \mu^\theta_t(I)) = 4.36 \times 10^{-3} \pm 1.36 \times10^{-2}$, and representative predictions for 2023--2025 are shown in Fig.~\ref{fig:sir_influenza_trajectories}. Because test positivity is not the same object as disease prevalence and only the infected coordinate is observed, this experiment should be interpreted as calibration to an influenza activity proxy. The susceptible and recovered trajectories shown in Fig.~\ref{fig:sir_influenza_trajectories} are model-implied latent coordinates, not independently observed CDC quantities.
\begin{figure}[h]
    \centering
    \begin{subfigure}[b]{0.28\textwidth}
        \includegraphics[width=\textwidth]{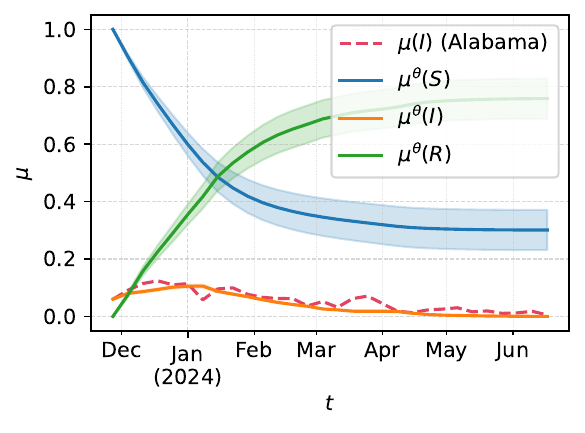}
    \end{subfigure}
    \hfill
    \begin{subfigure}[b]{0.28\textwidth}
        \includegraphics[width=\textwidth]{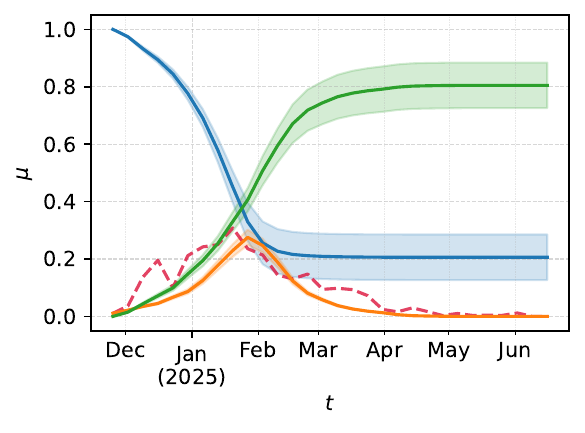}
    \end{subfigure}
    \hfill
    \begin{subfigure}[b]{0.28\textwidth}
\includegraphics[width=\textwidth]{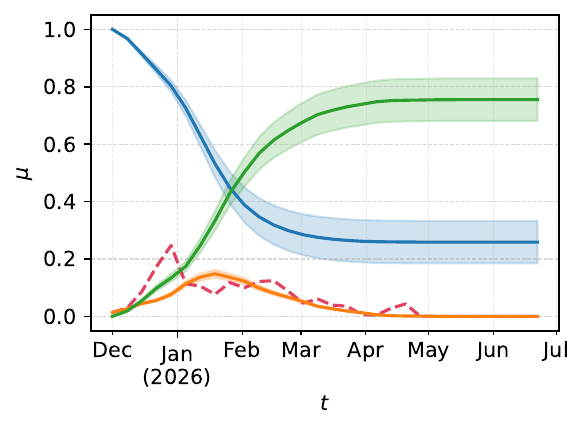}
    \end{subfigure}
    \caption{Evolution of the mean field distribution trajectories $\mu^\theta_t$ for the U.S.~state Alabama in for representative held-out seasons from 2023 to 2025. The predictions are compared to the ground truth $\mu_t(I)$. The predicted mean field flow describing the number of susceptible and recovered agents is derived from the SIR definition given $\mu^\theta_t(I)$. }\label{fig:sir_influenza_trajectories}
\end{figure}
Even though the predicted curves do not perfectly match the observed fraction of positive tests, the model captures important characteristics such as the magnitude and duration of the infection peak. 

Lastly, we evaluate our model on its zero-shot adaptability to novel scenarios. In doing so, we highlight a fundamental advantage of MFGs: in contrast to a forward baseline model, game-theoretic models have the ability to evaluate counterfactual scenarios without retraining. We illustrate this by revisiting the baseline comparison in the urban mobility MFG, while simulating an announced station closure: the city announces that a given station will be closed for maintenance for a fixed period of time $t \in [t_s, t_e]$ and wishes to predict the behavior of users.

Since the baseline model has no value function or cost structure, the only way to encode the closure is to mechanically set the transition rates into Station 1 to zero during the closure window. Crucially, the rates remain unchanged before the closure at $t_s$, as the model has no mechanism to anticipate the disruption. In contrast, the closure can be encoded in the MFG model as a state-dependent running cost penalty. Because the HJB equation propagates the new cost function from terminal time $T$ to time $0$, the value function $u(t, x)$ at times $t < t_s$ already reflects the anticipated penalty.

Fig.~\ref{fig:bikeshare_intervention} compares the two calibrated models in this counterfactual closure scenario. The forward dynamics baseline follows its learned historical pattern until the closure begins, at which point the imposed closure constraint changes the trajectory. The MFG model, however, gradually drains the affected station before the closure starts, as the closure penalty is propagated backward through the HJB equation. This simulation illustrates how the game-theoretic model can produce anticipatory responses to announced interventions without retraining. 

\begin{figure}[ht]
    \centering
    \begin{subfigure}[b]{0.24\textwidth}
\includegraphics[width=\textwidth]{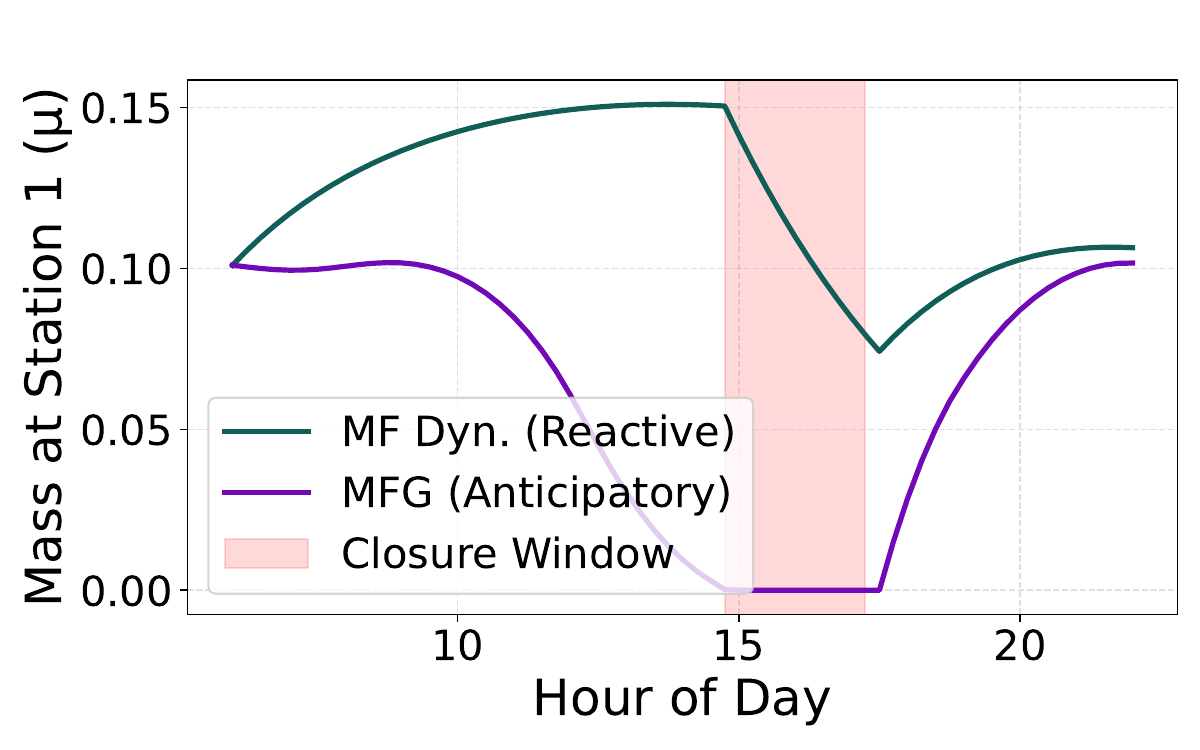}
\end{subfigure}
    \begin{subfigure}[b]{0.24\textwidth}
\includegraphics[width=\textwidth]{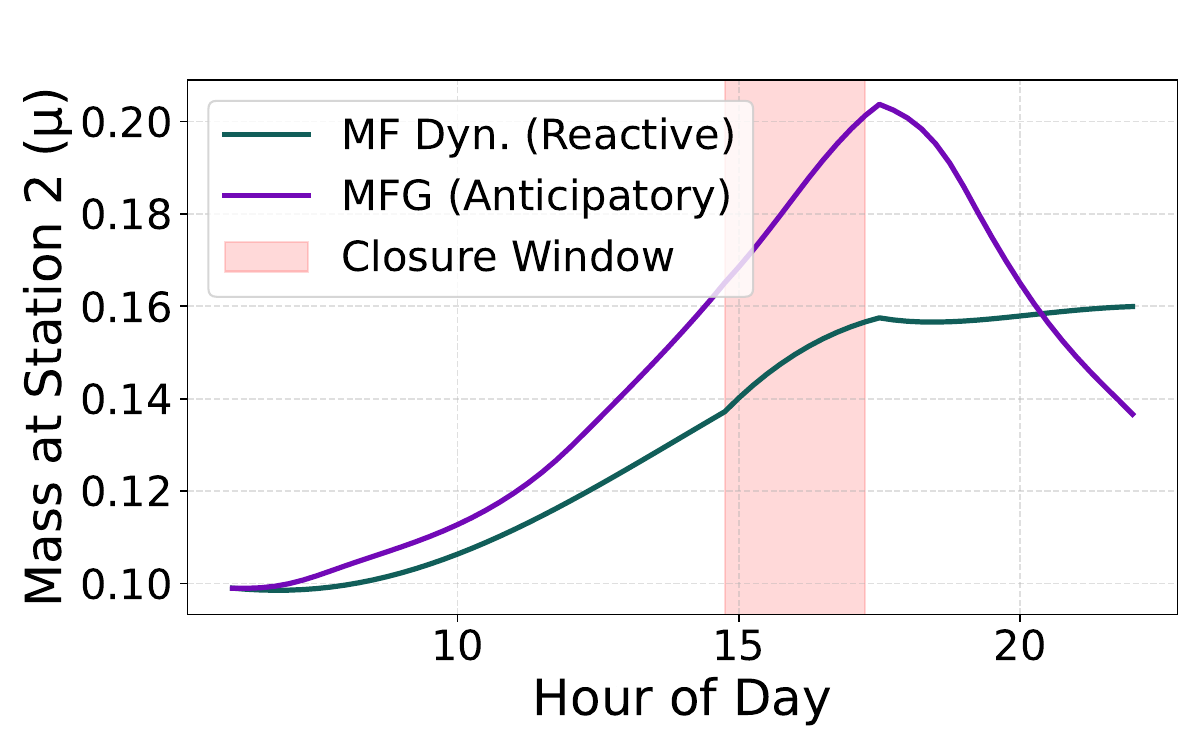}
    \end{subfigure}
    \hfill
        \begin{subfigure}[b]{0.24\textwidth}
\includegraphics[width=\textwidth]{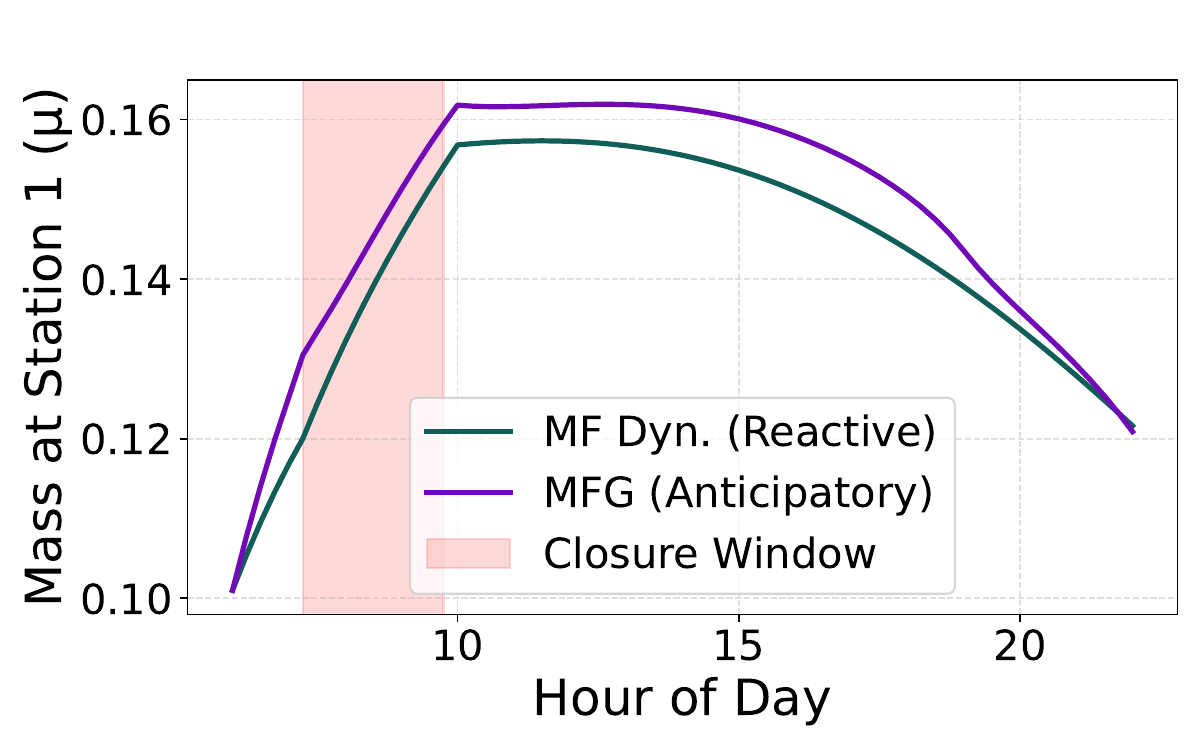}
\end{subfigure}
        \begin{subfigure}[b]{0.24\textwidth}
\includegraphics[width=\textwidth]{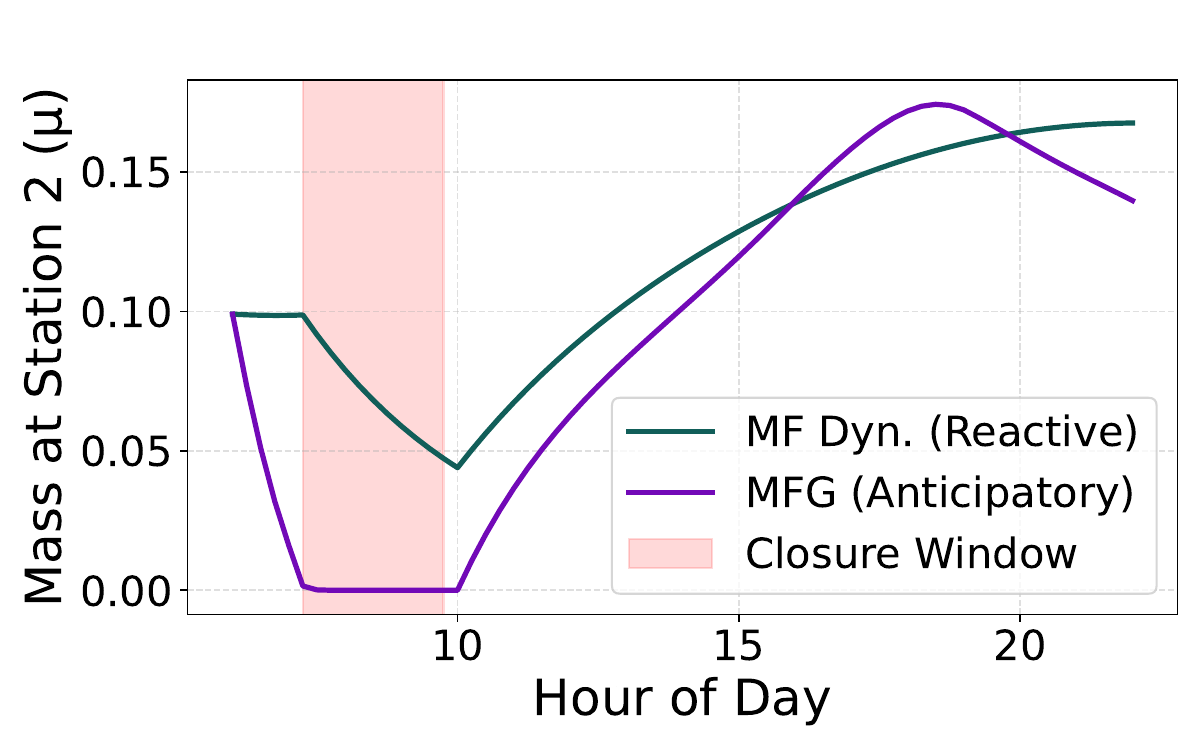}
    \end{subfigure}
    \caption{Counterfactual station-closure simulation, from left to right: station 1 (scenario 1), station 2 (scenario 1), station 1 (scenario 2), station 2 (scenario 2). The reactive MF dynamics baseline follows its historical routing pattern unchanged until the closure instant, then flatlines. The anticipatory MFG model propagates the penalty backward through the HJB equation, causing agents to reroute away from the station \emph{before} the closure begins. We refer to App.~\ref{app_bike_scenarios} for all other trajectories.}
    \label{fig:bikeshare_intervention}
\end{figure}

\section{Conclusion}
\label{sec: conclusion}
We introduced a novel, fully differentiable framework for learning trajectory-wise parameter paths induced by time- and mean-field-dependent neural parameterizations of finite-state mean field games directly from macroscopic population trajectories, supported by exact adjoint gradient computations and uniform generalization bounds under discrete-time contraction and smoothness assumptions. We validated the structural advantages of our approach across four domains, illustrating anticipatory responses in counterfactual station-closure simulations. 

\textbf{Limitations.} The current methodology is limited by its reliance on \textit{a priori} knowledge of the game's structural form and can struggle to generalize in out-of-distribution scenarios (such as the transition from weekday commutes to weekend leisure riding) when the underlying user utilities fundamentally change. Future work will focus on extending this framework to continuous state-action spaces, relaxing the i.i.d. and smoothness assumptions used in the theoretical analysis, relaxing the requirement for fully observable population densities, and integrating structural discovery mechanisms to improve robustness against unobserved distributional shifts.

\section*{Acknowledgements}
Tal Kachman and Anna C.M. Thöni acknowledge funding from the National Growth Fund project “Big Chemistry” (1420578) and the European Laboratory for Learning and Intelligent Systems (ELLIS). The authors have no conflicts of interest to declare.

\bibliography{references}
\clearpage

\appendix
\section{Picard Iteration for Solving MFGs}\label{appendix: picard iteration}
Picard iteration is a fixed-point numerical method for solving the forward-backward system associated with MFGs. In this section, we describe the implementation of the solver called in line~5 of Algo.~\ref{al: parameter calibration}. For a parameter path $\gamma=(\gamma_t)_{t\in[0,T]}$ and initial distribution $\mu_0$, we denote by
\(
    \bar\mu^{\gamma,\mu_0}=\Phi(\gamma,\mu_0)\in \cP([d])^{N+1}
\)
the discrete equilibrium trajectory returned by the solver; the associated value function and control are computed internally.

We aim to solve the system of PDEs described in~\eqref{eq: forward-backward} on the time interval $[0,T]$. 

Let $t_i=i\Delta t$, $i=0,\dots,N$, with $\Delta t=T/N$. Given a Picard iterate
$\mu^{(k)}=(\mu_i^{(k)})_{i=0}^{N}$, the backward step sets
\[
u^{(k+1)}_{N}(x)=g^\gamma(x,\mu^{(k)}_{N}),
\]
and, for $i=N-1,\dots,0$,
\[
u^{(k+1)}_i(x)
=
u^{(k+1)}_{i+1}(x)
+\Delta t\,
H^\gamma\!\left(
t_i,x,\mu_i^{(k)},u^{(k+1)}_{i+1}(\cdot),
\alpha^{H,k+1}_i(x)
\right),
\]
where
\[
\alpha^{H,k+1}_i(x)\in
\arg\min_{a\in\cA}
H^\gamma\!\left(t_i,x,\mu_i^{(k)},u^{(k+1)}_{i+1}(\cdot),a\right).
\]
The forward step then sets $\tilde\mu^{(k+1)}_0=\mu_0$ and, for $i=0,\dots,N-1$,
\[
\tilde\mu^{(k+1)}_{i+1}(x)
=
\tilde\mu^{(k+1)}_i(x)
+
\Delta t
\sum_{y\in[d]}
\tilde\mu^{(k+1)}_i(y)\,
Q^\gamma\!\left(t_i,\alpha^{F,k+1}_i(y),\tilde\mu^{(k+1)}_i\right)[y,x],
\]
with
\[
\alpha^{F,k+1}_i(x)\in
\arg\min_{a\in\cA}
H^\gamma\!\left(t_i,x,\tilde\mu^{(k+1)}_i,u^{(k+1)}_i(\cdot),a\right).
\]
Using fixed-point iteration, we produce approximations of the value function $\{u_i^{\gamma,\mu_0}(x)\}_{i=0}^{N}$ and mean field flow $\{\mu_i^{\gamma,\mu_0}(x)\}_{i=0}^{N}$ for $x \in [d]$ along the discrete time grid. To improve stability across a broad range of interaction strengths, we use a damping coefficient $\lambda \in [0,1)$ and set $\lambda=0.5$ in our experiments, blending each forward FPK update with the previous distribution iterate. The full fixed-point algorithm is summarized in Alg.~\ref{al: picard}.

For ease of notation, we omit the dependence on $\gamma$ and $\mu_0$ in Algorithm~\ref{al: picard}. The algorithm approximates the solution of a single MFG $\cM^\gamma$ for a fixed initial law $\mu_0$; in the calibration procedure, it is called separately for each sampled trajectory.

\begin{algorithm}[ht]
	\caption{Picard iteration for approximating the Nash equilibrium of finite-state MFGs}
	\begin{algorithmic}[1]
		\Require Parameter path $\gamma$, initial law $\mu_0$, initial trajectory $\mu^{(0)}$, damping $\lambda\in[0,1)$, tolerance $\varepsilon$, maximum iterations $K$
\State $k\gets 0$, $\delta\gets\infty$
\While{$\delta>\varepsilon$ and $k<K$}
    \State $u^{(k+1)} \gets \mathrm{BackwardHJB}(\mu^{(k)},\gamma)$
    \State $\tilde\mu^{(k+1)} \gets \mathrm{ForwardFPK}(u^{(k+1)},\gamma,\mu_0)$
    \State $\mu^{(k+1)} \gets (1-\lambda)\tilde\mu^{(k+1)}+\lambda\mu^{(k)}$
    \State $\delta \gets \|\mu^{(k+1)}-\mu^{(k)}\|_\infty$
    \State $k\gets k+1$
\EndWhile
\State \Return{$\mu^{(k)}$}
	\end{algorithmic}
	\label{al: picard}
\end{algorithm}

\section{Details on the Theoretical Guarantees}\label{app:math_foundations}

We provide theoretical results supporting the algorithmic choices made in Sec.~\ref{sec: learning parameter estimates}. The existence and uniqueness of Nash equilibria in finite-state MFGs have been established under various structural conditions, including monotonicity and smallness of the time horizon; see e.g.~\citet{Gomes2013} and \citet{CecchinFischer2020}. We take these as a starting point and work throughout in the discrete-time setting used by our numerical implementation.

\paragraph{Discrete-time setup.} We discretize the time horizon $[0, T]$ into $N$ steps of size $\Delta t = T/N$, with time points $t_k = k\,\Delta t$ for $k = 0, \ldots, N$. A discrete trajectory is a sequence $\mu = (\mu_0, \mu_1, \ldots, \mu_N)$ with $\mu_k \in \cP([d])$ for each $k$. We identify $\mu$ with an element of $\RR^{(N+1) \times d}$ and equip this space with the norm $\|\mu\|_\infty := \max_{0 \le k \le N} \max_{x \in [d]} |\mu_k(x)|$. 
For discrete trajectories, we use
\[
\|\mu\|_{2,N}^2 := \Delta t \sum_{k=0}^{N}\sum_{x\in[d]} |\mu_k(x)|^2.
\]
Throughout this appendix, $\Gamma$ denotes the finite-dimensional set of discretized parameter paths after vectorization, e.g. $\gamma=(\gamma_0,\ldots,\gamma_N)\in\RR^{(N+1)n}$. Whenever derivatives are used, we assume that $\Xi$ admits a $C^1$ extension to an open neighborhood of the relevant subset of $\RR^{(N+1)d}\times\RR^{(N+1)n}$; equivalently, derivatives on the simplex may be understood in local or tangent coordinates.

The discrete Picard operator $\Xi : \cP([d])^{N+1} \times \Gamma \to \cP([d])^{N+1}$ performs one full forward-backward cycle: given a trajectory $\mu$ and parameters $\gamma$, it (1)~solves the backward HJB equation on the discrete grid to obtain the value function, (2)~extracts the optimal control at each time step, and (3)~propagates forward the FPK equation to produce a new trajectory $\Xi(\mu, \gamma)$. We denote by $\Phi(\gamma, \mu_0)$ the fixed-point solver that returns the equilibrium trajectory $\bar\mu^\gamma$ satisfying $\bar\mu^\gamma = \Xi(\bar\mu^\gamma, \gamma)$.

\begin{assumption}[Picard Contractivity]\label{ass:contraction}
	For each $\gamma \in \Gamma$, the Picard operator $\Xi(\cdot, \gamma) : \cP([d])^{N+1} \to \cP([d])^{N+1}$ is a contraction in the $\|\cdot\|_\infty$ norm with constant $\kappa \in [0, 1)$:
	\[
		\|\Xi(\mu, \gamma) - \Xi(\mu', \gamma)\|_\infty \leq \kappa \, \|\mu - \mu'\|_\infty, \qquad \forall\, \mu, \mu' \in \cP([d])^{N+1}.
	\]
\end{assumption}

\begin{remark}[On damping]\label{rem:damping}
	Assumption~\ref{ass:contraction} is a standard smallness condition, typically satisfied when the time horizon $T$, the Lipschitz constants of $f$ and $Q$ in $\mu$, or the number of states $d$ are sufficiently small \cite{Gomes2013}. In practice, one often replaces $\Xi$ by a damped operator $\Xi_\alpha(\mu, \gamma) := (1 - \alpha)\,\Xi(\mu, \gamma) + \alpha\, \mu$ for some $\alpha \in (0, 1)$, which improves numerical stability. 

    In our experiments, damping is used as an empirical stabilization device and, when using large enough damping coefficient, the damped Picard iteration always converges.
\end{remark}

Under Assumption~\ref{ass:contraction}, the Banach fixed-point theorem guarantees that for each $\gamma \in \Gamma$, the iteration $\mu^{(k+1)} = \Xi(\mu^{(k)}, \gamma)$ converges to a unique fixed point $\bar\mu^\gamma = \Xi(\bar\mu^\gamma, \gamma)$ at geometric rate $\kappa^k$.

\paragraph{Differentiability of the equilibrium map.} The equilibrium trajectory $\bar\mu^\gamma$ is an implicit function of $\gamma$ defined by the fixed-point equation $\bar\mu^\gamma - \Xi(\bar\mu^\gamma, \gamma) = 0$. In the discrete-time setting, $\Xi$ is a composition of standard smooth operations (e.g., softplus, matrix-vector products).  We formalize the required regularity as follows.

\begin{assumption}[Smoothness]\label{ass:smooth}
	The Picard map $\Xi$ admits a continuously differentiable extension, in both variables, to an open neighborhood of the trajectories and parameter paths considered in $\RR^{(N+1)d}\times\RR^{(N+1)n}$.
\end{assumption}
In the numerical examples, some implementations use clipping, projection, or discrete optimization steps to enforce admissibility and numerical stability. These operations may introduce nonsmooth points, but the numerical results show that the algorithm still works in these situations. The above smoothness assumption should be understood as a sufficient condition for the theoretical analysis.

Under Assumptions~\ref{ass:contraction} and~\ref{ass:smooth}, we denote the Jacobian matrices of $\Xi$ at the fixed point $(\bar\mu^\gamma, \gamma)$ by
\[
	\mathcal{J}_\mu := \frac{\partial \Xi}{\partial \mu}\bigg|_{(\bar\mu^\gamma, \gamma)} \in \RR^{(N+1)d \times (N+1)d}, \qquad
	\mathcal{J}_\gamma := \frac{\partial \Xi}{\partial \gamma}\bigg|_{(\bar\mu^\gamma, \gamma)},
\]
where we identify $\mu \in \RR^{(N+1) \times d}$ with its vectorization in $\RR^{(N+1)d}$. Since $\Xi(\cdot, \gamma)$ is $\kappa$-Lipschitz (Assumption~\ref{ass:contraction}), its Jacobian satisfies $\|\mathcal{J}_\mu\|_\infty \leq \kappa < 1$, where $\|\cdot\|_\infty$ denotes the operator norm induced by the $\ell^\infty$ vector norm. In particular, the spectral radius satisfies $\rho(\mathcal{J}_\mu) \leq \|\mathcal{J}_\mu\|_\infty < 1$, so the matrix $(I - \mathcal{J}_\mu)$ is invertible. The Implicit Function Theorem then yields:
\begin{equation}\label{eq:implicit_grad}
	\frac{\partial \bar\mu^\gamma}{\partial \gamma} = \left(I - \mathcal{J}_\mu\right)^{-1} \mathcal{J}_\gamma.
\end{equation}

\paragraph{Efficient gradient computation via the adjoint method.} While~\eqref{eq:implicit_grad} provides an explicit formula for the sensitivity of the equilibrium to the parameters, computing it directly would require forming and inverting the $(N+1)d \times (N+1)d$ matrix $(I - \mathcal{J}_\mu)$, which is prohibitively expensive. In Algo.~\ref{al: parameter calibration}, what is actually needed is not the full Jacobian $\partial\bar\mu^\gamma / \partial\gamma$ itself, but rather the vector-Jacobian product $\nabla_\gamma \mathcal{L} = (\partial\bar\mu^\gamma / \partial\gamma)^\top \nabla_\mu \ell$. This product can be computed efficiently by solving the adjoint linear system $(I - \mathcal{J}_\mu^\top)\, w = \nabla_\mu \ell$ for $w$, and then evaluating $\mathcal{J}_\gamma^\top w$. Crucially, this adjoint system can be solved via a fixed-point iteration that only requires matrix-vector products with $\mathcal{J}_\mu^\top$ (computed cheaply by automatic differentiation), without ever forming or storing $\mathcal{J}_\mu$ explicitly.

The following two propositions formalize this approach: Proposition~\ref{prop:adjoint} establishes the convergence of the adjoint iteration and its equivalence with the exact implicit gradient from~\eqref{eq:implicit_grad}, while Prop.~\ref{prop:consistency} justifies the Pre-Evaluation idea of Sec.~\ref{sec: learning parameter estimates}.

\begin{proposition}[Correctness of Adjoint Gradient Computation]\label{prop:adjoint}
	Let Assumptions~\ref{ass:contraction} and~\ref{ass:smooth} hold. Fix an observed trajectory $\mu^{\mathrm{obs}} \in \RR^{(N+1) \times d}$ and define the per-sample loss $\ell : \RR^{(N+1) \times d} \to \RR$ by $\ell(\mu) := \|\mu - \mu^{\mathrm{obs}}\|_{2,N}^2$, where $\mu^{\mathrm{obs}}$ is held fixed. Define the composite loss $\mathcal{L}(\gamma) := \ell(\bar\mu^\gamma) = \|\bar\mu^\gamma - \mu^{\mathrm{obs}}\|_{2,N}^2$, where $\bar\mu^\gamma = \Phi(\gamma, \mu_0)$ is the equilibrium trajectory.

	Consider the adjoint sequence initialized at $w^{(0)} = 0 \in \RR^{(N+1)d}$ and defined by
	\begin{equation}\label{eq:adjoint_iteration}
		w^{(k+1)} = \mathcal{J}_\mu^\top\, w^{(k)} + \nabla_\mu \ell\big|_{\bar\mu^\gamma}.
	\end{equation}
	Then:
	\begin{enumerate}
		\item[(i)] The sequence $(w^{(k)})_{k \geq 0}$ converges to a unique limit $w^* \in \RR^{(N+1)d}$.
		\item[(ii)] The gradient of the composite loss with respect to $\gamma$ satisfies
		\begin{equation}\label{eq:adjoint_gradient}
			\nabla_\gamma \mathcal{L} = \mathcal{J}_\gamma^\top\, w^*,
		\end{equation}
		which equals the exact implicit gradient from~\eqref{eq:implicit_grad} composed with $\nabla_\mu \ell$.
	\end{enumerate}
\end{proposition}

\begin{proof}
	\textit{Part (i).} Let $A := \mathcal{J}_\mu \in \RR^{(N+1)d \times (N+1)d}$ and $v := \nabla_\mu \ell\big|_{\bar\mu^\gamma} \in \RR^{(N+1)d}$. The iteration~\eqref{eq:adjoint_iteration} reads $w^{(k+1)} = A^\top w^{(k)} + v$. Its unique fixed point, if it exists, is $w^* = (I - A^\top)^{-1} v$.

	By Assumption~\ref{ass:contraction}, $\|A\|_\infty \leq \kappa < 1$. Since the spectral radius satisfies $\rho(A) \leq \|A\|_\infty < 1$ for any induced matrix norm, and $A$ and $A^\top$ share the same eigenvalues, we have $\rho(A^\top) = \rho(A) < 1$. Consequently, $(A^\top)^k \to 0$ as $k \to \infty$, which implies both that $(I - A^\top)$ is invertible (with Neumann series $(I - A^\top)^{-1} = \sum_{j=0}^{\infty} (A^\top)^j$) and that the iteration converges. Indeed, $w^{(k)} - w^* = (A^\top)^k (w^{(0)} - w^*) \to 0$.

	\textit{Part (ii).} Since $\mathcal{L}(\gamma) = \ell(\bar\mu^\gamma)$ and $\bar\mu^\gamma = \Phi(\gamma, \mu_0)$, the chain rule gives
	\[
		\nabla_\gamma \mathcal{L}
		= \left(\frac{\partial \bar\mu^\gamma}{\partial \gamma}\right)^{\!\top} \nabla_\mu \ell\big|_{\bar\mu^\gamma}.
	\]
	Substituting the implicit derivative~\eqref{eq:implicit_grad} and transposing:
	\begin{align*}
		\nabla_\gamma \mathcal{L}
		&= \bigl((I - \mathcal{J}_\mu)^{-1} \mathcal{J}_\gamma\bigr)^{\!\top}\, v
		= \mathcal{J}_\gamma^\top\, (I - \mathcal{J}_\mu^\top)^{-1}\, v
		= \mathcal{J}_\gamma^\top\, w^*,
	\end{align*}
	where the last equality uses $w^* = (I - A^\top)^{-1} v$ from Part~(i). This is precisely~\eqref{eq:adjoint_gradient}.
\end{proof}

\begin{remark}[Practical approximation]\label{rem:approx}
Proposition~\ref{prop:adjoint} assumes that the Picard solver has converged to the exact fixed point $\bar\mu^\gamma$. In practice, the solver is terminated after $K$ iterations, yielding an approximate equilibrium $\mu^{(K)}$. By the contraction property (Assumption~\ref{ass:contraction}),
\(
\|\mu^{(K)} - \bar\mu^\gamma\|_\infty
\leq
\kappa^K \|\mu^{(0)} - \bar\mu^\gamma\|_\infty,
\)
so the fixed-point approximation error decays geometrically. Quantifying the corresponding finite-iteration error in the implicit gradient would require additional regularity of the Jacobians and of the adjoint solve; all theoretical statements below are stated for the exact fixed point.
\end{remark}

\paragraph{Consistency of the Pre-Evaluation idea.} As described in the Pre-Evaluation idea of Sec.~\ref{sec: learning parameter estimates}, Algo.~\ref{al: parameter calibration} evaluates the neural network $\varphi_\theta$ on the \emph{observed} trajectory $\mu^{\mathrm{obs}}$ rather than on the solver's own intermediate distributions. This decoupling greatly simplifies the computational graph, but introduces a discrepancy: the parameters fed to the solver are $\gamma^{\mathrm{d}}_k = \varphi_\theta(t_k, \mu^{\mathrm{obs}}_k)$, whereas the "ideal" coupled parameters would be $\gamma^{\mathrm{c}}_k = \varphi_\theta(t_k, \bar\mu_k)$, where $\bar\mu$ is the solver's own equilibrium. The following proposition shows that this discrepancy is controlled by $\|\bar\mu-\mu^{\mathrm{obs}}\|_\infty$, the distance between the observed trajectory and the corresponding self-consistent pair. It is a conditional stability statement for the computational shortcut.

\begin{proposition}[Pre-Evaluation Consistency]\label{prop:consistency}
	Let Assumption~\ref{ass:contraction} hold, and suppose $\Xi(\mu, \gamma)$ is $L_{\Xi, \gamma}$-Lipschitz in $\gamma$ uniformly over $\mu \in \cP([d])^{N+1}$:
	\[
		\|\Xi(\mu, \gamma_1) - \Xi(\mu, \gamma_2)\|_\infty \leq L_{\Xi, \gamma}\, \|\gamma_1 - \gamma_2\|_\infty, \qquad \forall\, \mu \in \cP([d])^{N+1}.
	\]
	Let $\varphi_\theta: [0,T] \times \cP([d]) \to \RR^n$ be $L_\varphi$-Lipschitz in its second argument, i.e., $\|\varphi_\theta(t, \mu) - \varphi_\theta(t, \mu')\|_\infty \leq L_\varphi \|\mu - \mu'\|_\infty$ for all $t$ and $\mu, \mu'$.

	Assume there exists a self-consistent pair, i.e., a trajectory $\bar\mu$ and parameters $\gamma^{\mathrm{c}}$ satisfying $\gamma^{\mathrm{c}}_k = \varphi_\theta(t_k, \bar\mu_k)$ for all $k$ and $\bar\mu = \Phi(\gamma^{\mathrm{c}}, \mu_0)$ simultaneously. Define the decoupled parameters $\gamma^{\mathrm{d}}_k := \varphi_\theta(t_k, \mu^{\mathrm{obs}}_k)$, where $\mu^{\mathrm{obs}}$ is the observed trajectory.

	Then:
	\begin{enumerate}
		\item[(i)] If $\mu^{\mathrm{obs}} = \bar\mu$, then $\gamma^{\mathrm{c}} = \gamma^{\mathrm{d}}$ and $\Phi(\gamma^{\mathrm{d}}, \mu_0) = \Phi(\gamma^{\mathrm{c}}, \mu_0)$.
		\item[(ii)] For general $\mu^{\mathrm{obs}}$, the parameter discrepancy and equilibrium error satisfy
		\begin{equation}\label{eq:consistency_bound}
			\|\gamma^{\mathrm{c}} - \gamma^{\mathrm{d}}\|_\infty \leq L_\varphi \cdot \|\bar\mu - \mu^{\mathrm{obs}}\|_\infty, \qquad
			\|\Phi(\gamma^{\mathrm{c}}, \mu_0) - \Phi(\gamma^{\mathrm{d}}, \mu_0)\|_\infty \leq \frac{L_{\Xi, \gamma}}{1-\kappa} \cdot L_\varphi \cdot \|\bar\mu - \mu^{\mathrm{obs}}\|_\infty.
		\end{equation}
	\end{enumerate}
\end{proposition}

\begin{proof}
	\textit{Part (i).} If $\mu^{\mathrm{obs}} = \bar\mu$, then $\gamma^{\mathrm{d}}_k = \varphi_\theta(t_k, \mu^{\mathrm{obs}}_k) = \varphi_\theta(t_k, \bar\mu_k) = \gamma^{\mathrm{c}}_k$ for all $k$, hence $\Phi(\gamma^{\mathrm{d}}, \mu_0) = \Phi(\gamma^{\mathrm{c}}, \mu_0)$.

	\textit{Part (ii).} By the Lipschitz property of $\varphi_\theta$, for each $k$:
	$\|\gamma^{\mathrm{c}}_k - \gamma^{\mathrm{d}}_k\|_\infty = \|\varphi_\theta(t_k, \bar\mu_k) - \varphi_\theta(t_k, \mu^{\mathrm{obs}}_k)\|_\infty \leq L_\varphi \|\bar\mu_k - \mu^{\mathrm{obs}}_k\|_\infty$.
	Taking the maximum over $k$ yields the first bound. For the second bound, let $\mu_j = \Phi(\gamma_j, \mu_0)$ for $j \in \{\mathrm{c}, \mathrm{d}\}$. By the fixed-point property and Assumption~\ref{ass:contraction}:
	$\|\mu_\mathrm{c} - \mu_\mathrm{d}\|_\infty = \|\Xi(\mu_\mathrm{c}, \gamma^{\mathrm{c}}) - \Xi(\mu_\mathrm{d}, \gamma^{\mathrm{d}})\|_\infty \leq \kappa \|\mu_\mathrm{c} - \mu_\mathrm{d}\|_\infty + L_{\Xi,\gamma} \|\gamma^{\mathrm{c}} - \gamma^{\mathrm{d}}\|_\infty$.
	Rearranging: $\|\mu_\mathrm{c} - \mu_\mathrm{d}\|_\infty \leq \frac{L_{\Xi,\gamma}}{1-\kappa} \|\gamma^{\mathrm{c}} - \gamma^{\mathrm{d}}\|_\infty$. Combining with the first bound concludes the proof.
\end{proof}

Proposition~\ref{prop:consistency} justifies the Pre-Evaluation idea of Sec.~\ref{sec: learning parameter estimates}: for a fixed network, the error introduced by decoupling is controlled by the discrepancy $\|\bar\mu-\mu^{\mathrm{obs}}\|_\infty$, and vanishes when $\bar\mu=\mu^{\mathrm{obs}}$. The result is a stability statement for the computational shortcut.

Additionally, Part~(ii) establishes that the equilibrium map $\Phi$ is Lipschitz in $\gamma$ with constant $L_{\Xi,\gamma}/(1-\kappa)$; this sensitivity bound is also used in the proof of the generalization result below.

\paragraph{Generalization to unseen trajectories.} In practice, the neural network $\varphi_\theta$ is trained on a finite dataset of observed trajectories. A natural question is whether the learned parameters generalize to new, unseen trajectories drawn from the same distribution. The following proposition provides a uniform bound on the gap between the population risk (expected loss on new data) and the empirical risk (average loss on training data). The key insight is that the implicit layer $\Phi$ inherits Lipschitz regularity from the Picard operator (via Proposition~\ref{prop:consistency}), which in turn controls the complexity of the loss landscape.

\begin{proposition}[Uniform Generalization Bound]\label{prop:generalization}
	Let Assumption~\ref{ass:contraction} hold, and assume $\Xi(\mu, \gamma)$ is $L_{\Xi,\gamma}$-Lipschitz in $\gamma$ uniformly over $\mu$, as in Proposition~\ref{prop:consistency}.
	Let the neural network parameter space $\Theta \subset \RR^p$ be compact with $\ell^2$-diameter $D := \sup_{\theta, \theta' \in \Theta} \|\theta - \theta'\|_2$.
	Assume the network map $\theta \mapsto \varphi_\theta(\cdot, \mu)$ is $L_{\varphi, \theta}$-Lipschitz in $\theta$ (in the $\ell^2$ norm) uniformly over $\mu$.
	Suppose the training data $\mathcal{S} = \{(\mu^{\mathrm{obs},i}, \mu_0^i)\}_{i=1}^M$ consists of $M$ i.i.d.\ samples from a distribution $\mathcal{P}$.
	Define the per-sample loss $\ell(\theta; \mu^{\mathrm{obs}}, \mu_0)
:= \|\mu^{\mathrm{obs}} - \Phi(\varphi_\theta(\cdot, \mu^{\mathrm{obs}}), \mu_0)\|_{2,N}^2$, where $\|\cdot\|_{2,N}$ denotes the discrete $L^2$ norm introduced above. Since for any two simplex-valued discrete trajectories $\mu,\nu$ one has $\|\mu-\nu\|_{2,N}^2 \leq 2(N+1)\Delta t$, the loss is bounded by $B := 2(N+1)\Delta t = 2T(1+1/N)$. 
	Define the empirical and population risks as $\hat{\mathcal{R}}(\theta) = \frac{1}{M}\sum_{i=1}^M \ell(\theta; \mu^{\mathrm{obs},i}, \mu_0^i)$ and $\mathcal{R}(\theta) = \EE_{\mathcal{P}}[\ell(\theta; \mu^{\mathrm{obs}}, \mu_0)]$.
	Then, for any $\delta \in (0,1)$ and $\epsilon > 0$, with probability at least $1 - \delta$ over the draw of $\mathcal{S}$, for all $\theta \in \Theta$:
	\begin{equation}\label{eq:generalization}
		\mathcal{R}(\theta) \leq \hat{\mathcal{R}}(\theta) + 2 L_\ell \epsilon + B\sqrt{\frac{\log(2 \mathcal{N}(\epsilon, \Theta) / \delta)}{2M}},
	\end{equation}
	where $\mathcal{N}(\epsilon,\Theta)$ denotes the $\epsilon$-covering number of $\Theta$ in the $\ell^2$ norm, $C_{N,d}:=\sqrt{dT(1+1/N)}$, and $L_\ell = 2 \sqrt{B} \cdot C_{N,d} \cdot \frac{L_{\Xi,\gamma}}{1-\kappa} \cdot L_{\varphi, \theta}$ is the Lipschitz constant of $\ell$ with respect to $\theta$.
\end{proposition}

For fixed $\epsilon$, $\delta$, and parameter class $\Theta$, the stochastic term in this bound scales as $M^{-1/2}$; the additional term $2L_\ell\epsilon$ is the price of passing from a finite cover to the full parameter class. For a bounded subset of $\mathbb R^p$ with diameter $D$ one may use the standard bound $\mathcal{N}(\epsilon,\Theta)\leq (1+2D/\epsilon)^p$.

\begin{proof}
	We first bound the Lipschitz constant $L_\ell$ of $\theta \mapsto \ell(\theta; \mu^{\mathrm{obs}}, \mu_0)$. For $\theta_1, \theta_2 \in \Theta$, write $\bar\mu_j := \Phi(\varphi_{\theta_j}(\cdot, \mu^{\mathrm{obs}}), \mu_0)$ for $j = 1, 2$. Then:
	\[
		|\ell(\theta_1) - \ell(\theta_2)| = \bigl|\|\mu^{\mathrm{obs}} - \bar\mu_1\|_{2,N}^2 - \|\mu^{\mathrm{obs}} - \bar\mu_2\|_{2,N}^2\bigr| \leq 2\sqrt{B}\, \|\bar\mu_1 - \bar\mu_2\|_{2,N},
	\]
	using the identity $|a^2 - b^2| \leq (a+b)|a-b|$ and $\|\mu^{\mathrm{obs}} - \bar\mu_j\|_{2,N} \leq \sqrt{B}$. 
	
	By the norm comparison $\|\cdot\|_{2,N} \leq C_{N,d}\, \|\cdot\|_\infty$ and the Lipschitz bound on $\Phi$ from Proposition~\ref{prop:consistency}:
	\[
		\|\bar\mu_1 - \bar\mu_2\|_\infty \leq \frac{L_{\Xi,\gamma}}{1-\kappa}\, \|\varphi_{\theta_1} - \varphi_{\theta_2}\|_\infty \leq \frac{L_{\Xi,\gamma}}{1-\kappa}\, L_{\varphi, \theta}\, \|\theta_1 - \theta_2\|_2.
	\]

	Combining: $L_\ell = 2\sqrt{B}\, C_{N,d}\, \frac{L_{\Xi,\gamma}}{1-\kappa}\, L_{\varphi, \theta}$.

	The generalization bound then follows from a standard covering number argument. Let $C_\epsilon \subset \Theta$ be an $\epsilon$-cover with $|C_\epsilon| = \mathcal{N}(\epsilon, \Theta)$. For any $\theta$, there exists $\theta_j \in C_\epsilon$ with $\|\theta - \theta_j\|_2 \leq \epsilon$, so $|\mathcal{R}(\theta) - \mathcal{R}(\theta_j)| \leq L_\ell \epsilon$ and $|\hat{\mathcal{R}}(\theta) - \hat{\mathcal{R}}(\theta_j)| \leq L_\ell \epsilon$. By Hoeffding's inequality applied to each fixed $\theta_j$ (with $\ell$ bounded by $B$) and a union bound over $C_\epsilon$:
	\[
		\PP\!\left(\max_{\theta_j \in C_\epsilon} \bigl(\mathcal{R}(\theta_j) - \hat{\mathcal{R}}(\theta_j)\bigr) \geq t\right) \leq \mathcal{N}(\epsilon, \Theta)\, \exp\!\left(-\frac{2 M t^2}{B^2}\right).
	\]
	Choosing $t = B\sqrt{\frac{\log(2\mathcal{N}(\epsilon,\Theta)/\delta)}{2M}}$ makes the right-hand side at most $\delta/2$, hence at most $\delta$; the extra factor $2$ in the logarithm is harmless for this one-sided bound. For any $\theta$: $\mathcal{R}(\theta) \leq \mathcal{R}(\theta_j) + L_\ell \epsilon \leq \hat{\mathcal{R}}(\theta_j) + t + L_\ell \epsilon \leq \hat{\mathcal{R}}(\theta) + t + 2L_\ell \epsilon$.
\end{proof}

In the next section, we illustrate this $M^{-1/2}$ finite-sample scaling on a simple sampled-parameter experiment.

\section{Empirical Illustration of the Generalization Bound}\label{app:generalization_algo}

To empirically illustrate the $\mathcal{O}(M^{-1/2})$ scaling predicted by Proposition~\ref{prop:generalization}, we evaluate a finite-sample proxy for the generalization gap. The core objective is to isolate the effect of the finite sample size ($M$) from optimization error; the observation-noise level is fixed throughout the experiment. 

The procedure works by evaluating, for each candidate parameter $\theta \in \Theta$, the difference between the empirical risk $\hat{\mathcal{R}}_M(\theta)$ on a finite sample of $M$ initial conditions and the population risk $\mathcal{R}(\theta)$ approximated via Monte Carlo integration over a large, separate test set. As a Monte Carlo proxy for the uniform gap in Proposition~\ref{prop:generalization}, we sample many parameters $\theta$ uniformly at random from $\Theta$ and report the maximum gap over this sampled set. This process is repeated across multiple subsamples to compute the standard error. The full methodology is detailed in Algorithm~\ref{al: generalization_experiment}.

\begin{algorithm}[H]
	\caption{Generalization Bound Empirical Verification.}
	\label{al: generalization_experiment}
	\begin{algorithmic}[1]
		\Require True parameter $\theta^*$; training set sizes $\mathbb{M} = \{M_1, M_2, \ldots\}$; subsamples $S \in \NN$; test size $N_{\text{test}} \in \NN$; number of random parameters $P \in \NN$; parameter space $\Theta$
		\State Generate $N_{\text{test}}$ test initial conditions $\mu_0^{\text{test},i} \sim \text{Dir}(\mathbf{1})$ \Comment{{\footnotesize Data Generation}}
        \State Simulate ground-truth test trajectories: $\bar{\mu}^{\text{test},i} \gets \Phi(\theta^*, \mu_0^{\text{test},i})$ for $i \in \{1, \dots, N_{\text{test}}\}$
        \State Generate a large pool of $N_{\text{train}}$ training initial conditions and corresponding true trajectories
        \State Sample $P$ parameters $\{\theta_p\}_{p=1}^P$ uniformly at random from $\Theta$ \Comment{{\footnotesize No training; fixed $\theta$}}
        \For{$M \in \mathbb{M}$}
            \For{$s \in \{1, \dots, S\}$}
                \State Sample $M$ trajectories $\{(\mu_0^{j}, \bar{\mu}^{j})\}_{j=1}^M$ from the training pool \Comment{{\footnotesize Draw a finite dataset}}
                \For{$p \in \{1, \dots, P\}$}
                    \State $\hat{\mathcal{R}}_M(\theta_p) \gets \frac{1}{M} \sum_{j=1}^M \|\bar{\mu}^{j} - \Phi(\theta_p, \mu_0^{j})\|_{2,N}^2$ \Comment{{\footnotesize Empirical Risk}}
                    \State $\mathcal{R}(\theta_p) \gets \frac{1}{N_{\text{test}}} \sum_{i=1}^{N_{\text{test}}} \|\bar{\mu}^{\text{test},i} - \Phi(\theta_p, \mu_0^{\text{test},i})\|_{2,N}^2$ \Comment{{\footnotesize Population Risk}}
                \EndFor
                \State $\text{Gap}_{M, s} \gets \sup_{p} \left| \mathcal{R}(\theta_p) - \hat{\mathcal{R}}_M(\theta_p) \right|$ \Comment{{\footnotesize Supremum over sampled $\theta$}}
            \EndFor
            \State $\mu_{\text{Gap}}(M) \gets \frac{1}{S}\sum_{s=1}^S \text{Gap}_{M, s}$ \Comment{{\footnotesize Mean Gap}}
            \State $\sigma_{\text{Gap}}(M) \gets \text{std}(\text{Gap}_{M}) / \sqrt{S}$ \Comment{{\footnotesize Standard Error}}
        \EndFor
        \State \Return{$\{\mu_{\text{Gap}}(M), \sigma_{\text{Gap}}(M)\}_{M \in \mathbb{M}}$ for log-log regression analysis}
	\end{algorithmic}
\end{algorithm}

\paragraph{Experimental setup.}
We apply Algorithm~\ref{al: generalization_experiment} to the Linear-Quadratic MFG model from Sec.~\ref{sec: experiments} with a $2$-dimensional linear parameterization $\gamma_\theta(t) = \theta_0 + \theta_1 t$, true parameter $\theta^* = (1.0, 2.0)$, and $d = 3$ discrete states. To study a sampled proxy for the uniform bound of Proposition~\ref{prop:generalization} without conflating optimization error, we bypass any training step: we evaluate the generalization gap $|\mathcal{R}(\theta) - \hat{\mathcal{R}}_M(\theta)|$ for $P=500$ parameters $\theta$ drawn uniformly at random from $\Theta = [0.1, 5.0]^2$, and report the maximum over these sampled parameters. The hyperparameters are summarized in Table~\ref{tab:gen_bound_params}.

\begin{table}[h]
	\centering
	\caption{Hyperparameters for the generalization bound verification experiment.}
	\label{tab:gen_bound_params}
	\small
	\begin{tabular}{ll}
		\toprule
		\textbf{Parameter} & \textbf{Value} \\
		\midrule
		Discrete states ($d$) & 3 \\
		Time horizon ($T$) & 1.0 \\
		Number of points in time ($N+1$) & 20 \\
		True parameter $\theta^*$ & $(1.0, 2.0)$ \\
		Parameter domain $\Theta$ & $[0.1, 5.0]^2$ \\
		Parameterization & Linear: $\gamma(t) = \theta_0 + \theta_1 t$ \\
		Noise level (Dirichlet $\kappa^{-1}$) & $10^{-4}$ \\
		Training pool size $N_{\text{train}}$ & 5000 \\
		Test set size $N_{\text{test}}$ & 5000 \\
		Sample sizes $\mathbb{M}$ & $\{1, 5, 10, 20, 50, 100, 200\}$ \\
		Number of random $\theta$'s & 500 \\
		Subsamples $S$ & 100 \\
		\bottomrule
	\end{tabular}
\end{table}

\begin{figure}[h]
	\centering
	\includegraphics[width=0.75\textwidth]{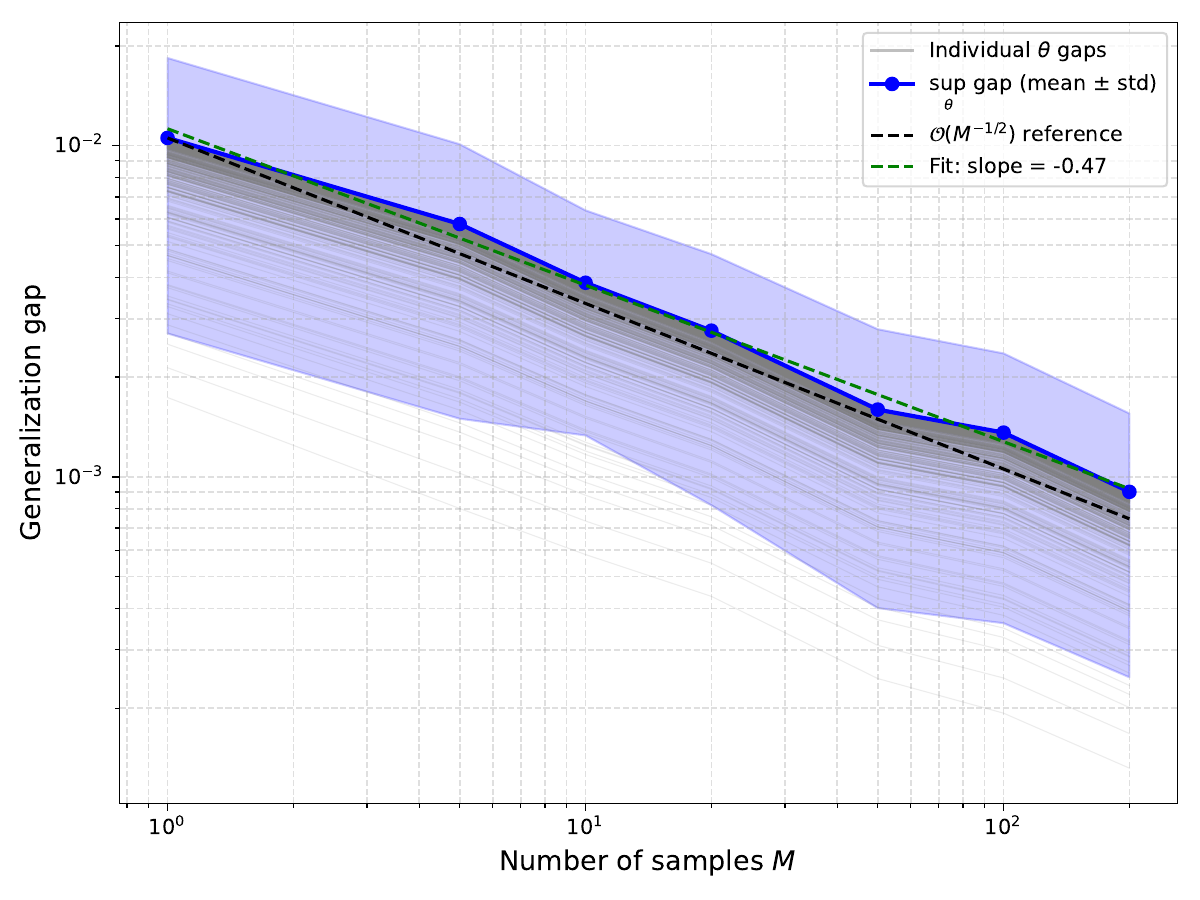}
	\caption{Empirical illustration of the generalization-bound scaling. Gray curves show the generalization gap $|\mathcal{R}(\theta) - \hat{\mathcal{R}}_M(\theta)|$ averaged over $S = 100$ subsamples for each of the $500$ randomly sampled parameters $\theta \in \Theta$. The blue curve shows the maximum over the sampled parameters of these gaps, with shaded $\pm 1$ standard deviation. The dashed black line is the theoretical $\mathcal{O}(M^{-1/2})$ reference. The fitted log-log slope of $-0.50$ is consistent with the predicted rate.}
	\label{fig:gen_bound}
\end{figure}

\paragraph{Discussion.}
Fig.~\ref{fig:gen_bound} shows that the empirical supremum over the sampled parameters decays at a rate of roughly $M^{-0.50}$ in log-log scale, in agreement with the $\mathcal{O}(M^{-1/2})$ rate predicted by Proposition~\ref{prop:generalization}. The gap ranges from approximately $10^{-2}$ at $M = 1$ to $9 \times 10^{-4}$ at $M = 200$. The individual $\theta$-level gaps (gray curves) all exhibit the same scaling behavior, and the supremum envelope (blue curve) closely tracks them, indicating that no single parameter dominates the worst case disproportionately. We note that the error bars (standard deviation across subsamples) shrink proportionally with the mean, confirming the stability of the Monte Carlo estimation procedure.

\begin{remark}[Gap between the theoretical constant and empirical observations]\label{rem:bound_tightness}
	While the empirical generalization gap confirms the theoretically predicted $\mathcal{O}(M^{-1/2})$ convergence rate, the absolute value of the theoretical upper bound (Proposition~\ref{prop:generalization}) is typically several orders of magnitude larger than the observed gap. This discrepancy arises from several compounding sources of conservatism:
	\begin{enumerate}[leftmargin=*]
		\item \textbf{Worst-case loss bound $B$.} The bound assumes $\ell(\theta; \mu^{\mathrm{obs}}) \leq B$ for all $\theta \in \Theta$ and all possible observations, where $B = 2T(1 + 1/N)$ bounds the squared $L^2$ distance between any two distributions on the simplex. In practice, the actual losses are concentrated in a narrow range far below this worst-case envelope.
		\item \textbf{Parameter space diameter $D$.} The covering number $\mathcal{N}(\epsilon, \Theta) \leq (1+2D/\epsilon)^p$ grows polynomially with the diameter $D = \mathrm{diam}(\Theta)$ of the full parameter space. In practice, the loss function varies appreciably only over a much smaller effective region around the true parameter $\theta^*$.
		\item \textbf{Distribution-free nature.} The Hoeffding--union bound argument holds for \emph{any} data-generating distribution over initial conditions $\mu_0$. It cannot exploit the specific concentration properties of the Dirichlet distribution or the smoothness of the MFG equilibrium map $\Phi$, both of which reduce the variance of finite-sample averages in our setting.
	\end{enumerate}
	We emphasize that the primary contribution of Proposition~\ref{prop:generalization} is the \emph{rate} $\mathcal{O}(M^{-1/2})$ and its structural dependence on the Lipschitz constants of the MFG equilibrium map, rather than the absolute value of the constant.
\end{remark}

\section{Experimental Details}\label{appendix: experimental_details}

\subsection{Synthetic Data Generation} In the LQ and cybersecurity setting, we generate synthetic data by solving the forward-backward MFG system with true parameters sampled from specific functional forms. Specifically, we consider three forms for all true parameters $\gamma_i \in [\gamma_{i, \min}, \gamma_{i, \max}]$:

\begin{enumerate}
    \item \textbf{Constant}: $\gamma_i(t, \mu) = (\gamma_{i, \min} + \gamma_{i, \max}) / 2$.
    \item \textbf{Time-dependent}: $\gamma_i(t, \mu) = \gamma_{i, \min} + \frac{4(\gamma_{i, \max} - \gamma_{i, \min})}{T^2} t(T-t)$.
    \item \textbf{Mean field dependent}: $\gamma_i(t, \mu) = \gamma_{i, \min} + (\gamma_{i, \max} - \gamma_{i, \min})\mu(0)$, scaling linearly with the mass in the first state. 
\end{enumerate}
All hyperparameters and instantiations for $\gamma_{i, \min}$ and  $\gamma_{i, \max}$ are reported in App.~\ref{appendix: experimental_parameters}.
\paragraph{Noise Robustness.} We evaluate the robustness of our approach by introducing Dirichlet noise to the observed mean field distributions. Specifically, for a given noise level $\sigma > 0$, the noisy observation $\tilde{\mu}_t$ at each time step is sampled from a Dirichlet distribution:
\begin{equation}
    \tilde{\mu}_t \sim \text{Dir}(p_t), \quad \text{where} \quad p_t = \frac{1}{\sigma} (\mu_t + \epsilon \bm{1})
    \label{eq: noise}
\end{equation}
Here, $\mu_t$ is the true mean field distribution, $\epsilon = 10^{-6}$ is a small constant to ensure strict positivity, and $\sigma$ acts as the variance-controlling noise level (with smaller $\sigma$ yielding samples tightly concentrated around the true mean). We apply this analysis to the synthetic linear-quadratic and cybersecurity settings, where we have full control over the noise.

\subsection{Linear-Quadratic}
In the LQ MFG, the players are free to choose their transition rates $a_{xy}$ between the states without being influenced by predefined dynamics. In the absence of predefined transitions, the game is flexible and can be scaled to an arbitrary number of states $d$. The players' transition rates between those states are subject to the following running and terminal costs,
\begin{equation}
	f^{\gamma}(x, a, \mu):= b \sum_{y\neq x} (a_{xy} - 2)^2 + \mu(x); \quad g(x, \mu) = \frac{1}{2} \mu(x)^2.
\end{equation}
Here, the parameter of interest is $\gamma \equiv b \in \RR$, which determines the cost associated with the state transitions, while the addition of $\mu(x)$ penalizes the players for being in crowded states. As a consequence, large values for $b$ disincentivize the players to deviate from $a_{xy} = 2$, which is expected to result in a smoothly evolving mean field distribution $\mu$. In contrast, small values of $b$ increase the dominance of $\mu(x)$ in the running cost, resulting in a fast redistribution of the population over the states. Due to the linear-quadratic structure of the game, an analytical solution for the optimal transition rate $a^*$ can be obtained through a first-order minimization of the Hamiltonian \cite{Bensoussan2016}:
\begin{equation}
\begin{aligned}
H_b(t,x,\mu,u)
&= \mu(x)+\inf_{a\in [1,3]^{d-1}}
\left\{\sum_{y\neq x} b(a_{xy}-2)^2+a_{xy}(u(y)-u(x))\right\},\\
a_{xy}^*(x,u)
&= \Pi_{[1,3]}\!\left(2-\frac{u(y)-u(x)}{2b}\right).
\end{aligned}
\label{eq: hamiltonian and optimal action}
\end{equation}
where $\Pi_{[1,3]}(z)=\min\{3,\max\{1,z\}\}$.

\subsection{Cybersecurity}
The cybersecurity game considers a botnet-like population of computers that may become victims of a cyberattack. Each computer can be defended against such attacks by updating its protection level. In this case, the computer moves from the Unprotected ($U$) to the Defended ($D$) state. An update takes the time $\lambda >0$ to complete. In addition, each computer is either Susceptible ($S$) to the attack, or already Infected ($I$) with the malware. Therefore the $d=4$ states are $UI,US,DI,DS$; in the transition matrix below we order rows and columns as $(DI,DS,UI,US)$. The game admits two actions $\cA = \{0, 1\}$. If $a= 0$ then the security level of the computer remains unchanged. For $a=1$, the security level of the computer is updated, effectively reversing its protection level.

Depending on their protection level, infected computers can recover to the susceptible state at a recovery rate $q^D_{rec}$ or $q^U_{rec}$. There are two ways for a susceptible computer to become infected: either (1) directly, via an attack from the hacker, or (2) indirectly, via an infected computer. Direct attacks occur at the rates $v_Hq^D_{inf}$ and $v_Hq^U_{inf}$, for the defended and unprotected computers, respectively. The rate of infection through indirect attacks depends on the distribution of states $\mu$ among the vulnerable computers. A computer in $DI$ infects a computer in $DS$ with rate $\beta_{DD}\mu(DI)$ and a computer in $US$ with rate $\beta_{DU}\mu(DI)$. Similarly, computer in $UI$ infects a computer in $DS$ with rate $\beta_{UD}\mu(UI)$ and a computer in $US$ with rate $\beta_{UU}\mu(UI)$. The resulting transition matrix is presented in~\eqref{eq: cyber_Q}. 
{\small
	\begin{equation}
		Q^{\gamma}(\cdot, \cdot, \mu, v_H, a) =
		\begin{bNiceArray}{cccc}[
				first-row,
				first-col,
				code-for-first-row=\scriptstyle,
				code-for-first-col=\scriptstyle]
			& DI & DS & UI & US\\
			DI & \dots  & q_{rec}^D  & \bm{1}_{\{1\}}(a) \lambda  & 0 \\
			DS & \termDS  & \dots  & 0  & \bm{1}_{\{1\}}(a) \lambda \\
			UI & \bm{1}_{\{1\}}(a) \lambda & 0 & \dots & q_{rec}^U \\
			US & 0 & \bm{1}_{\{1\}}(a) \lambda & \termUS & \dots\\
		\end{bNiceArray}
        	\label{eq: cyber_Q}
	\end{equation}
    }
Each computer owner aims to minimize the running cost that combines $k_D$, the costs of the defense system per unit of time, and the $k_I$ per unit of time loss associated with an infected device:
\begin{equation}
	f^{\gamma} (t, x, a, \mu)= k_D \bm{1}_{\{D\}}(x) + k_I \bm{1}_{\{I\}}(x).
	\label{eq: cost_cyber}
\end{equation}

We model the unknown parameters $\gamma \equiv (q_{rec}^D, q_{rec}^U, q_{inf}^D,\allowbreak q_{inf}^U, \beta_{DD}, \beta_{DU}, \beta_{UU}, \beta_{UD})$. The parameters $\lambda$ and $v_H $ are not learned, as they describe the unobservable dynamics of the representative player ($\lambda$) or are inseparable from $q_{inf}^U$ and $q_{inf}^D$ ($v_H$). 

\subsection{Susceptible-Infected-Recovered}
In real-world infectious disease dynamics, state transitions are not just influenced by parameters that are associated with the disease. Instead, the infectious disease transmission also depends on the intensity of contact between individuals. This is the player's control: each player has the option to minimize their contact with others, for example, by minimizing their movements or wearing protective equipment. Therefore, we define the transition rate of a susceptible representative agent to the infected state as $\beta \alpha_t \mu_t(I)$. Here, $\beta$ denotes the base infection rate, $\alpha_t$ is the socialization-level control of the susceptible individual, and $\mu_t(x)$ is the fraction of the population that is in state $x$ at time $t$. We implicitly assume that the infected individuals will follow any distancing guidelines that are suggested to them.  We keep the state transitions for infected and recovered players (i.e., the recovery and waning of immunity rates, respectively) fixed, making them dependent on disease properties only. Because we are modelling the viral dynamics per season, we assume that the players retain full immunity after having recovered. Summarizing these dynamics, we construct the transition matrix as
\begin{equation}
	Q^{\gamma}(a_t, \mu_t) = \begin{bmatrix}
		-\beta a_t\mu_t(I) & \beta a_t \mu_t(I) & 0       \\
		0                  & -\gamma          & \gamma  \\
		0             & 0                & 0
	\end{bmatrix},
\end{equation}
where the row-wise sum of the elements is 0 by design to ensure that $Q^{\gamma} \in \MM^d$. The agents face a cost for being ill $c_I$ and for having a low contact factor $c_\alpha$ in the Susceptible state
\begin{equation}
   f^{\gamma}(t, x, a, \mu, c_\alpha, c_I)=  \frac{c_\alpha}{2}(1-\alpha_t(x))^2 \bm{1}_{\{S\}} + c_{I} \bm{1}_{\{I\}}.
\end{equation}
Both cost weights are learned by the neural network. When these weights are small, the strategic component has little influence on the fitted dynamics; however, in the degenerate case $c_\alpha=c_I=0$, the optimal control is not uniquely selected, so any mean-field calibration baseline requires an explicit selection rule, such as fixing $\alpha_t\equiv 1$.

\subsection{Urban Mobility}
To further explore the relevance of the MFG mechanics, we consider a bike-sharing network where commuters strategically choose their destinations based on station congestion. We compare two models calibrated on the same data: a \emph{mean field dynamics} model (baseline), which describes population flow via exogenous transition rates, and a \emph{mean field game} model, which additionally captures strategic congestion avoidance through an optimal control formulation. The MFG model contains the forward dynamics baseline as a special case: if the congestion/general-cost term is zero, $w_g=0$, and the zero HJB solution is selected, then $a_{xy}^*=q_{xy}^{\mathrm{base}}$ and the FPK equation reduces to the baseline.

\paragraph{Mean field dynamics model (baseline).} In this model, agents follow exogenous, time-dependent transition rates without any strategic behavior. The population distribution evolves according to the Kolmogorov forward equation:
\begin{equation}
	\dot{\mu}_i = \sum_{j \neq i} \mu_j \, q_{ji}(\gamma) - \mu_i \sum_{j \neq i} q_{ij}(\gamma),
\end{equation}
where the rates $q_{ij}(t) = \mathrm{softplus}(\gamma_{ij}(t))$ are parameterized by a neural network $\varphi_\theta : [0,T] \to \RR^{d(d-1)}$. Depending on the test regime (see below), $\varphi_\theta$ is either a constant vector, a single affine layer, or a multi-layer perceptron. The network is trained by minimizing the $L_2$ discrepancy between the predicted trajectory $\mu^{\gamma_\theta}$ and the observed distribution $\hat{\mu}$, analogously to the procedure described in Sec.~\ref{sec: method}.

\paragraph{Mean field game model.} In the MFG formulation, each commuter optimizes their destination choice to minimize a cost that penalizes both deviating from the natural demand pattern and arriving at congested stations. The running cost is:
\begin{equation}
	f(x, a, \mu) = b \sum_{y \neq x} (a_{xy} - q_{xy}^{\mathrm{base}})^2 + c \sum_{y \neq x} \mu(y) \, a_{xy},
\end{equation}
where $q_{xy}^{\mathrm{base}}(t) = \mathrm{softplus}(\gamma_{xy}(t))$ are the same base rates as in the baseline model, $b > 0$ controls the cost of deviating from the natural demand, and $c \geq 0$ weights the destination congestion penalty. The terminal cost is $g(x) = w_g(x)$, where $w_g \in \RR^d$ is a learnable vector initialized to zero. The Hamiltonian minimization yields the optimal transition rate:
\begin{equation}
	a_{xy}^* = \max\!\left(0, \; q_{xy}^{\mathrm{base}} - \frac{c\,\mu(y) + u(y) - u(x)}{2b}\right),
\end{equation}
where $u$ is the value function obtained from the backward HJB equation with terminal condition $u(T,x) = w_g(x)$. 

\paragraph{Physical interpretation.} The baseline model captures the dominant origin-destination demand patterns, morning commute from residential to commercial areas, and the evening reverse flow, through the learned rates $q_{ij}^{\mathrm{base}}(t)$. The MFG model adds a strategic correction: when a destination station $y$ is crowded ($\mu(y)$ is large), the effective rate $a_{xy}^*$ is reduced by $c\mu(y)/(2b)$. In Tests 1--2, the implementation additionally clips these rates above at $100$ for numerical stability; Test~3 uses only the nonnegativity projection shown in the formula. Furthermore, the backward equation provides \emph{forward-looking} congestion avoidance through the value function $u$, allowing agents to anticipate future bottlenecks rather than merely reacting to current occupancy.

\paragraph{Experimental setup.} For the MFG model, the cost parameters are initialized to $b_0 = 1.0$ and $c_0 = 0.5$, and $w_g = 0$. All parameters ($\theta$, $b$, $c$, $w_g$) are jointly optimized using the Adam optimizer \cite{Kingma2014}. The MFG equilibrium is computed at each training step using the damped Picard iteration described in App.~\ref{appendix: picard iteration}. We report the mean field $L_2$ and relative $L_2$ prediction errors on both training and held-out test trajectories, averaged over 5 seeds.

\section{Mean Field Control: Formulation and Optimality Conditions}\label{appendix: mfc_derivation}
While MFGs describe the Nash equilibrium of non-cooperative agents, Mean Field Control (MFC) describes the social optimum where a central planner minimizes the aggregate cost of the population. In this section, we derive the optimality conditions for MFC in a finite-state space setting using the calculus of variations, following the approach of \citet{Bensoussan2013}.

\subsection{Problem Definition}
Consider a population of agents distributed over a finite-state space $[d]$. The central planner aims to choose a feedback control policy $\alpha_t = (\alpha_{t, xy})_{x,y \in [d]}$ to minimize the total cost functional:
\begin{equation}
    J(\alpha) = \int_0^T \sum_{x \in [d]} f_{\params}(t, x, \alpha_t(x), \mu_t) \mu_t(x) dt + \sum_{x \in [d]} g_{\params}(x, \mu_T) \mu_T(x),
    \label{eq:mfc_cost}
\end{equation}
subject to the Fokker-Planck equation (forward ODE) for the mean field $\mu_t \in \mathcal{P}([d])$:
\begin{equation}
    \partial_t \mu_t(x) = \sum_{y \neq x} \left( \mu_t(y) \alpha_{t, yx} - \mu_t(x) \alpha_{t, xy} \right), \quad \mu_0(x) \text{ given.}
    \label{eq:mfc_fpk}
\end{equation}
For simplicity, this appendix restricts to rate-control dynamics of the form~\eqref{eq:mfc_fpk}, where the control variables $\alpha_{t,xy}$ are the transition rates themselves. This covers the LQ MFC example, where $\gamma_\theta(t,\mu)$ enters the running cost, and hence the optimal rates, but not an additional primitive generator $Q^\gamma(t,\alpha,\mu)$. Extending the derivation to a general parameter-dependent generator would require adding the corresponding derivatives of the dynamics to the adjoint equation. For clarity, the variational derivation below is written in the interior case. When transition rates are constrained to be nonnegative or bounded, the pointwise minimization condition should be read as the corresponding variational inequality condition over the admissible rate set.

\subsection{The Optimality System}
The necessary conditions for the optimal control $\alpha^*$ are characterized by a forward-backward system of ODEs. Specifically, there exists an adjoint state $u_t(x)$ such that:
\begin{enumerate}
    \item \textbf{Forward Equation}: The mean field $\mu_t$ satisfies the FPK equation \eqref{eq:mfc_fpk} with the optimal control $\alpha^*$.
    \item \textbf{Backward Adjoint Equation}: The adjoint state $u_t$ satisfies:
    \begin{equation}
        -\partial_t u_t(x) = \sum_{y \neq x} \alpha^*_{t, xy} (u_t(y) - u_t(x)) + f_{\params}(t, x, \alpha^*_t, \mu_t) + \sum_{y \in [d]} \mu_t(y) \left[ \frac{\partial f_{\params}}{\partial \mu(x)} + \frac{\partial f_{\params}}{\partial \params} \frac{\partial \params}{\partial \mu(x)} \right],
    \end{equation}
    where the partial derivatives inside the sum are evaluated at $(t, y, \alpha^*_t(y), \mu_t)$. The terminal condition is given by:
    \begin{equation}
        u_T(x) = g_{\params}(x, \mu_T) + \sum_{y \in [d]} \mu_T(y) \left[ \frac{\partial g_{\params}}{\partial \mu(x)} + \frac{\partial g_{\params}}{\partial \params} \frac{\partial \params}{\partial \mu(x)} \right].
    \end{equation}
    \item \textbf{Optimality Condition}: At each time $t$, the control $\alpha^*_t(x)$ minimizes the Hamiltonian:
    \begin{equation}
        \alpha^*_t(x) \in \arg\min_{a} \left[ f_{\params}(t, x, a, \mu_t) + \sum_{y \neq x} a_{xy} (u_t(y) - u_t(x)) \right].
    \end{equation}
\end{enumerate}

\subsection{Derivation via Calculus of Variations}
We use a perturbation argument to derive the adjoint equation. Let $\alpha$ be the optimal control and let $\beta$ be an arbitrary perturbation. We define $\alpha^\epsilon = \alpha + \epsilon \beta$. Let $\mu^\epsilon$ be the mean field corresponding to $\alpha^\epsilon$, and write $\mu^\epsilon = \mu + \epsilon \nu + \mathcal{O}(\epsilon^2)$, where $\nu$ satisfies the linearized dynamics:
\begin{equation}
    \partial_t \nu_t(x) = \sum_{y \neq x} (\nu_t(y) \alpha_{yx} + \mu_t(y) \beta_{yx} - \nu_t(x) \alpha_{xy} - \mu_t(x) \beta_{xy}), \quad \nu_0(x) = 0.
\end{equation}
The first-order variation of the cost functional \eqref{eq:mfc_cost} is given by:
\begin{equation}
\resizebox{\hsize}{!}{$
    \delta J = \int_0^T \sum_{x \in [d]} \left[ \frac{\partial f_{\params}}{\partial \alpha} \beta_t(x) \mu_t(x) + \left( f_{\params} + \sum_{y \in [d]} \mu_t(y) \frac{d f_{\params}(t, y, \alpha_t(y), \mu_t)}{d \mu_t(x)} \right) \nu_t(x) \right] dt + \sum_{x \in [d]} \left[ g_{\params} + \sum_{y \in [d]} \mu_T(y) \frac{d g_{\params}(y, \mu_T)}{d \mu_T(x)} \right] \nu_T(x),$}
\end{equation}
where $\frac{d f_{\params}}{d \mu(x)}$ denotes the total derivative accounting for the implicit dependence of $\params$ on $\mu$. 
To eliminate the dependence on $\nu$, we introduce the adjoint state $u_t(x)$ and observe that for any $u_t$:
\begin{equation}
    \int_0^T \sum_x u_t(x) \left( \partial_t \nu_t(x) - \sum_{y \neq x} (\nu_t(y) \alpha_{yx} + \mu_t(y) \beta_{yx} - \nu_t(x) \alpha_{xy} - \mu_t(x) \beta_{xy}) \right) dt = 0.
\end{equation}
Integrating by parts $\int u \partial_t \nu dt = [u\nu]_0^T - \int \partial_t u \nu dt$ and using $\nu_0=0$:
\begin{equation}
    \sum_x u_T \nu_T - \int_0^T \sum_x \left[ \partial_t u + \sum_{y \neq x} \alpha_{xy}(u(y) - u(x)) \right] \nu dt - \int_0^T \sum_x \mu(x) \sum_{y \neq x} \beta_{xy}(u(y) - u(x)) dt = 0.
\end{equation}
Substituting this expression into $\delta J$, we require the terms involving $\nu$ and $\nu_T$ to cancel out. This leads directly to the backward adjoint ODE and the terminal condition:
\begin{equation}
    -\partial_t u_t(x) = \sum_{y \neq x} \alpha_{xy}(u(y) - u(x)) + f_{\params} + \sum_y \mu(y) \frac{d f_{\params}}{d \mu(x)}, \quad u_T(x) = g_{\params} + \sum_y \mu_T(y) \frac{d g_{\params}}{d \mu(x)}.
\end{equation}
Finally, the remaining terms in $\delta J$ involve only the arbitrary perturbation $\beta$:
\begin{equation}
    \delta J = \int_0^T \sum_x \mu_t(x) \left[ \frac{\partial f_{\params}}{\partial \alpha} \beta + \sum_{y \neq x} \beta_{xy} (u(y) - u(x)) \right] dt.
\end{equation}
Requiring $\delta J = 0$ for all $\beta$ yields the optimality condition for the control policy.

\subsection{Numerical Approximation via Fixed-Point Iteration}
To numerically compute the social optimum characterized by the forward-backward optimality system, we employ a Picard iteration scheme analogous to Algo.~\ref{al: picard}. However, solving the Mean Field Control system presents a unique mathematical and computational challenge when the true parameter $\gamma$ depends explicitly on the population distribution, i.e., $\gamma_t = \gamma(t, \mu_t)$. 

In this scenario, the backward adjoint equation strictly requires the evaluation of the chain rule term $\frac{\partial \gamma}{\partial \mu}$ to internalize the planner's impact on the environment. If we were to embed the nonlinear function $\gamma(\cdot)$ directly into the backward ODE solver, the implicit numerical differentiation of the fixed-point operator would become highly recursive and computationally prohibitive. 

For a practical implementation, we linearly decouple the parameter and Jacobian trajectories during the iteration. At each Picard step $k$, given the previous distribution trajectory $\mu^{(k)}$, we pre-compute both the parameter values and their Jacobian across the entire time horizon:
\begin{equation}
    \gamma^{(k)}_t = \gamma(t, \mu^{(k)}_t), \quad \quad D^{(k)}_t = \frac{\partial \gamma}{\partial \mu}(t, \mu^{(k)}_t).
\end{equation}
These pre-computed trajectories, $(\gamma^{(k)}, D^{(k)})$, are passed as fixed, exogenous inputs into the ODE solver. The backward adjoint equation is then integrated using the fixed trajectory $D^{(k)}_t$ to evaluate the extra MFC term $\frac{\partial f}{\partial \gamma} D^{(k)}_t$. 

At a converged fixed point with $\mu^{(k+1)}=\mu^{(k)}$, the precomputed $\gamma^{(k)}$ and $D^{(k)}$ coincide with the quantities evaluated along that fixed point, so the iteration matches the discretized optimality system above for the restricted dynamics.

Crucially, this decoupled architecture elegantly supports the inverse problem of neural parameter calibration. In the inverse setting, we aim to learn the neural network parameters $\theta$ that define $\gamma_\theta(t, \mu)$. Mirroring the approach taken for MFGs, we employ a teacher-forcing methodology where we replace the iterative trajectory $\mu^{(k)}_t$ entirely with the \textit{observed} empirical data $\mu^{obs}_t$. 

During training, we evaluate the model predictions and its Jacobians explicitly on the observed trajectory:
\begin{equation}
    \gamma^{obs}_t = \gamma_\theta(t, \mu^{obs}_t), \quad \quad D^{obs}_t = \frac{\partial \gamma_\theta}{\partial \mu}(t, \mu^{obs}_t).
\end{equation}
These pre-computed trajectories $(\gamma^{obs}, D^{obs})$ are then passed to the ODE solver to generate the predicted mean field $\mu^{pred}$, which is scored against $\mu^{obs}$ via the $L_2$ loss. This approach decouples the neural network training loop from nested, recursive ODE backpropagation and provides a scalable teacher-forced training objective.

\subsection{Numerical Experiments \& Results}
To demonstrate the versatility of our implicit differentiation framework, we extend the LQ model from finding a non-cooperative Nash equilibrium (MFG) to identifying the social optimum formulated by a central planner (Mean Field Control). 

While the Nash equilibrium is characterized by selfish agents who only consider their immediate running costs, the social planner minimizes the aggregate cost of the entire population $\int_0^T \sum_x \mu_t(x) f(x, a, \mu) dt$. The fundamental mathematical distinction lies in the backward HJB equation: the social planner's value function must account for the exact derivative of the Lagrangian with respect to the population density. 

To highlight the physical divergence between the Nash equilibrium and the social optimum, we modify the LQ cost structure in two ways:
\begin{enumerate}
    \item \textbf{Reduced Base Transition Rate:} We lower the base transition rate to $a_{xy} \approx 0.3$. This significantly reduces the natural "free" diffusion of the agents, meaning they remain stationary unless explicitly pushed by the congestion gradients.
    \item \textbf{Squared Congestion Cost:} We adopt a squared congestion penalty $c \mu(x)^2$, with $c=2$. By taking the derivative with respect to $\mu(x)$, the MFC social planner incurs an additional penalty of $2 c \mu(x)$ in the HJB equation. This anticipatory "systemic risk" term strongly forces the planner to evacuate crowded states faster than the selfish MFG agents.
\end{enumerate}

Due to the hypersensitivity introduced by the reduced base transition rate, the forward-backward Picard solver naturally oscillates or diverges. To restore mathematical stability, we applied a Picard damping factor of $\lambda=0.5$ (as detailed in App.~\ref{appendix: picard iteration}).

\begin{figure}[h]
    \centering
    \begin{subfigure}[b]{0.8\textwidth}
        \includegraphics[width=\textwidth]{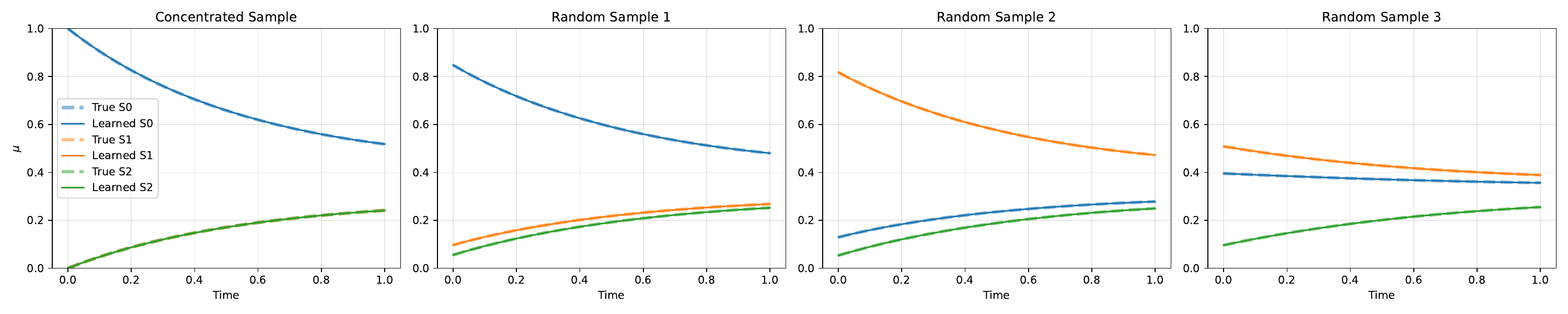}
        \caption{Nash Equilibrium (MFG)}
    \end{subfigure}

    \begin{subfigure}[b]{0.8\textwidth}
        \includegraphics[width=\textwidth]{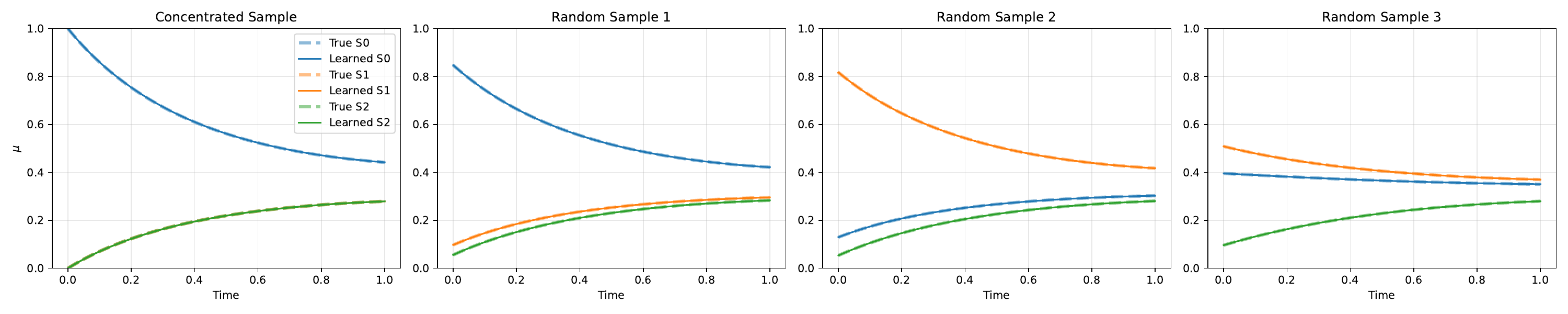}
        \caption{Social Optimum (MFC)}
    \end{subfigure}
    \caption{Comparison of the mean field evolution $\mu_t$ between the non-cooperative Nash equilibrium (top) and the social optimum (bottom). The social planner preemptively diffuses the population out of the concentrated initial state faster to avoid the squared congestion penalty.}
    \label{fig:mfg_vs_mfc_evolution}
\end{figure}

\section{Mean Field Game Calibration of the Cybersecurity Model}\label{appendix: extra_analyzes}
\begin{figure}[H]
    \centering
    \begin{subfigure}[b]{0.24\textwidth}
    \includegraphics[width=\textwidth]{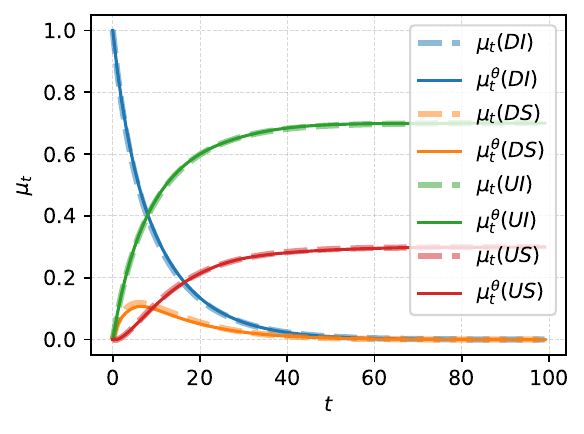}
    \end{subfigure}
    \hfill
    \begin{subfigure}[b]{0.24\textwidth}
        \includegraphics[width=\textwidth]{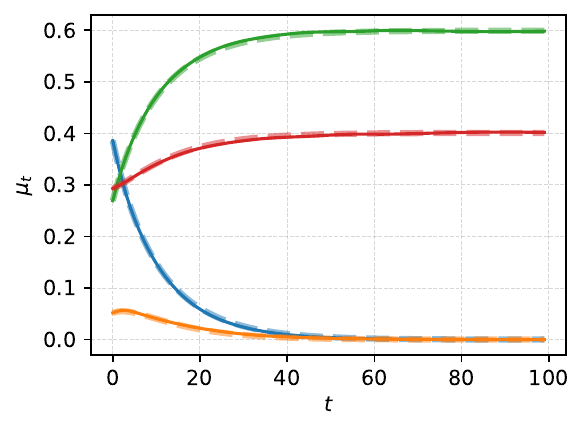}
    \end{subfigure}
    \hfill
        \begin{subfigure}[b]{0.24\textwidth}
    \includegraphics[width=\textwidth]{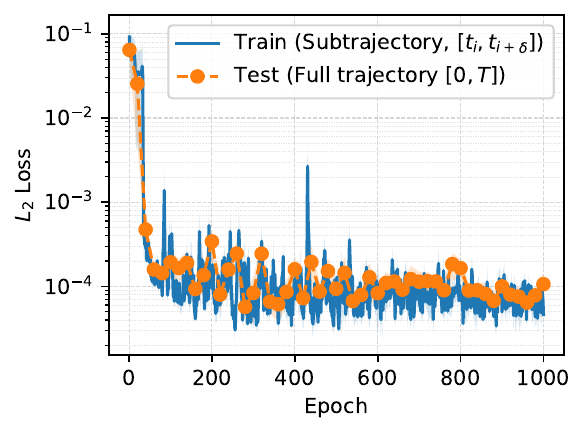}
    \end{subfigure}
        \begin{subfigure}[b]{0.24\textwidth}
    \includegraphics[width=\textwidth]{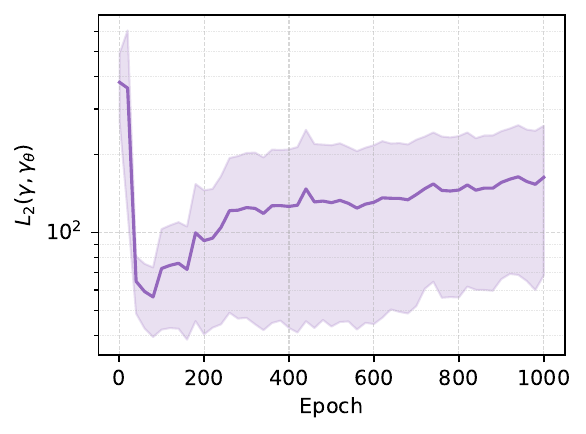}
    \end{subfigure}
    \caption{Predicted mean field flow for the Cybersecurity MFG for two random samples from the test set. The two rightmost panes respectively show the train and test error and the $L_2$ error of the predicted parameters. }
\end{figure}

\section{Cycling with Running Costs}\label{appendix: biking_running}
To isolate and demonstrate the structural, forward-looking power of the MFG formulation, we evaluate the two models from Sec.~\ref{sec: experiments} under three restricted-capacity testing regimes. 

\noindent\textbf{Test 1: Constant Base Rates.} We restrict the base transition rates to be entirely constant over time: $\varphi_\theta(t) = \gamma_0$ (a learned but time-independent vector). A standard ODE with constant rates decays toward a stationary distribution and cannot recreate the complex, double-peaked morning and evening commute patterns. The MFG, however, learns a non-trivial terminal cost $w_g(x)$ representing commuters' end-of-day location preferences. The backward HJB equation propagates this terminal incentive backward in time, causing agents to strategically adjust their flows throughout the day. As shown in Table~\ref{tab:bikeshare}, the MFG achieves a $41.7\%$ reduction in Test~ID $L_2$ error compared to the baseline.

\noindent\textbf{Test 2: Linear Base Rates.} We allow the base rates to be an affine function of time: $\varphi_\theta(t) = W t + b_{\text{bias}}$, where $W \in \RR^{d(d-1) \times 1}$ and $b_{\text{bias}} \in \RR^{d(d-1)}$. This enables the baseline to capture monotonic trends over the day but prevents it from fitting highly non-linear dynamics. Even with this added expressivity, the MFG maintains a $12.0\%$ Test~ID improvement (Table~\ref{tab:bikeshare}).

\noindent\textbf{Test 3: Restricted Dynamics with General Cost.} We consider a practitioner who relies on a rigid phenomenological model for the baseline dynamics (linear base rates, as in Test~2) and does not wish to alter it. Instead, we augment the MFG with a flexible cost network $\psi_\phi : [0,T] \times \cP([d]) \to \RR^{d(d-1)}$ (a 2-layer MLP with 64 hidden units and $\tanh$ activations), which replaces the scalar congestion cost $c \cdot \mu(y)$ in~\eqref{eq:bike_cost} with a general, potentially negative, cost vector $\gamma_{\text{cost}}(t, \mu_t) := \psi_\phi(t, \mu_t) \in \RR^{d(d-1)}$. The running cost becomes:
\begin{equation}\label{eq:bike_cost_general}
	f(x, a, \mu, t) = b \sum_{y \neq x} (a_{xy} - q_{xy}^{\mathrm{base}})^2 - \sum_{y \neq x} [\gamma_{\text{cost}}]_{xy}(t, \mu) \, a_{xy},
\end{equation}
and the corresponding optimal control is:
\begin{equation}
	a_{xy}^* = \max\!\left(0, \; q_{xy}^{\mathrm{base}} + \frac{[\gamma_{\text{cost}}]_{xy} - (u(y) - u(x))}{2b}\right).
\end{equation}
Note that the sign of the linear cost term is reversed compared to~\eqref{eq:bike_cost}, and the network $\psi_\phi$ can output negative values, enabling the model to learn both congestion penalties and incentives. As shown in Table~\ref{tab:bikeshare}, this formulation achieves a $19.3\%$ reduction in Test~ID $L_2$ error over the linear baseline, nearly doubling the $12.0\%$ improvement of Test~2.

\begin{table}[ht]
    \centering
    \caption{Prediction Errors on Citi Bike Trajectories (Mean $\pm$ SD over 5 seeds). Test ID is the held-out weekdays; Test OOD is the weekends. Absolute $L_2$ errors are scaled by $10^3$.}
    \label{tab:bikeshare}
    \resizebox{\textwidth}{!}{
    \begin{tabular}{llcc|cc|cc}
        \toprule
        & & \multicolumn{2}{c}{\textbf{Train}} & \multicolumn{2}{c}{\textbf{Test ID}} & \multicolumn{2}{c}{\textbf{Test OOD}} \\
        \textbf{Model} & \textbf{Complexity} & \textbf{$L_2$} & \textbf{Rel. $L_2$} & \textbf{$L_2$} & \textbf{Rel. $L_2$} & \textbf{$L_2$} & \textbf{Rel. $L_2$} \\
        \midrule
        MF Dynamics & Constant & $1.55 \pm 0.00$ & $0.153 \pm 0.000$ & $1.32 \pm 0.00$ & $0.054 \pm 0.000$ & $3.09 \pm 0.02$ & $0.091 \pm 0.000$ \\
        Mean-Field Game & Constant & $\mathbf{0.94 \pm 0.00}$ & $\mathbf{0.083 \pm 0.000}$ & $\mathbf{0.77 \pm 0.01}$ & $\mathbf{0.036 \pm 0.000}$ & $3.14 \pm 0.05$ & $0.101 \pm 0.001$ \\
        \midrule
        MF Dynamics & Linear & $1.02 \pm 0.00$ & $0.092 \pm 0.001$ & $0.83 \pm 0.00$ & $0.035 \pm 0.000$ & $\mathbf{3.07 \pm 0.02}$ & $0.104 \pm 0.000$ \\
        Mean-Field Game & Linear & $0.91 \pm 0.00$ & $0.085 \pm 0.001$ & $0.73 \pm 0.00$ & $0.033 \pm 0.000$ & $3.18 \pm 0.06$ & $0.103 \pm 0.001$ \\
        Mean-Field Game & Linear + General Cost & $\mathbf{0.84 \pm 0.00}$ & $\mathbf{0.081 \pm 0.001}$ & $\mathbf{0.67 \pm 0.00}$ & $\mathbf{0.032 \pm 0.000}$ & $3.19 \pm 0.00$ & $0.105 \pm 0.002$ \\
        \bottomrule
    \end{tabular}
    }
\end{table}

\begin{figure}[ht]
    \centering
    \includegraphics[width=0.48\textwidth]{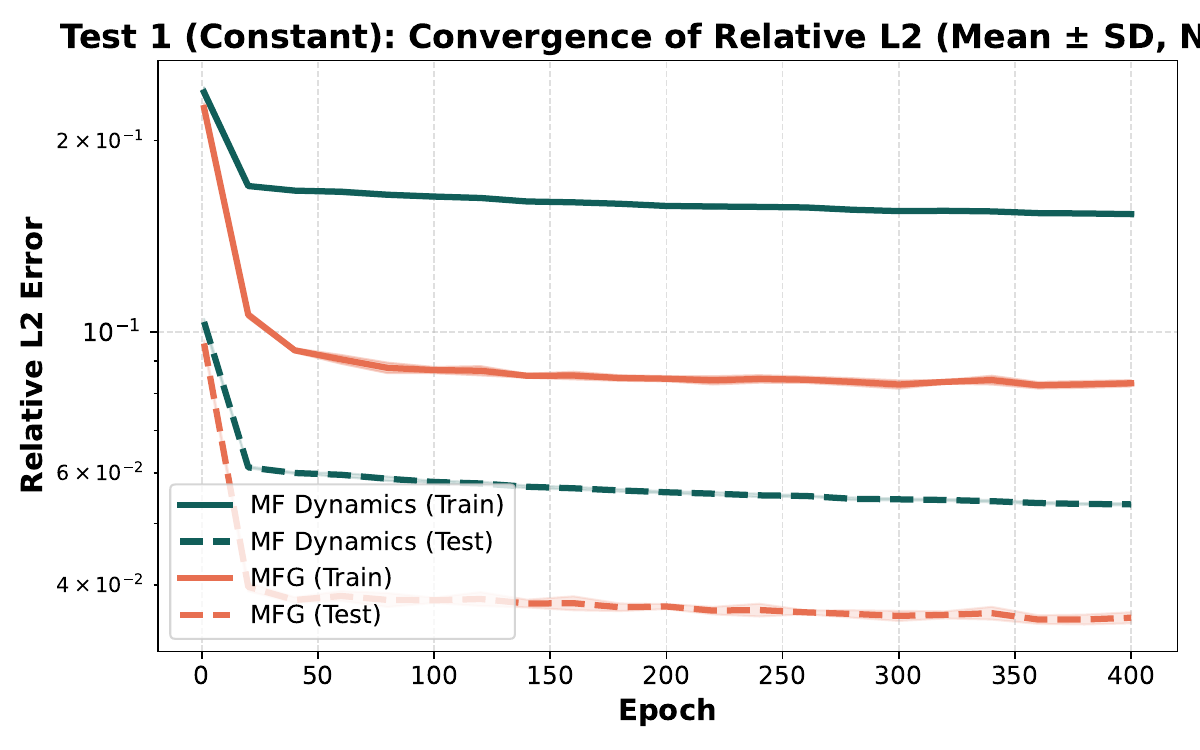}
    \includegraphics[width=0.48\textwidth]{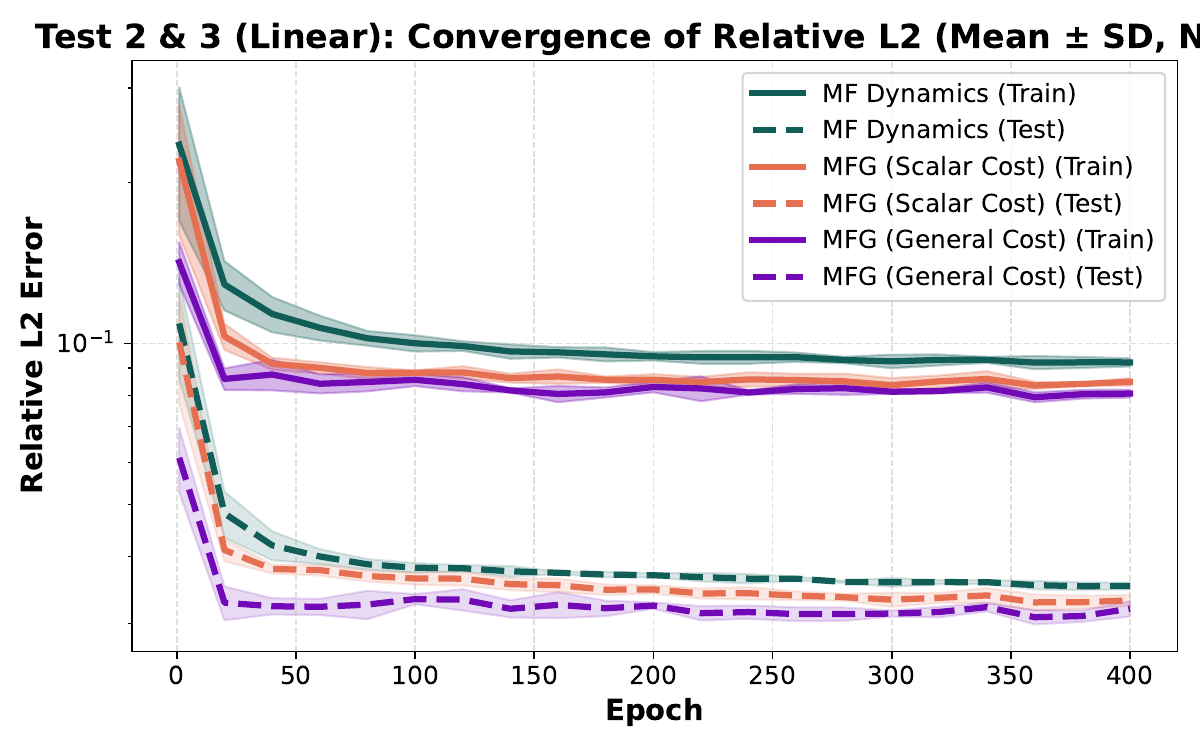}
    \caption{Relative $L_2$ error convergence. Left (Test 1): Constant Base Rates. Right (Test 2 \& 3): Linear Base Rates. The General Cost MFG (purple) clearly achieves the lowest test error, demonstrating the advantage of highly flexible cost networks. Shaded regions denote $\pm 1$ standard deviation over 5 seeds.}
    \label{fig:bikeshare_convergence}
\end{figure}

\begin{figure}[ht]
    \centering
    \includegraphics[width=0.8\textwidth]{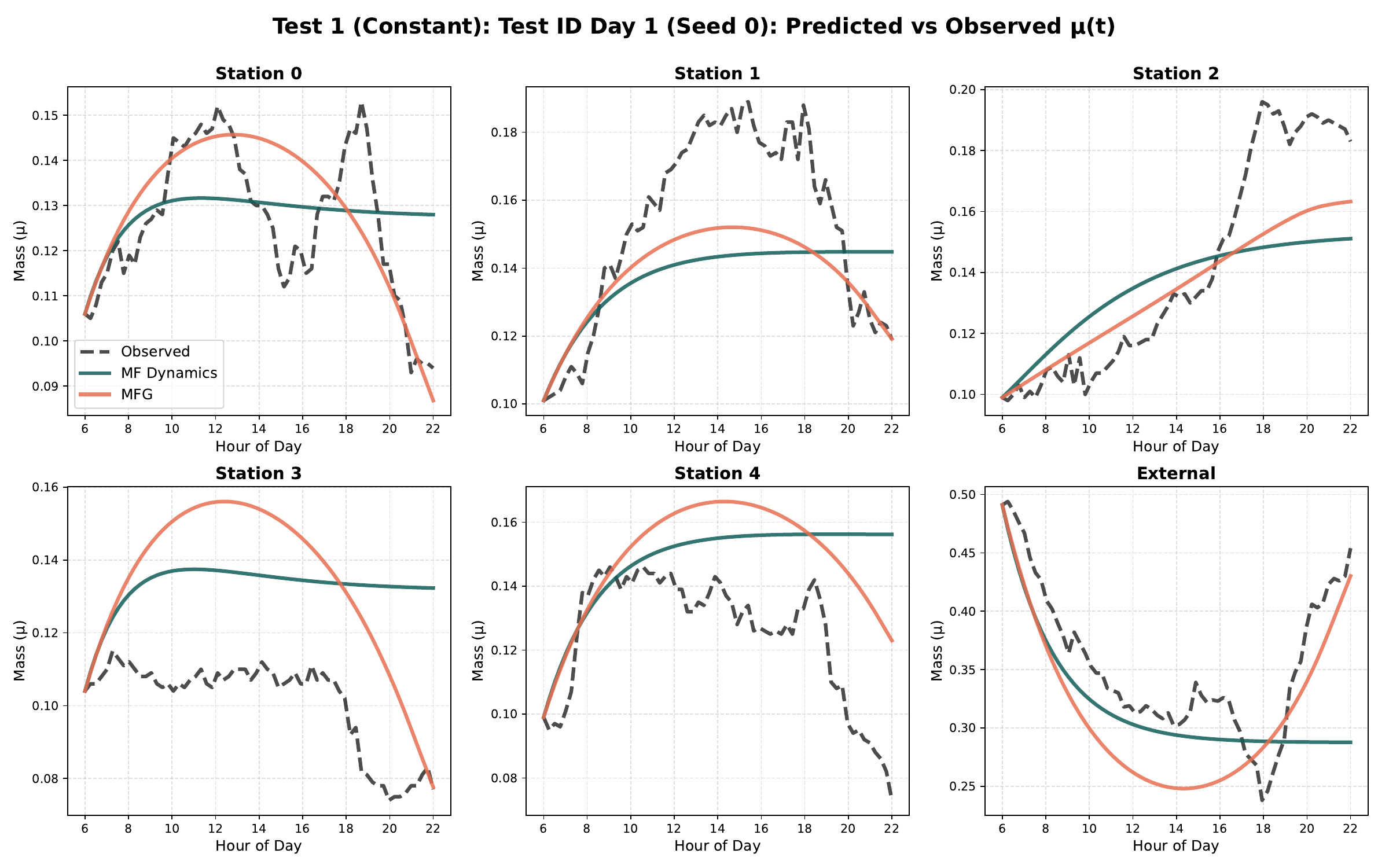}

    \includegraphics[width=0.8\textwidth]{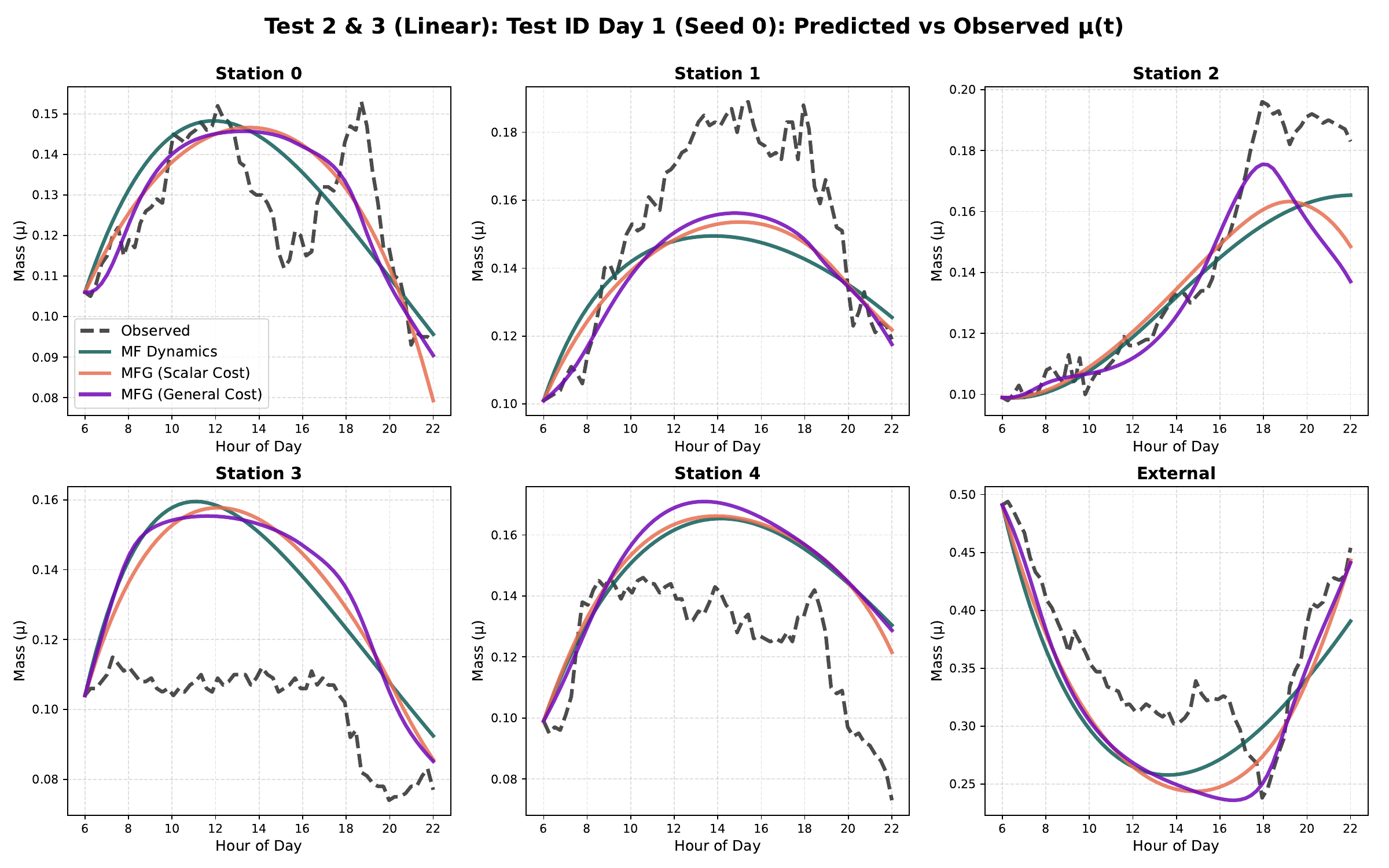}
    \caption{Predicted versus Observed Trajectories on a typical Weekday (Test ID, Seed 0). Top (Test 1): Constant Base Rates. Bottom (Test 2 \& 3): Linear Base Rates. Please refer to the text for interpretation.}
    \label{fig:bikeshare_trajectories_id}
\end{figure}

\begin{figure}[ht]
    \centering
    \includegraphics[width=0.48\textwidth]{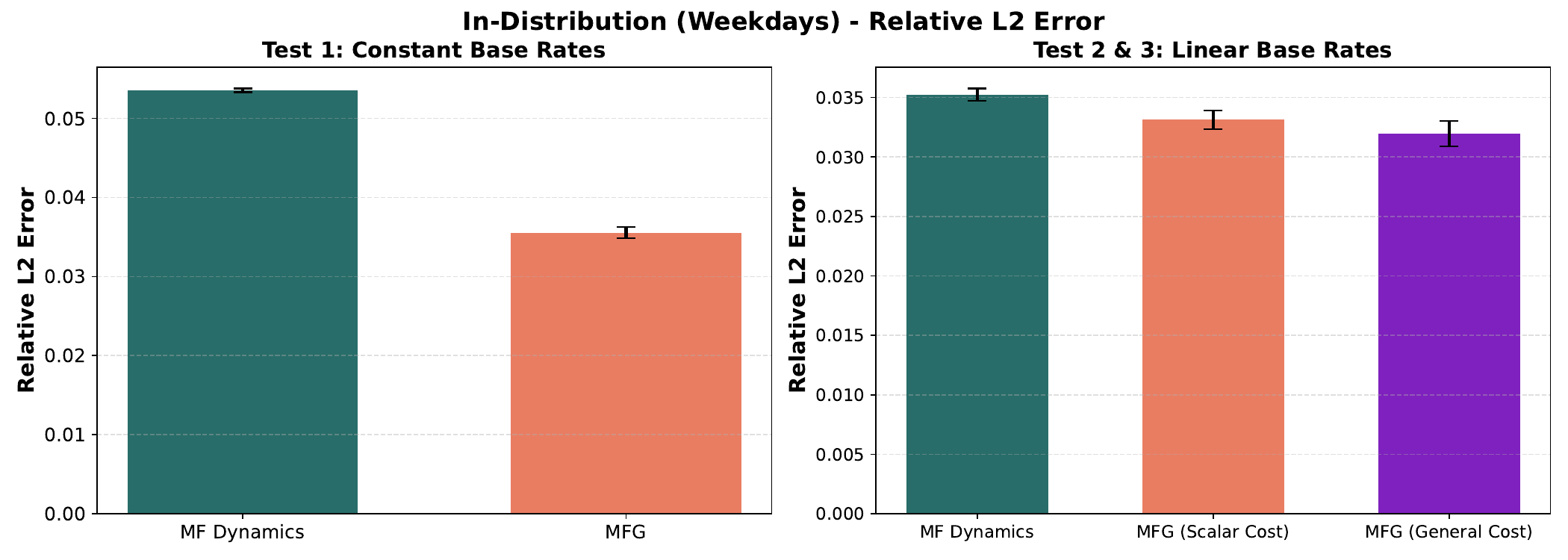}
    \includegraphics[width=0.48\textwidth]{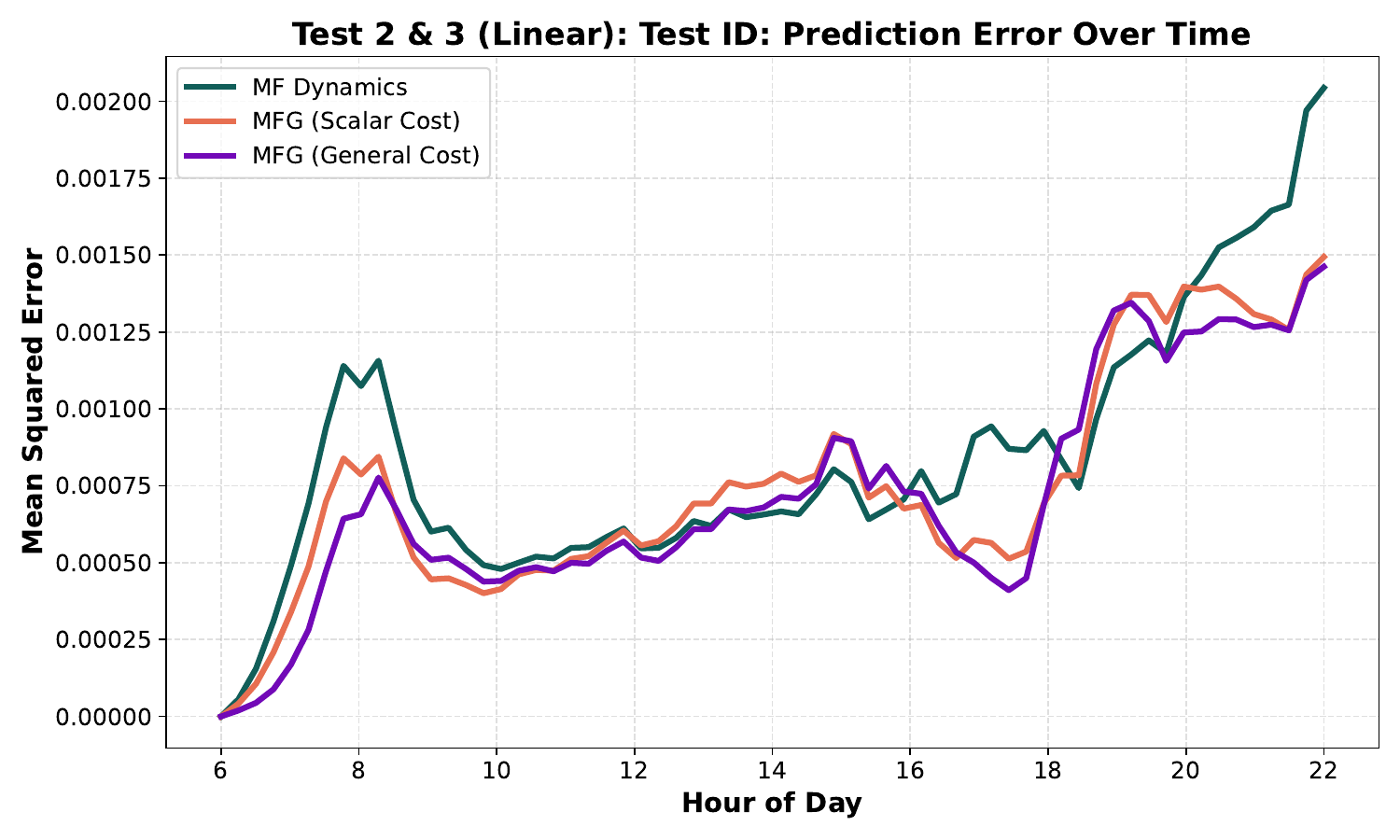}
    \caption{Left: In-Distribution (Test ID) Relative $L_2$ comparison across all regimes. \textbf{Right:} Prediction Error Over Time for Test 2 \& 3.}
    \label{fig:bikeshare_comparison_id}
\end{figure}

Interestingly, both models fail to generalize to the Out-of-Distribution (OOD) test set, which comprises the weekends in March 2024. This occurs due to two distinct domain shifts. First, environmental shocks: the first OOD day (Saturday, March 2nd) experienced rain, causing station usage to stagnate. Second, latent structural shifts: sunny weekend days (e.g., Sunday, March 3rd) exhibit broad, smooth leisure riding rather than the sharp bimodal commuting spikes of weekdays. Because the models were trained exclusively on weekday commute rhythms, they incorrectly predict commute-like surges on weekends. As illustrated in Fig.~\ref{fig:bikeshare_trajectories_ood}, this yields substantial errors, demonstrating a generalization failure when the underlying user utility shifts.
The slightly smaller OOD error of the MF dynamics baseline should therefore be interpreted as a less severe failure under this particular weekend shift, rather than as evidence of better out-of-distribution modelling.

\begin{figure}[ht]
    \centering
    \includegraphics[width=0.8\textwidth]{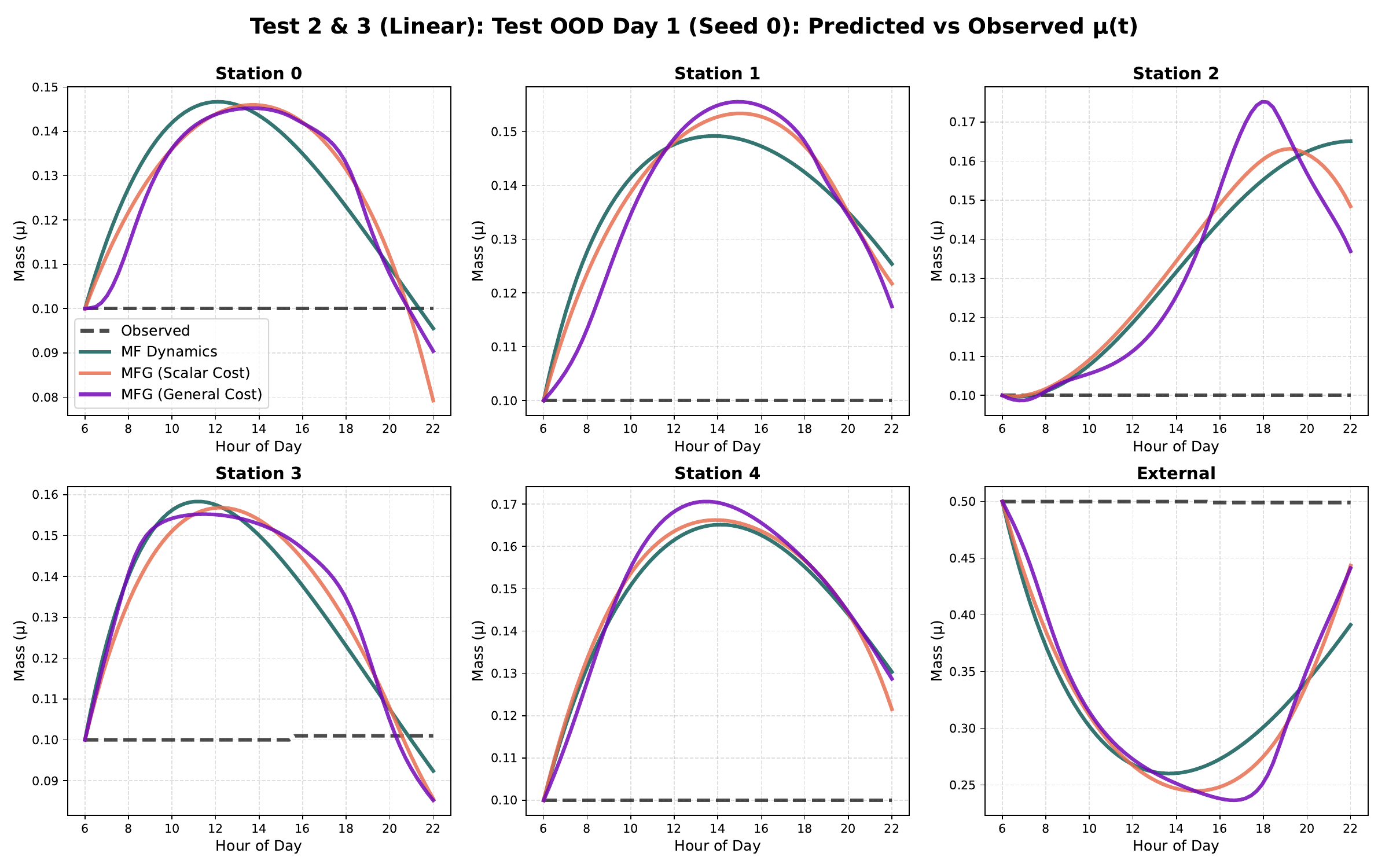}

    \includegraphics[width=0.8\textwidth]{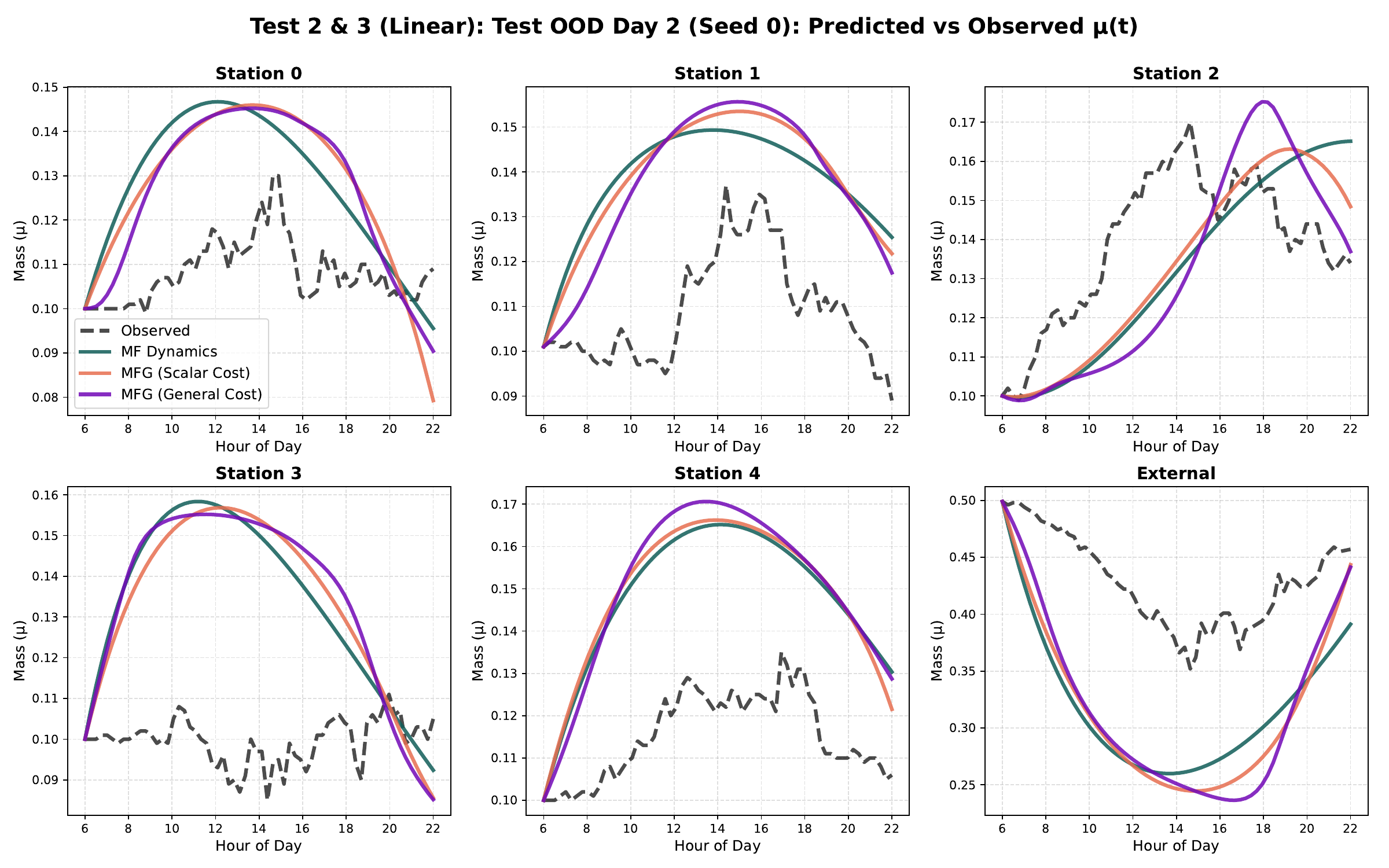}
    \caption{OOD Test Set (Weekends) using Test 2 \& 3 models. Top (Day 1): A rainy day caused reduced movement; models predicted weekday commutes. Bottom (Day 2): A sunny Sunday with smooth leisure riding, unlike the sharp commute peaks the models learned.}
    \label{fig:bikeshare_trajectories_ood}
\end{figure}

\clearpage

\section{Cycling with Interventions}\label{app_bike_scenarios}
We present the mean field mass per station for the two different intervention scenarios of the urban mobility analysis below. Fig.~\ref{fig: intervention_1} presents the predicted dynamics for intervention scenario 1, where station 1 is closed from 14:45 until 17:30. Fig.~\ref{fig: intervention_2} presents the predicted dynamics for intervention scenario 2, which involves the closure of station 2 from 07:15 until 09:45. 
\begin{figure}[H]
    \centering
    \includegraphics[width=0.8\linewidth]{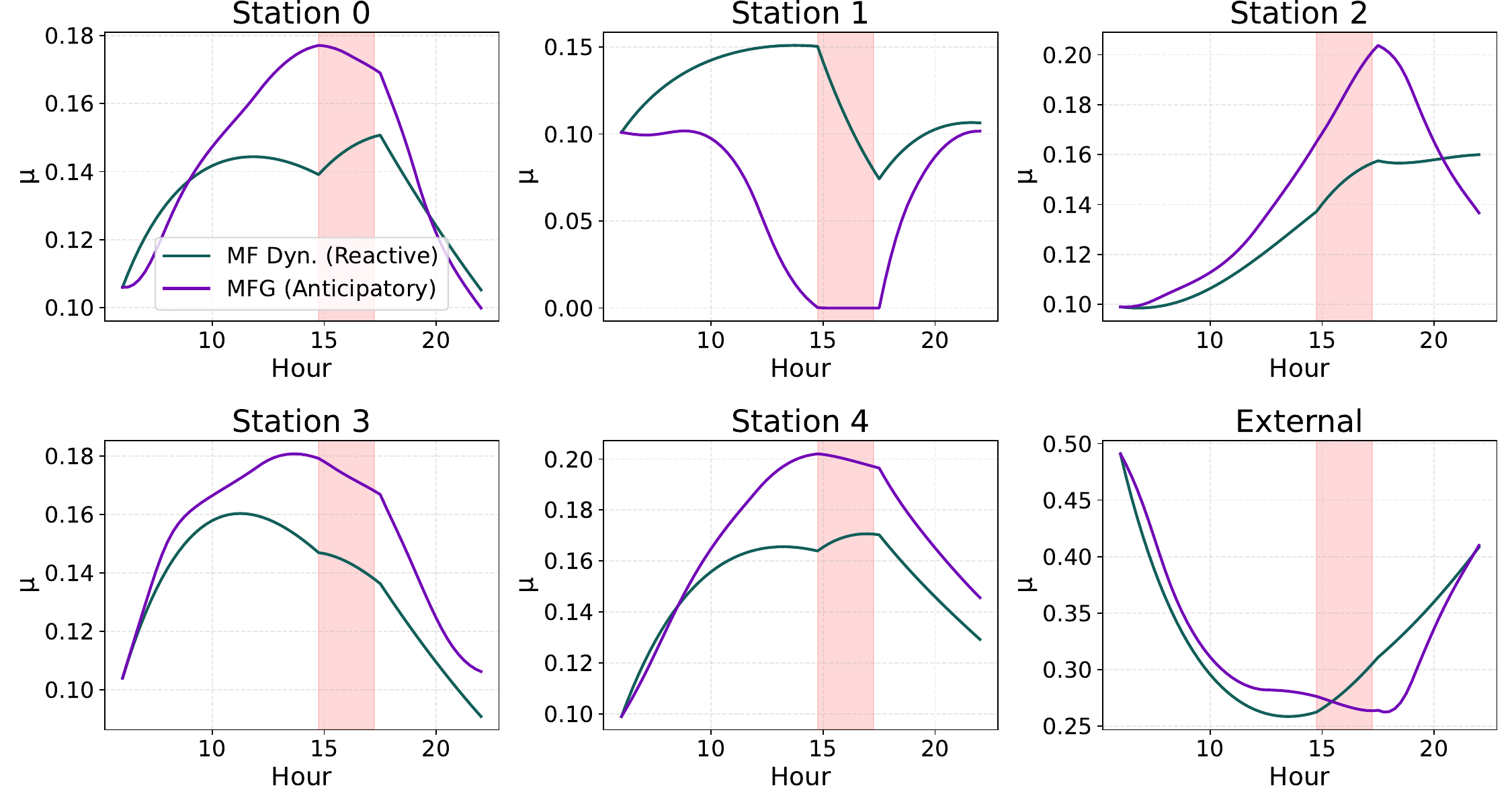}
    \caption{The mean field mass per station for intervention scenario 1.}
    \label{fig: intervention_1}
\end{figure}

\begin{figure}[H]
    \centering
    \includegraphics[width=0.8\linewidth]{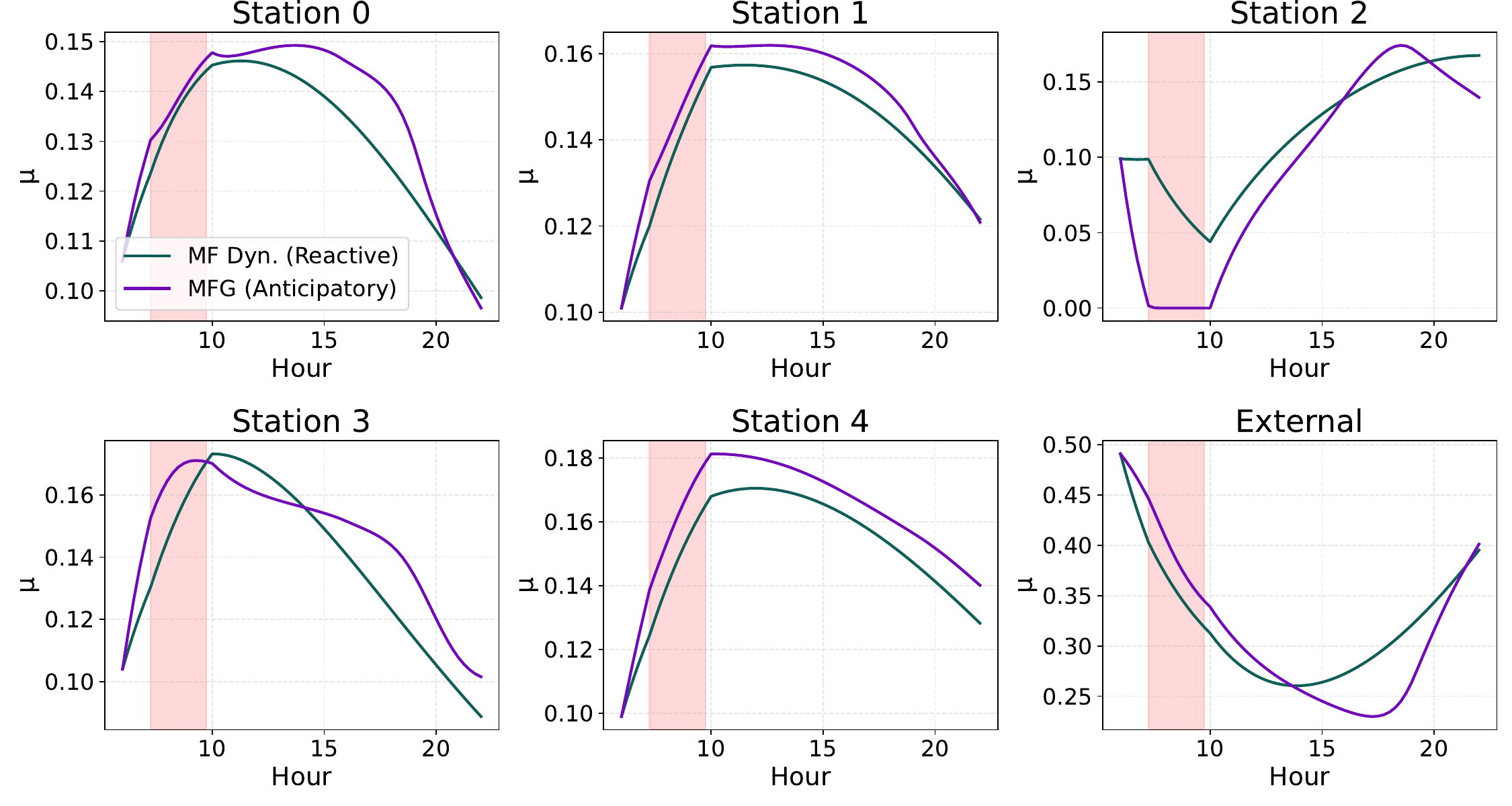}
    \caption{The mean field mass per station for intervention scenario 2.}
    \label{fig: intervention_2}
\end{figure}

\section{Experimental Parameters}\label{appendix: experimental_parameters}

\subsection{Hardware}
All experiments are implemented in Python using JAX \cite{jax2018} and are executed on a MacBook Pro equipped with an Apple M1 Pro chip  and 16 GB of RAM. Each experiment is performed $n=5$ times with different random seeds. 

\subsection{Random Seeds} For all experiments we use the random seeds $[42, 43, 44, 45, 46]$ to initialize the numpy and JAX pseudorandom number generation processes. 
\subsection{Linear-Quadratic}
We model the unknown parameter $\gamma \equiv b(t, \mu)$ using a feed-forward neural network. To improve the network's ability to capture highly dynamic and localized parameter shapes, we employ a sub-trajectory training approach, where gradients are computed over short, randomized time windows ($20\%$ of the total horizon). The full set of physical and optimization parameters used for data generation and network training is summarized in Table~\ref{tab:lq_params}.

\begin{table}[h]
\centering
\caption{Physical parameters and neural network hyperparameters for the LQ MFG experiment.}
\label{tab:lq_params}
\resizebox{\textwidth}{!}{%
\begin{tabular}{lc|lr}
\toprule
\textbf{Physical Parameters} & \textbf{Value} & \textbf{Training Hyperparameters} & \textbf{Value} \\
\midrule
Number of states ($d$) & 3 & Learning rate & $5 \times 10^{-3}$ \\
Time horizon ($T$) & 2.0 & Optimizer & Adam \\
Number of points in time ($N+1$) & 100 & Batch size & 10 \\
Min. parameter ($b_{\min}$) & 0.1 & Training samples & 200 \\
Max. parameter ($b_{\max}$) & 3.0 & Test samples & 20 \\
& & Sub-trajectory length & $20\%$ of $T$ \\
& & Number of neural network layers & 3\\
& & Layer width & 64\\
& & Activation & ReLU\\
& & Final activation & Softplus\\
  &  & Training Epochs & 500 \\
\bottomrule
\end{tabular}%
}
\end{table}
\subsection{Cybersecurity}
We model the unknown parameters $\gamma \equiv (q_{rec}^D, q_{rec}^U, q_{inf}^D,\allowbreak q_{inf}^U, \beta_{DD}, \beta_{DU}, \beta_{UU}, \beta_{UD})$ using a feed-forward neural network. To improve the network's ability to capture highly dynamic and localized parameter shapes, we employ a sub-trajectory training approach, where gradients are computed over short, randomized time windows ($20\%$ of the total horizon). The full set of physical and optimization parameters used for data generation and network training is summarized in Table~\ref{tab:cs_params}.
\begin{table}[H]
\centering
\caption{Physical parameters and neural network hyperparameters for the cybersecurity MFG experiment.}
\label{tab:cs_params}
\begin{tabular}{lc|lr}
\toprule
\textbf{Physical Parameters} & \textbf{Value} & \textbf{Training Hyperparameters} & \textbf{Value} \\
\midrule
Number of states ($d$) & 4 & Learning rate & $5 \times 10^{-3}$ \\
Time horizon ($T$) & 10.0 & Optimizer & Adam \\
Number of points in time ($N+1$) & 100 & Batch size & 10 \\
$q_{rec}^D$ & $[0.0, 0.5]$ & Training samples & 200 \\
$q_{rec}^U$ & $[0.0, 0.4]$& Test samples & 20 \\
$q_{inf}^D$ & $[0.0, 0.4]$& Sub-trajectory length & $20\%$ of $T$ \\
$q_{inf}^U$ & $[0.0, 0.3]$ & Number of neural network layers & 3\\
$\beta_{DD}$ & $[0.0, 0.4]$ & Layer width & 64\\
$\beta_{DU}$ & $[0.0, 0.3]$ & Activation & ReLU\\
$\beta_{UU}$ & $[0.0, 0.3]$ & Final activation & Softplus\\
$\beta_{UD}$ & $[0.0, 0.4]$ & Training Epochs& 1000\\
$\lambda$ & 0.8 & &\\
$v_H$ & 0.6 & &\\
$k_D$ & 0.3 & &\\
$k_I$ & 0.5 & &\\
\bottomrule
\end{tabular}
\end{table}
\subsection{Susceptible-Infected-Recovered}
We model the unknown parameters $\gamma \equiv (\beta, \gamma, c_I, c_\alpha)$ using a feed-forward neural network. To improve the network's ability to capture highly dynamic and localized parameter shapes, we employ a sub-trajectory training approach, where gradients are computed over short, randomized time windows ($20\%$ of the total horizon). The full set of physical and optimization parameters used for data generation and network training is summarized in Table~\ref{tab:sir_params}. Because the parameters $\beta$ and $\gamma$ interact with the players' contact factor $\alpha \in [0, 1]$, they are unbounded in theory. In practice, the parameters often converge to a reasonable domain under Picard Damping. 
\begin{table}[H]
\centering
\caption{Physical parameters and neural network hyperparameters for the SIR MFG experiment.}
\label{tab:sir_params}
\begin{tabular}{lc|lr}
\toprule
\textbf{Physical Parameters} & \textbf{Value} & \textbf{Training Hyperparameters} & \textbf{Value} \\
\midrule
Number of states ($d$) & 3 & Learning rate & $5 \times 10^{-3}$ \\
Time horizon ($T$) & 30 & Optimizer & Adam \\
Number of points in time ($N+1$) & 30 & Batch size & 10 \\
$\gamma$ & $[0, \infty)$ & Training samples & 200 \\
$\beta$ & $[0, \infty)$& Test samples & 20 \\
$c_a$ & $[0, \infty)$ & Sub-trajectory length & $20\%$ of $T$ \\
$c_I$& $[0, \infty)$ & Number of neural network layers & 3\\
Picard Max Iterations & $200$ & Layer width & 64\\
Picard Tolerance & $10^{-5}$  & Activation & ReLU\\
Picard Damping ($\alpha$) & $0.5$  & Final activation & Softplus\\
  &  & Training Epochs & 1000 \\
\bottomrule
\end{tabular}
\end{table}

\subsection{Urban Mobility}
The physical parameters and training hyperparameters for the bike-sharing experiment are presented in Table~\ref{tab:bikeshare_hyperparams}.
\begin{table}[H]
    \centering
    \caption{Hyperparameters for Urban Mobility Experiments.}
    \label{tab:bikeshare_hyperparams}
    \begin{tabular}{lc|lr}
        \toprule
        \textbf{Physical parameters} & \textbf{Value} & \textbf{Training Hyperparameters} & \textbf{Value} \\
        \midrule
        Number of Stations ($d$) & $6$ & Learning Rate & $5 \times 10^{-3}$ \\
        Time Horizon ($T$) & $1.0$ (16 hours) & Optimizer & Adam \\
        Number of points in time ($N+1$) & $64$ & Epochs & $400$ \\
        Picard Max Iterations & $200$ &  & \\
        Picard Tolerance & $10^{-5}$ & Initial $b_0$ & $1.0$ \\
        Picard Damping ($\alpha$) & $0.5$ & Initial $c_0$ & $0.5$ \\
        Base Network (Tests 1--2) & Constant / Affine & Cost Network (Test 3) & 2-Layer MLP (64, $\tanh$) \\
        \bottomrule
    \end{tabular}
\end{table}

\clearpage
\section{Convergence Analysis}
Figs.~\ref{fig:lq_losses}-\ref{fig:loss_bike} present the convergence of the training and testing losses for the mean field $\mu$, as well as the $L_2$ error of the parameter $\gamma$ across training epochs. The test loss evaluates the full trajectory prediction, while the training loss reflects the localized sub-trajectory optimization.
\subsection{Linear-Quadratic}
\phantom{LQ}
\begin{figure}[H]
    \centering
        \begin{subfigure}[b]{0.32\textwidth}
        \includegraphics[width=\textwidth]{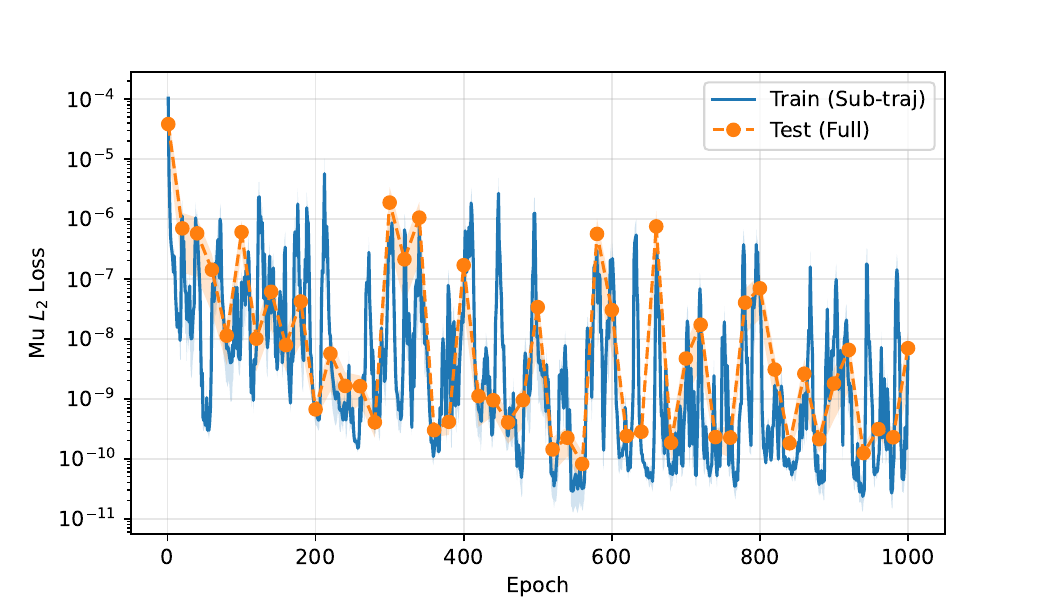}
    \end{subfigure}
    \hfill
    \begin{subfigure}[b]{0.32\textwidth}
        \includegraphics[width=\textwidth]{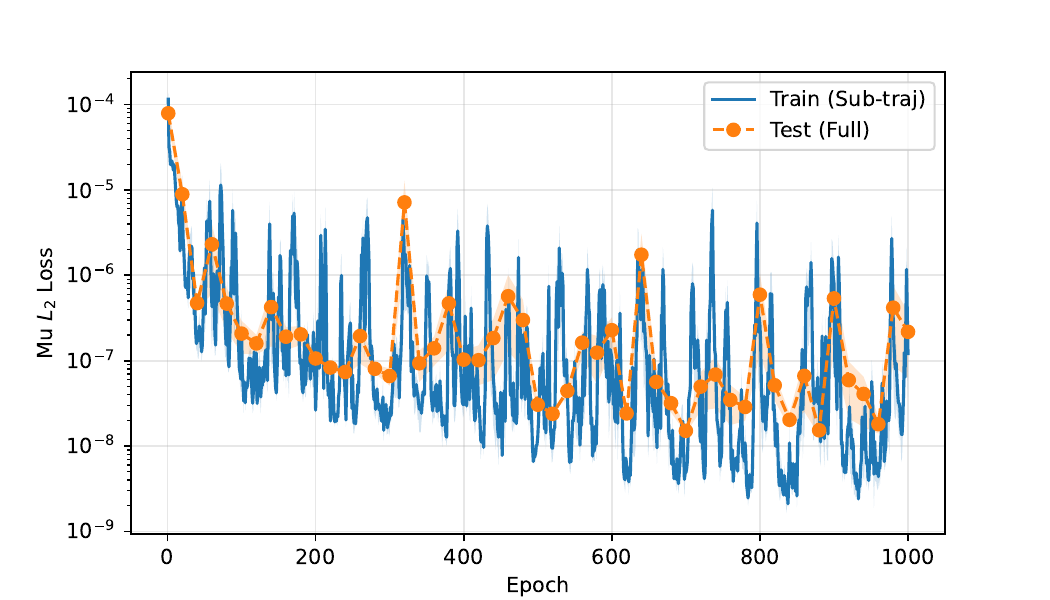}
    \end{subfigure}
    \hfill
    \begin{subfigure}[b]{0.32\textwidth}
        \includegraphics[width=\textwidth]{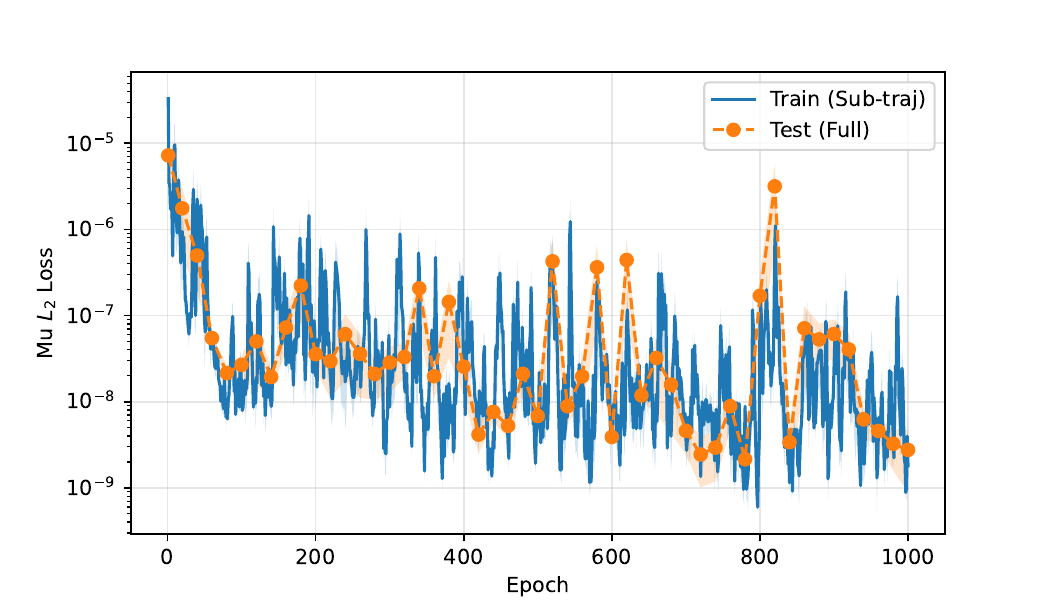}
    \end{subfigure}
    
    \vspace{0.2cm}

    \begin{subfigure}[b]{0.32\textwidth}
        \includegraphics[width=\textwidth]{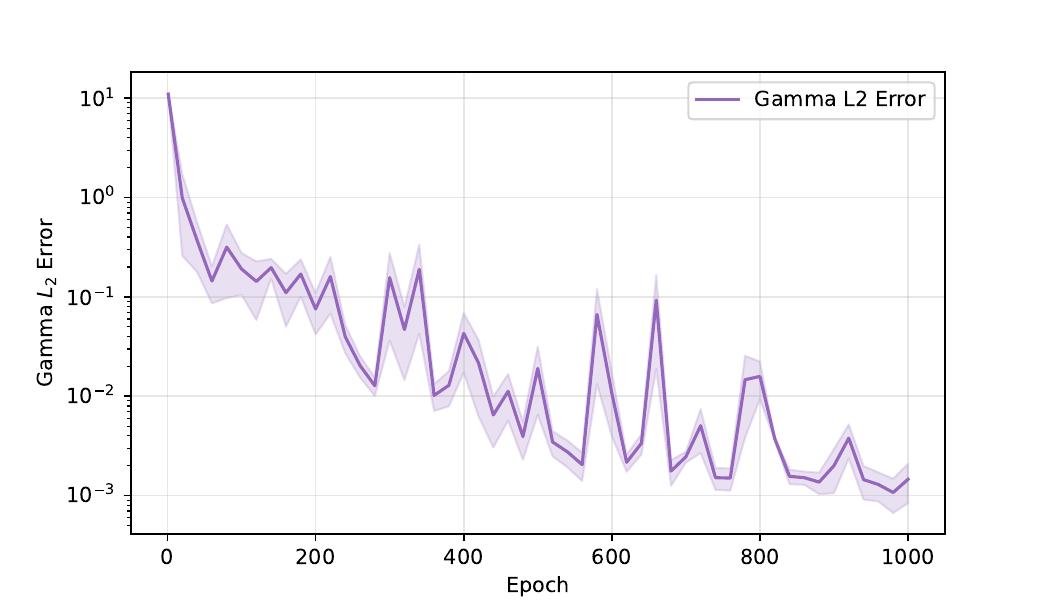}
        \caption{Constant}
    \end{subfigure}
    \hfill
    \begin{subfigure}[b]{0.32\textwidth}
        \includegraphics[width=\textwidth]{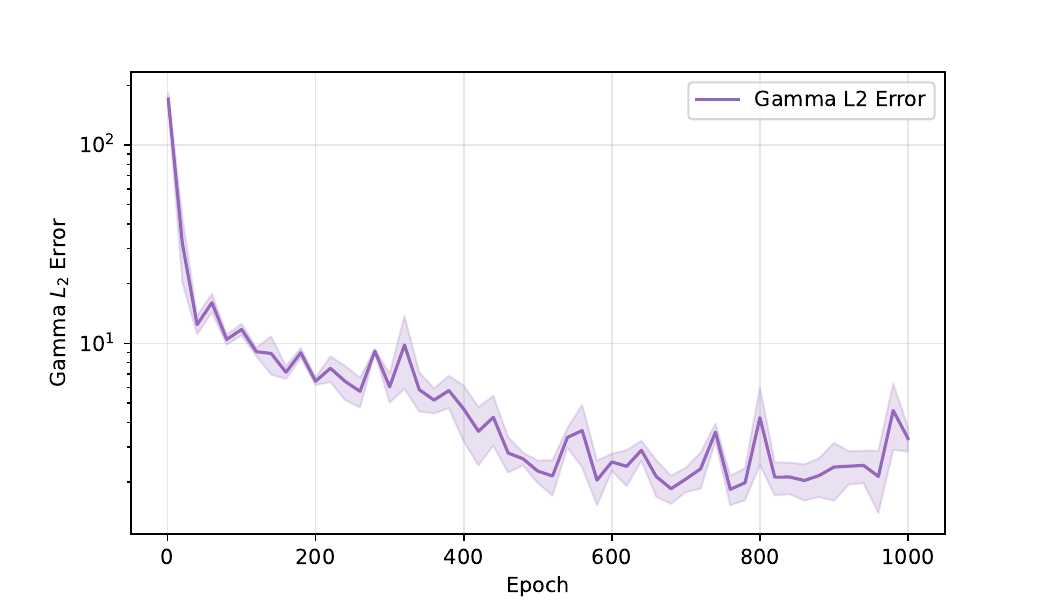}
        \caption{Time-dependent}
    \end{subfigure}
    \hfill
    \begin{subfigure}[b]{0.32\textwidth}
        \includegraphics[width=\textwidth]{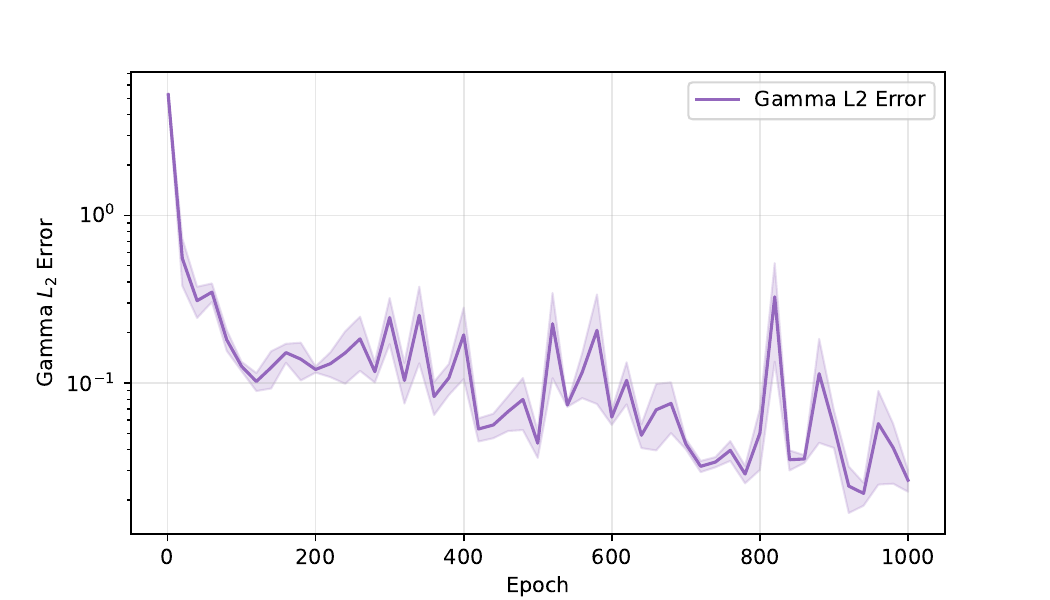}
        \caption{MF-dependent}
    \end{subfigure}
    \caption{Convergence of the mean field $L_2$ loss (top row) and the parameter $L_2$ error (bottom row) during training for the three test cases without noise. Shaded regions denote the standard error over multiple random seeds.}
    \label{fig:lq_losses}
\end{figure}

\subsection{Mean Field Control}
\phantom{MFC}
\begin{figure}[H]
    \centering
    \begin{subfigure}[b]{0.4\textwidth}
        \includegraphics[width=\textwidth]{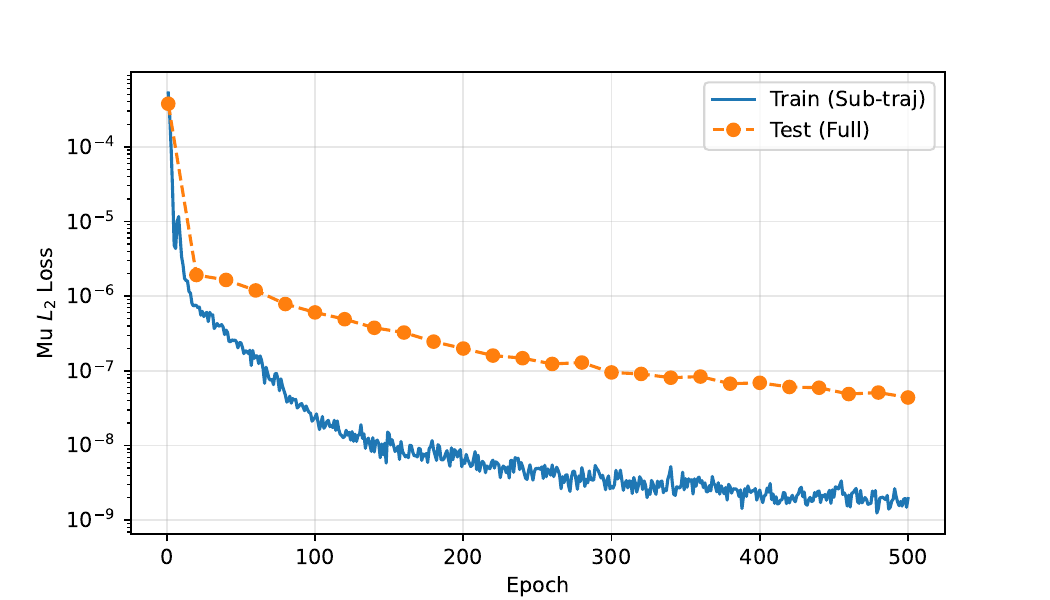}
        \caption{MFG Convergence}
    \end{subfigure}
    \hspace{1cm}
    \begin{subfigure}[b]{0.4\textwidth}
        \includegraphics[width=\textwidth]{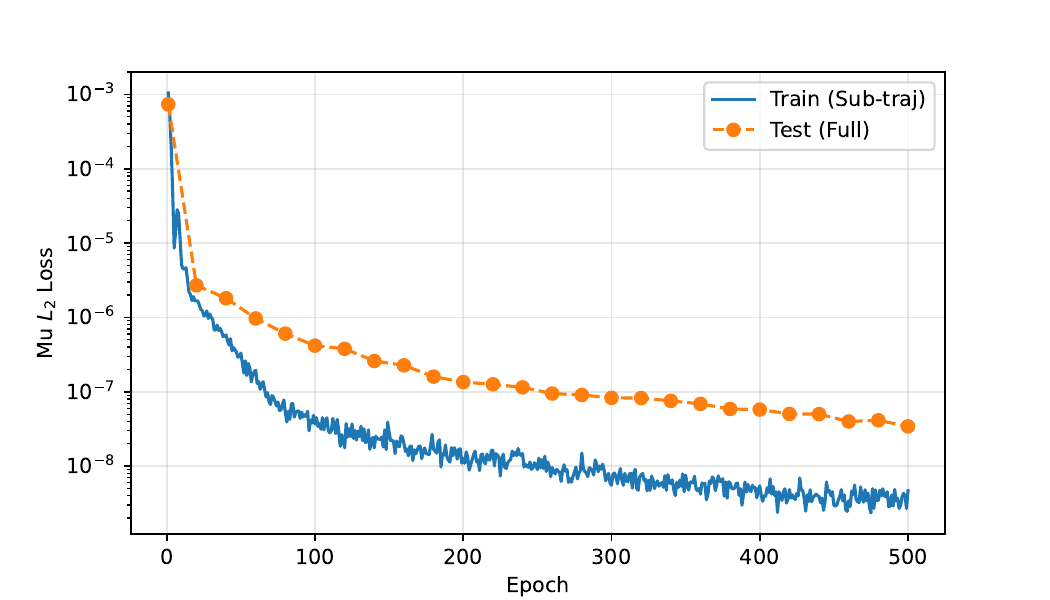}
        \caption{MFC Convergence}
    \end{subfigure}
    
    \vspace{0.2cm}
    
    \begin{subfigure}[b]{0.4\textwidth}
        \includegraphics[width=\textwidth]{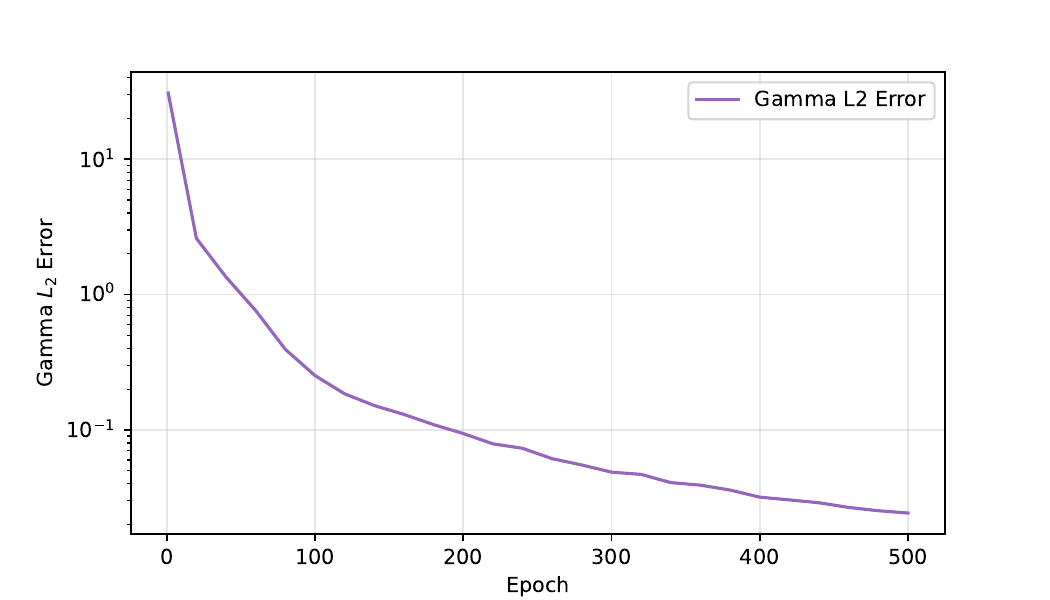}
        \caption{MFG Parameter Error}
    \end{subfigure}
     \hspace{1cm}
    \begin{subfigure}[b]{0.4\textwidth}
        \includegraphics[width=\textwidth]{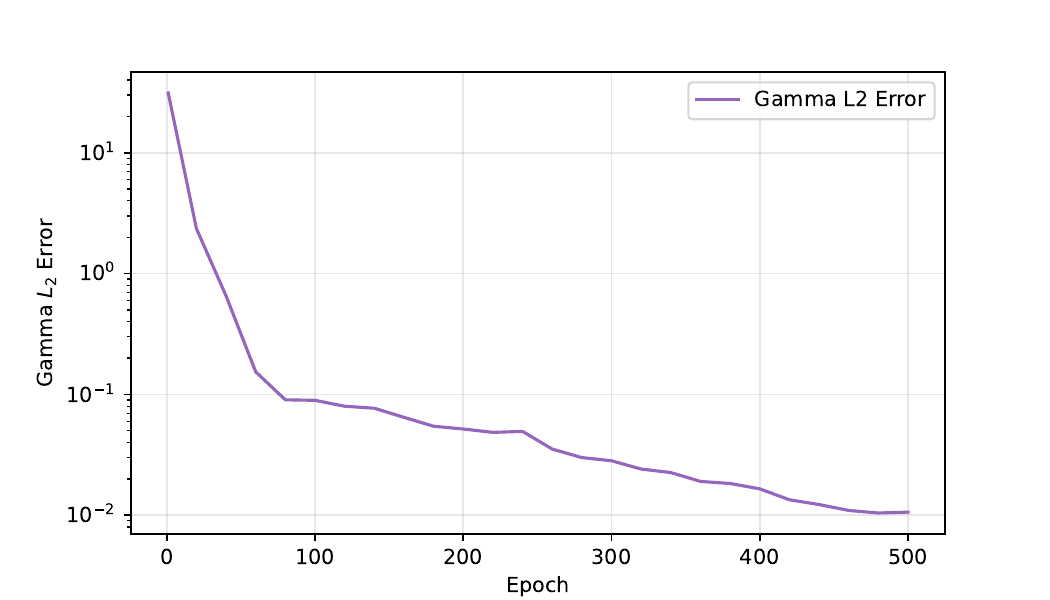}
        \caption{MFC Parameter Error}
    \end{subfigure}
    \caption{Training loss for the mean field predictions (top) and the corresponding parameter estimation error (bottom). Despite the modified ODE structure and the inclusion of Picard damping, our framework successfully recovers the underlying parameters for both the MFG and MFC configurations.}
    \label{fig:mfg_vs_mfc_losses}
\end{figure}

\subsection{Cybersecurity}
\begin{figure}[H]
    \centering
    \begin{subfigure}[b]{0.32\textwidth}
        \includegraphics[width=\textwidth]{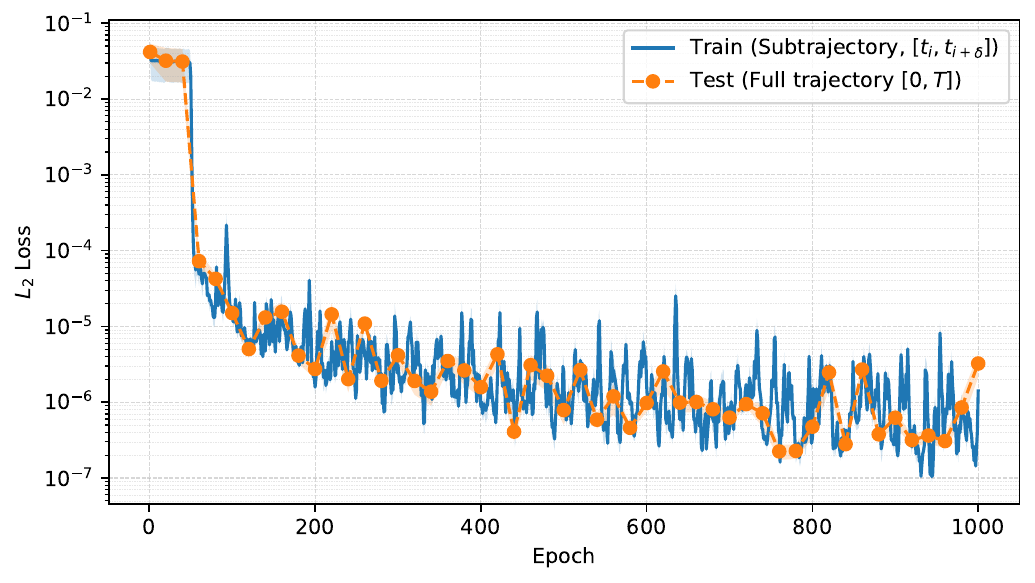}
    \end{subfigure}
    \hfill
    \begin{subfigure}[b]{0.32\textwidth}
        \includegraphics[width=\textwidth]{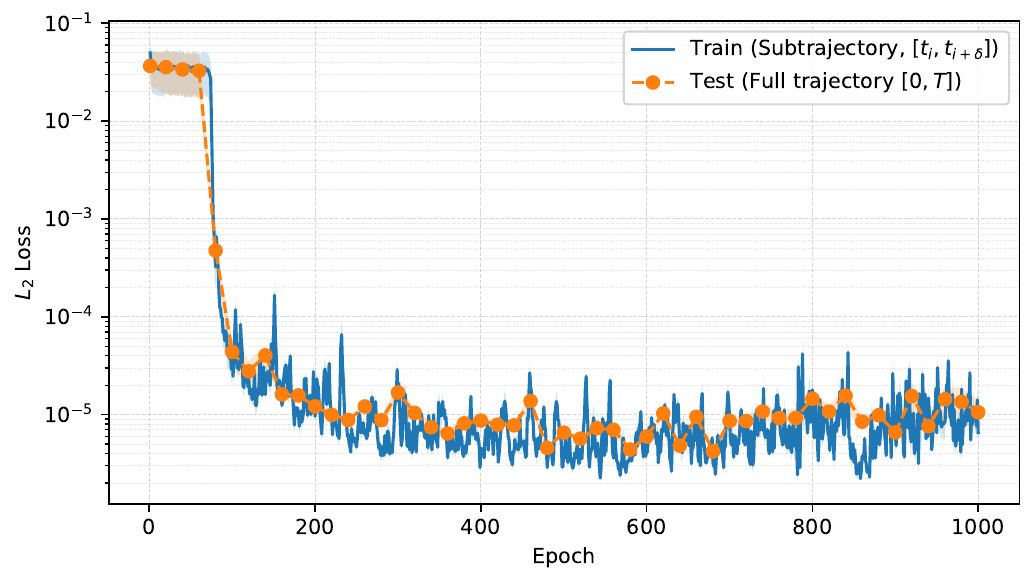}
    \end{subfigure}
    \hfill
    \begin{subfigure}[b]{0.32\textwidth}
        \includegraphics[width=\textwidth]{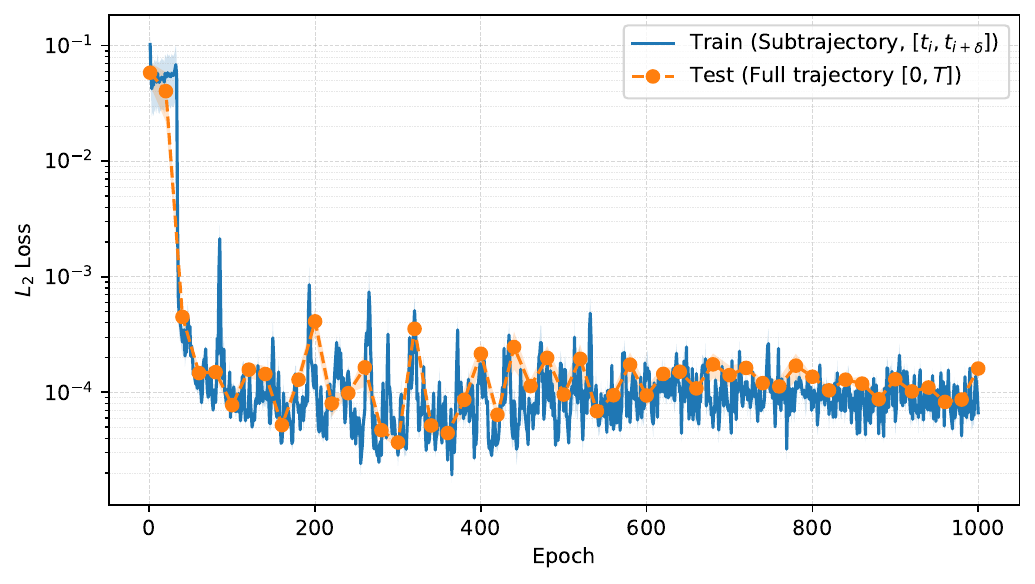}
    \end{subfigure}
    
    \vspace{0.2cm}

    \begin{subfigure}[b]{0.32\textwidth}
        \includegraphics[width=\textwidth]{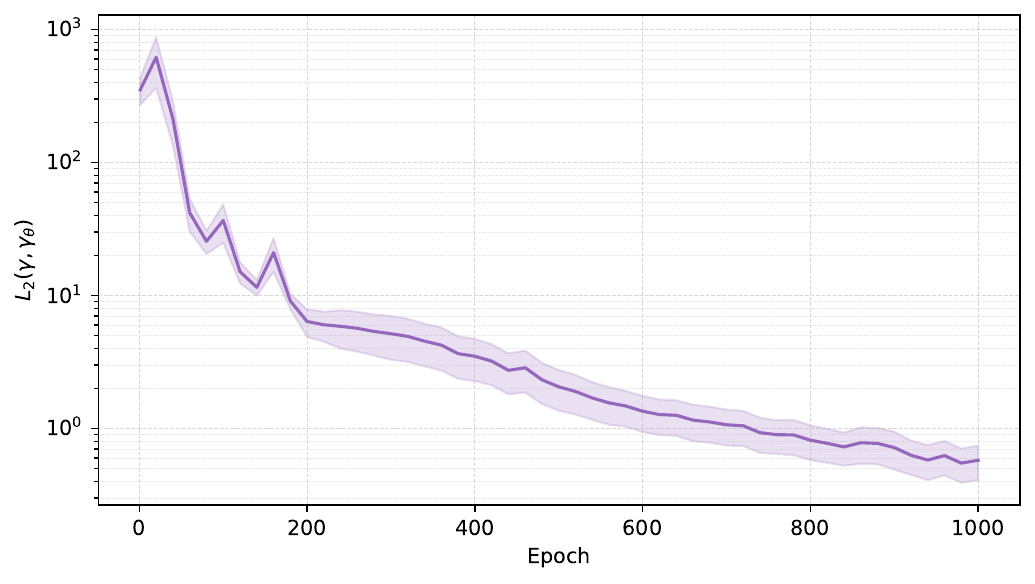}
        \caption{Constant}
    \end{subfigure}
    \hfill
    \begin{subfigure}[b]{0.32\textwidth}
        \includegraphics[width=\textwidth]{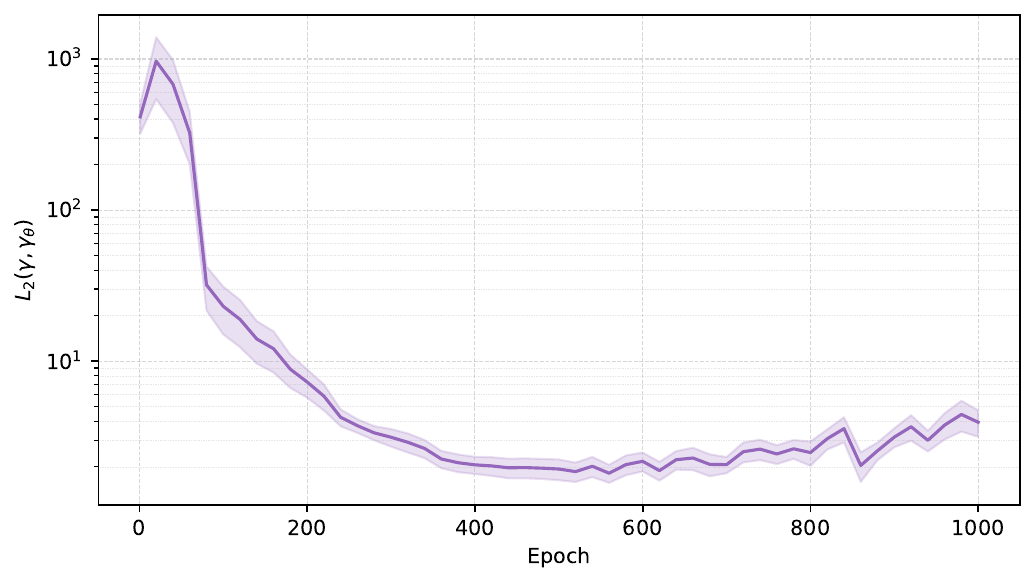}
        \caption{Time-dependent}
    \end{subfigure}
    \hfill
    \begin{subfigure}[b]{0.32\textwidth}
        \includegraphics[width=\textwidth]{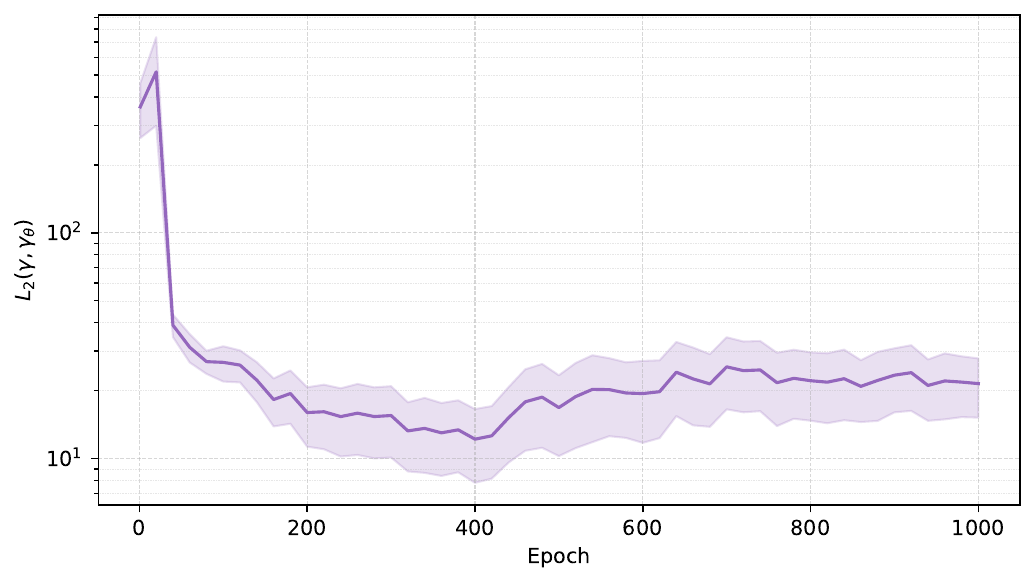}
        \caption{MF-dependent}
    \end{subfigure}
    \caption{Convergence of the mean field $L_2$ loss (top row) and the parameter $L_2$ error (bottom row) during training for the three test cases without noise. Shaded regions denote the standard error over multiple random seeds.}
    \label{fig:cs_losses}
    \end{figure}
\subsection{Susceptible-Infected-Recovered}
We provide the $L_2$ loss associated with the predictions on the CDC influenza dataset in Fig.~\ref{fig: sir_loss}.

When considering the real-world influenza data, we only have access to a proxy for the infected population fraction, namely the influenza test positivity rate $\mu_t(I)$. Therefore, the calculated $L_2$ only considers the agents at the infected state. We cannot report the parameter $L_2$ error due to the unavailability of the true parameters $\gamma$.

\begin{minipage}{.45\textwidth}
\begin{figure}[H]
    \centering
\includegraphics[width=\textwidth]{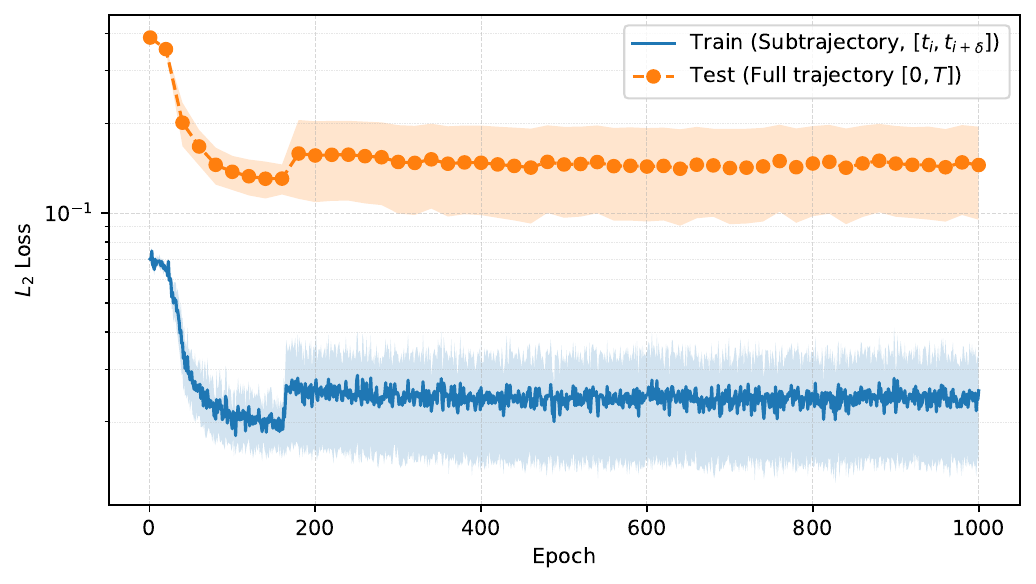}
    \caption{Convergence of the mean field $L_2$ loss between the predicted infected coordinate $\mu_t^\theta(I)$ and the observed influenza test positivity rate. Shaded regions denote the standard error over multiple random seeds.}
    \label{fig: sir_loss}
\end{figure}
\end{minipage}%
\hfill
\begin{minipage}{0.45\textwidth}
\begin{figure}[H]
    \centering
    \includegraphics[width=\textwidth]{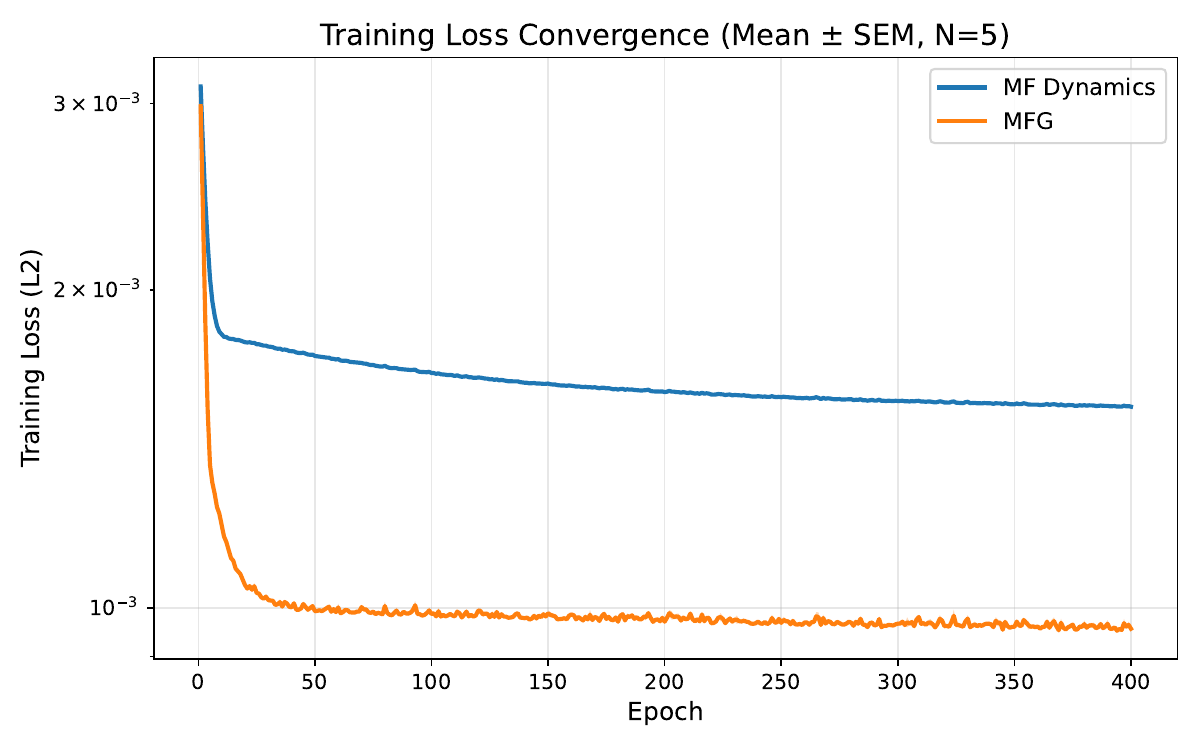}
    \caption{Convergence of the mean field dynamics model and the mean field game for the urban mobility example.\\[3mm]}
    \label{fig:loss_bike}
\end{figure}
\end{minipage}

\subsection{Urban Mobility.}
Lastly, Fig.~\ref{fig:loss_bike} presents the convergence of the mean field dynamics model and mean field game for the urban mobility example. The figure presents the mean and standard error over the mean as calculated over the five seeds.

\end{document}